\documentclass[11pt]{article}
\usepackage{amstext,amsmath,amssymb,amsfonts,bbm}

\usepackage[utf8]{inputenc}
\usepackage{epsfig}
\usepackage{dsfont}

\usepackage[backend=biber, doi=true, url=true, style = numeric, sorting = none]{biblatex}
\usepackage{amsthm}
\usepackage{subfigure}
\usepackage{color}
\usepackage{multirow}
\usepackage{psfrag}
\usepackage{graphicx}
\usepackage{extarrows}
\usepackage{authblk}
\usepackage[table]{xcolor}
\usepackage{microtype,booktabs,longtable,array,enumitem}
\usepackage{geometry}
\usepackage{parskip}
\usepackage{physics}
\usepackage{comment}
\usepackage{stmaryrd}
\usepackage{mathbbol}
\usepackage{empheq}
\usepackage{tikz-cd}
\usepackage{tikz}
\usetikzlibrary{arrows.meta,calc,positioning,decorations.markings}
\usepackage{caption}
\usepackage{needspace}
\definecolor{ribbonblue}{RGB}{35,101,139}
\definecolor{facegreen}{RGB}{47,117,83}
\definecolor{cutred}{RGB}{157,61,55}
\definecolor{jtSage}{HTML}{E5EFE9}
\definecolor{jtBlue}{HTML}{E5EEF4}
\definecolor{jtGrey}{HTML}{F1F1F1}
\definecolor{jtRose}{HTML}{EDC3B5}
\definecolor{jtGold}{HTML}{F0D8A8}
\definecolor{jtInkBlue}{HTML}{47718D}
\definecolor{jtInkGreen}{HTML}{4D7965}
\tikzset{corner/.style={circle,fill=black,inner sep=1.7pt},
 transport/.style={thick,postaction={decorate},decoration={markings,
 mark=at position .55 with {\arrow{Latex}}}},
 treeedge/.style={thick,ribbonblue,-{Latex[length=2mm]}},
 freeside/.style={thick,dashed,cutred,-{Latex[length=2mm]}}}
\tikzset{
  jt edge/.style={draw=black!90,line width=.72pt,line cap=round,line join=round},
  jt faint/.style={draw=black!25,line width=.45pt,dash pattern=on 2pt off 2pt},
  jt arrow/.style={postaction={decorate},decoration={markings,mark=at position .55 with {\arrow{Stealth[length=1.65mm,width=1.15mm]}}}},
  jt tree/.style={draw=jtInkBlue,line width=1.35pt,line cap=round,line join=round},
  jt label/.style={font=\small,inner sep=1.4pt},
  jt move/.style={draw=black!75,line width=.8pt,-{Stealth[length=2.2mm,width=1.5mm]}}
}

\usepackage{scalerel}
\usepackage{xurl}
\usepackage[colorlinks=true, urlcolor=blue, citecolor=blue,linktoc=page]{hyperref}

\usepackage{tablefootnote}

\newcommand{\ee}{\end{equation}}
\newcommand{\beq}{\begin{equation}}
\newcommand{\eeq}{\end{equation}}
\newcommand{\bes}{\begin{eqnarray}}
\newcommand{\ees}{\end{eqnarray}}
\newcommand{\bqa}{\begin{eqnarray}}
\newcommand{\eqa}{\end{eqnarray}}
\newcommand{\bea}{\begin{eqnarray}}
\newcommand{\eea}{\end{eqnarray}}

\newcommand{\so}{\mathfrak{so}}

\newcommand{\an}{\mathfrak{an}}

\newcommand{\cD}{{\cal D}}

\newcommand{\cF}{{\cal F}}

\newcommand{\cG}{{\cal G}}

\newcommand{\cR}{{\mathcal{R}}}

\newcommand{\g}{\mathfrak{g}}

\def\tell{\tilde{\ell}}

\def\g{\mathfrak{g}}

\def\th{{\tilde h} }

\def\fg{\mathfrak{g}}

\newcommand{\bbm}{\begin{bmatrix}}
\newcommand{\ebm}{\end{bmatrix}}

\usepackage{bbding}

\newlength\myheight  

\theoremstyle{definition}
\newtheorem{theorem}{Theorem}[section]
\AtEndEnvironment{theorem}{\null\hfill\ensuremath{\scaleto{\diamond}{\myheight}}}

\AtEndEnvironment{definition}{\null\hfill\ensuremath{\scaleto{\diamond}{\myheight}}}

\newtheorem{proposition}[theorem]{Proposition}
\AtEndEnvironment{proposition}{\null\hfill\ensuremath{\scaleto{\diamond}{\myheight}}}

\AtEndEnvironment{example}{\null\hfill\ensuremath{\scaleto{\diamond}{\myheight}}}

\newtheorem{lemma}[theorem]{Lemma}
\AtEndEnvironment{lemma}{\null\hfill\ensuremath{\scaleto{\diamond}{\myheight}}}

\AtEndEnvironment{remark}{\null\hfill\ensuremath{\scaleto{\diamond}{\myheight}}}

\AtEndEnvironment{assumption}{\null\hfill\ensuremath{\scaleto{\diamond}{\myheight}}}

\AtEndEnvironment{corollary}{\null\hfill\ensuremath{\scaleto{\diamond}{\myheight}}}

\title{Regge Dilaton Gravity from BF Theory}

\author[1]{Florian Girelli\footnote{florian.girelli@uwaterloo.ca}}
\author[1,2]{Christopher Pollack\footnote{cajpolla@uwaterloo.ca}}

\affil[1]{\textit{Department of Applied Mathematics, University of Waterloo, 200 University Avenue West, Waterloo, Ontario, Canada, N2L 3G1}}
\affil[2]{\textit{Perimeter Institute for Theoretical Physics, 31 Caroline St. N., Waterloo, Ontario, Canada, N2L 2Y5}}
\date{\today}

\begin{document}

\maketitle

\begin{abstract}
In this work we investigate the relation between gauge and geometric descriptions of two-dimensional dilaton gravity at the level of discrete path integrals. We develop a Regge discretization of Euclidean two-dimensional dilaton gravity using homogeneously curved simplices, as well as a ribbon refined discretization of two-dimensional BF theory adapted to homogeneously curved spacetimes. For Euclidean signature and vanishing or negative cosmological constant we show how the ribbon BF discretization, through appropriate gauge fixings, leads to the Regge dilaton gravity formulation. The resulting correspondence makes explicit how a geometric Regge path integral with  its associated measure emerges from its gauge theoretic formulation. 

\end{abstract}

\tableofcontents

\newpage

\section*{Introduction}\addcontentsline{toc}{section}{Introduction}

One of the central objectives of quantum gravity is to give a precise meaning to the formal gravitational partition function
$$
Z=\int [dg_{\mu\nu}]\,e^{iS_{\rm GR}[g]}.
$$
A natural strategy is to regularize this functional integral by replacing the smooth geometry with a discrete one. Regge calculus provides a natural discretization of general relativity in which the metric geometry is encoded in the edge lengths\footnote{In four dimensions one may alternatively formulate the discrete geometry using triangle areas as fundamental variables, although these variables are subject to additional consistency constraints. This leads to the different versions of area or area--angle Regge calculus \cite{barrett_note_1999,dittrich_areaangle_2008}. } of a simplicial complex and curvature is concentrated on codimension-two hinges, measured by deficit angles \cite{regge_general_1961,williams_regge_1992}.
Under suitable refinement and non-degeneracy assumptions, the curvature measures associated with piecewise-flat geometries converge to those of the smooth geometry they approximate \cite{cheeger_curvature_1984}. The usual formulation  treats each simplex as intrinsically flat, even in the presence of a non-vanishing cosmological constant $\Lambda$. In that case the homogeneous curvature generated by $\Lambda$ has to be reconstructed collectively from deficit angles and from the cosmological volume term, so that even a maximally symmetric solution is represented by a non-trivial piecewise-flat approximation. A more natural starting point is therefore to use simplices of constant sectional curvature $K$, with $K$ chosen according to the cosmological constant \cite{bahr_improved_2009,bahr_regge_2010}. The cosmological curvature is then already incorporated locally in each simplex, while the deficit angles measure the curvature in excess of this homogeneous background. This distinction is particularly important on coarse triangulations: curved simplices describe constant curvature solutions exactly rather than approximating them by an assembly of flat cells. In three-dimensional gravity this viewpoint emerges naturally from the construction of perfect discretizations, for which integrating out finer degrees of freedom leads precisely to a Regge action built from homogeneously curved simplices \cite{bahr_improved_2009}. It also improves the realization of the continuum gauge symmetries, which are otherwise generically broken at finite triangulation in a flat-simplex discretization \cite{ bahr_improved_2009, bahr_regge_2010, Dittrich:2008pw, bahr_broken_2009}. 

At the quantum level, specifying the discrete action is not sufficient: one must also determine the integration measure. The choice of measure in the gravitational path integral is a difficult problem because there is no unique canonical measure on the space of geometries. Formally, one would like to integrate over metrics modulo diffeomorphisms, and the separation between physical and gauge directions introduces non-trivial Jacobians and functional determinants \cite{dewitt_quantum_1967_II}. A common starting point is to endow the space of metrics with a field-space metric, such as the DeWitt supermetric, from which a functional measure can be induced; however, this construction itself involves choices and its implementation after gauge fixing is non-trivial \cite{dewitt_quantum_1967_I,menotti_diffeomorphism_1997}. The ambiguity persists after discretization. In Regge calculus, several measures on the edge lengths can be considered, and requiring a sensible relation with the continuum theory places non-trivial restrictions on them. In particular, Hamber and Williams showed that some proposed non-local simplicial measures fail to reproduce the continuum measure in the weak-field, low-momentum limit \cite{hamber_measure_1999}. Thus, while the discrete classical action and its symmetries can be characterized rather explicitly, the measure remains an additional ingredient in the definition of the quantum theory and cannot in general be regarded as a harmless normalization convention. Different choices of measure may lead to different quantum amplitudes and, potentially, to inequivalent quantum theories.

\medskip 

Baratin and Freidel developed a series of 
constructions showing how quantum topological models can emerge from Feynman amplitudes in three- and four-dimensional Euclidean spacetimes \cite{baratin_hidden_2007,baratin_hidden_2007-1}. The resulting models are Regge-like, in the sense that their amplitudes can be expressed in terms of geometric data such as lengths and dihedral angles. In three dimensions, 
their construction treats both flat and homogeneous spherical or hyperbolic spaces and recovers the corresponding Regge action, including the cosmological-volume contribution in the curved cases.
An important feature is that the Feynman amplitudes determine not only the discrete action, but also a specific measure on the geometric variables,  compatible with the topological character of the theory and with Pachner move invariance. The explicit relation of their construction with a potential BF theory discretization was an open question from their perspective. 

In four dimensions, the Baratin--Freidel topological model for \textit{zero} cosmological constant was subsequently related to Poincar\'e 2-group geometry and to the BFCG/Yetter description \cite{asante_quantum_2020}. The continuum relation between ordinary Poincar\'e BF theory and the BFCG formulation can be made explicit by adding an appropriate boundary term  before discretization \cite{girelli-2-groups}. At the discrete level, this semi-dual structure reflects the fact that not all components of the gauge data are discretized on the same complex: part of the connection data is naturally associated with the dual complex, while another component, identified with the analogue of the frame field, is discretized directly on the triangulation. The construction therefore naturally involves lattice 2-gauge theory techniques and  is consistent with the representation-theoretic description in terms  of the Poincar\'e 2-group \cite{Baez:2008hz, Baratin:2014era}.

These results raise two natural questions. Can one construct analogous Regge-like models in two-dimensional spacetime? Can such models be related to a BF formulation, including in the presence of a non-vanishing cosmological constant? In the present work, we address these questions by constructing a homogeneous Regge discretization of two-dimensional dilaton gravity and establishing its relation to a ribbon discretization of BF theory. In particular, the BF formulation determines the measure entering the Regge path integral, rather than requiring it to be postulated independently. We leave the derivation of this model from Feynman amplitudes to future work. 

Since the Einstein--Hilbert action is topological in two dimensions and does not provide non-trivial   dynamics, we instead consider dilaton gravity \cite{grumiller_dilaton_2002}. Its classical solutions are geometries of homogeneous curvature, and the theory admits a BF formulation whose gauge group depends on the value and sign of the cosmological constant. We construct a Regge discretization of the dilaton-gravity action using curved simplices, which are particularly natural in this setting because the homogeneous curvature is already incorporated in the elementary building blocks. We restrict here to Euclidean Regge gravity in order to keep the geometric discussion transparent, although the Lorentzian theory is an important extension and can be approached through several existing formulations \cite{tate_realizability_2012, ito_pinched_2026}. At this stage, the Regge action alone does not determine a unique quantum path integral, since the appropriate measure on the discrete geometric variables remains to be fixed. This ambiguity will instead be resolved by passing to the BF formulation and deriving the Regge description from its discretization.

We therefore proceed by discretizing the relevant BF theories using a ribbon structure rather than the standard lattice description. The ribbon formulation retains information about both the triangulation and its dual complex while still allowing the topological invariance of the theory to be established. After an appropriate gauge fixing, different components of the gauge data are naturally associated with these two dual structures. In the case of non-vanishing cosmological constant, the Iwasawa decomposition of the BF gauge group plays a central role in identifying the relevant holonomies. The short ribbon side holonomies are naturally associated with the dual complex and encode the deficit angles around the vertices. In order to express these deficit angles in terms of the geometric data of the triangulation, however, one must also introduce long ribbon side holonomies along the triangulation (see Fig. \ref{fig:three-BF-discretizations} and \ref{fig:ribboncells}). This provides a discrete explanation for the semi-dual structure already visible in the continuum BF formulation: the variables required to reconstruct the Regge geometry cannot all be placed on the same combinatorial support.

By solving an appropriate subset of the discrete BF constraints, we recover the Regge discretization of two-dimensional dilaton gravity. The BF formulation also provides an important advantage at the quantum level. The path integral of the gauge theory comes with a natural measure determined by the Haar measures of the relevant gauge groups, together with the usual treatment of gauge redundancy. Pushing this measure through the reduction from the BF variables to the Regge variables then determines the corresponding measure for the dilaton-gravity path integral. In this way, the measure entering the Regge formulation is not postulated independently, but is inherited from the underlying BF theory in a manner compatible with its topological invariance.

The paper is organized as follows. In section \ref{sec:homogeneous-spacetime-and-BF}, we fix notations and recall the gauge formulation of (linear) dilaton gravity. We emphasize in particular how a homogeneous geometry can be characterized in terms of a gauge theory. In section \ref{sec:regge}, we discretize the (linear) dilaton gravity action using Regge calculus with homogeneously curved simplices and discuss the ambiguity in the choice of measure for the corresponding path integral.  
In section \ref{sec:split-discretization-of-2D-BF}, we construct the ribbon BF discretization and its Regge-adapted partial gauge fixing, resulting in a path integral we call \textit{split} BF. Section \ref{sec:recovering-homo-regge} derives a suitable Regge length measure from this split BF theory and proves the Regge path integral invariance under Pachner moves. In appendix \ref{app:discrete-dilaton-eom}, we discuss how the  dilaton equations of motion from the Regge action can be viewed as a discretization of the continuum dilaton equations of motion. Appendix \ref{app:measure-considerations} discusses the Haar measure properties when using the Iwasawa decomposition of a given group.   

Table \ref{table:signature-and-curvature-cases} summarizes for which signatures and signs of $\Lambda$ these Regge, ribbon BF, and split BF theories have been defined. The overlapping case of Euclidean signature and $\Lambda \leq 0$ is thus where the split BF to Regge correspondence is established. 

\begin{table}[h]
    \centering
    \begingroup
    \setlength{\tabcolsep}{14pt}
    \renewcommand{\arraystretch}{1.2}
    \begin{tabular}{@{}r|ccc|ccc@{}}
        \toprule
        &  & Euclidean & &  & Lorentzian & \\
         & $\Lambda>0$ & $\Lambda=0$ & $\Lambda<0$ & $\Lambda>0$ & $\Lambda=0$ & $\Lambda<0\:$ \\
        \midrule
        Regge
            & \checkmark & \checkmark & \checkmark & - & - & - \\
        Ribbon BF
            & \checkmark & \checkmark & \checkmark  & \checkmark & \checkmark & \checkmark\\
        Split BF
            & - & \checkmark & \checkmark & \checkmark & \checkmark & \checkmark\\
        \bottomrule
    \end{tabular}
    \endgroup
    \caption{We define three models in this work: Homogeneous Regge discretized linear dilaton gravity, ribbon discretized BF theory, and its partially gauged fixed split BF theory. This table illustrates for which signatures and signs of $\Lambda$ these models are constructed. The common case of Euclidean $\Lambda = 0$ and $\Lambda <0$ amongst all three is where our BF to Regge correspondence is constructed.}\label{table:signature-and-curvature-cases}
\end{table}

\section{Two-dimensional gravity and two-dimensional BF theory}\label{sec:homogeneous-spacetime-and-BF}

We recall that the geometry of a constant curvature spacetime can be locally modeled on a homogeneous space constructed by quotienting Lie groups  \cite{goldman_locally_2010,lee_introduction_2018}. We distinguish the homogeneous model space $X$ from a spacetime $M$ carrying this local geometry. This section fixes our conventions for these notions and reviews how BF theory can similarly encode the local geometry of a constant curvature spacetime. We close by reviewing the continuum relation between BF theory and linear dilaton gravity (LDG).

\subsection{Two-dimensional homogeneous spacetimes and their isometries}\label{subsec:homo-spacetime-isometries}

Consider a two-dimensional constant curvature spacetime $(M, g_{\mu\nu})$, with $\mu,\nu = 0,1$ denoting spacetime indices. We denote the homogeneous curvature of $M$ by
\begin{equation}
    \Lambda = \kappa|\Lambda|:=\frac{\kappa}{l^{2}},\label{eq:cosmo-constant}
\end{equation}
where $l = 1/\sqrt{\abs{\Lambda}}$ is the cosmological radius and $\kappa \in \{1,0,-1\}$. The associated first order metric variables are a zweibein $e^a$ and a spin connection $\omega$, where $a,b = 0,1$ denote internal Lorentz frame indices. These indices are raised and lowered by
\begin{equation}
\eta_{ab}=\mathrm{diag}(\sigma,1),
\quad
\sigma=\pm1,
\end{equation}
where $\sigma=1$ is Euclidean and $\sigma=-1$ is Lorentzian, and
\begin{equation}
    g_{\mu\nu} = \eta_{ab}e^{a}_{\mu}e^{b}_{\nu}.
\end{equation}
We fix the Levi-Civita symbol by
\begin{equation}
    \epsilon_{01} = 1, \quad \epsilon^{01} = \sigma \implies \epsilon_{ab}\epsilon^{bc} = -\sigma \delta_{a}^{c}.
\end{equation}
Table \ref{table:2d-spacetimes} lists the corresponding local model geometries $X$.
\begin{table}[h]
    \centering
    \begingroup
    \setlength{\tabcolsep}{14pt}
    \renewcommand{\arraystretch}{1.2}
    \begin{tabular}{@{}lccc@{}}
        \toprule
        Metric signature/curvature sign & $\kappa=1$ & $\kappa=0$ & $\kappa=-1$ \\
        \midrule
        Euclidean ($\sigma=1$)
            & $S^2$ & $\mathbb{R}^2$ & $H^2$ \\
        Lorentzian ($\sigma=-1$)
            & $\mathrm{dS}_2$ & $\mathbb{R}^{1,1}$ & $\mathrm{AdS}_2$ \\
        \bottomrule
    \end{tabular}
    \endgroup
    \caption{Two-dimensional constant curvature model spaces $X$, classified by signature $\sigma$ and curvature sign $\kappa$. A general spacetime $M$ is only locally isometric to the corresponding model.}\label{table:2d-spacetimes}
\end{table}
Let $\fg_{\sigma,\kappa}$ denote the three-dimensional isometry algebra of $X$, with associated model isometry group $G_{\sigma,\kappa}$. In the Cartan decomposition, the generators $J$ and $P_a$ satisfy
\begin{equation}
[J,P_a]=\epsilon_a{}^bP_b,
\quad 
[P_a,P_b]=\sigma\frac{\kappa}{l^2}\epsilon_{ab}J = \sigma\Lambda\epsilon_{ab}J.
\label{eq:isometry_algebra}
\end{equation}
In Euclidean signature $\sigma = 1$, the curvature signs $\kappa \in\{1,0,-1\}$ correspond respectively to the spherical, flat, and hyperbolic cases:
\begin{equation}
\fg_{1,1} \cong \mathfrak{so}(3) ,\quad \fg_{1,0} \cong \mathfrak{iso}(2),\quad \fg_{1,-1} \cong \mathfrak{so}(2,1).\label{eq:euclidean-cartan-lie-alg}
\end{equation}
For Lorentzian signature $\sigma = -1$, $\kappa\in\{1,0,-1\}$ correspond respectively to positive curvature, flat, and negative curvature kinematics, i.e., de Sitter, Poincar\'e, and anti-de Sitter:
\begin{equation}
    \fg_{-1,1} \cong \mathfrak{so}(1,2),\quad \fg_{-1,0} \cong \mathfrak{iso}(1,1),\quad \fg_{-1,-1} \cong \mathfrak{so}(2,1).\label{eq:lorentzian-cartan-lie-alg}
\end{equation}
We will also work with the real three-dimensional \textit{Abelian} Lie algebra $\fg^{*}$, taken to mean the dual of $\fg_{\sigma,\kappa}$ as a vector space. Denoting the generators of $\fg^{*}$ by $\tilde{J}$ and $\tilde{P}^{a}$, there exists a canonical invariant pairing $\langle \bullet,\bullet\rangle$ between $\fg_{\sigma,\kappa}$ and $\fg^{*}$ satisfying \cite{snobl_classification_2002},
\begin{equation}
    \langle J,\tilde{J}\rangle = 1,\quad \langle P_{a},\tilde{P}^{b}\rangle = \delta_{a}^{b},\quad \langle J,\tilde{P}^{a}\rangle = \langle P_{a},\tilde{J}\rangle = 0,\label{eq:g-g*-Cartan-pairing}
\end{equation}
as well as a (dressing) action $\lhd$ of $\fg_{\sigma,\kappa}$ on $\fg^{*}$ given by,\footnote{Note there is no ``back-action'' of $\fg^{*}$ on $\fg_{\sigma,\kappa}$. As well, it may be useful to recall that the dressing action has the anti-symmetric property $\alpha\lhd \beta = -\beta\rhd \alpha$ for all $\beta\in \fg_{\sigma,\kappa}$ and $\alpha\in \fg^{*}$.}
\begin{equation}
\tilde{J}\lhd J = 0,\quad \tilde{J}\lhd P_{a} = \sigma\frac{\kappa}{l^{2}}\epsilon_{ab}\tilde{P}^{b},\quad \tilde{P}^{a}\lhd J = \epsilon_{b}{}^{a}\tilde{P}^{b},\quad \tilde{P}^{a}\lhd P_{b} = -\epsilon_{b}{}^{a}\tilde{J}.\label{eq:Cartan-Dressing-Lie-Alg}
\end{equation}
A constant curvature spacetime $M$ is locally modeled on the homogeneous space $X$  \cite{goldman_locally_2010}.
More precisely, $M$ admits local charts with values in $X$ whose transition functions are restrictions of elements of $G_{\sigma,\kappa}$. Choosing  a base point $x_0\in X$, we write
\begin{equation}
     X:=G_{\sigma,\kappa}/H_{\sigma},\quad H_{\sigma}:=\{g\in G_{\sigma,\kappa}\:|\:g\cdot x_{0} = x_{0}\}\subset G_{\sigma,\kappa},\label{eq:X-as-quotient}
\end{equation}
for a base point $x_0\in X$. The model metric defines an invariant measure $d\mu_X$ and invariant distance function $d_{X}(\bullet,\bullet)$.\footnote{In a Lorentzian signature this induced distance function is not necessarily positive definite. We restrict to the Euclidean signature in \S\ref{sec:recovering-homo-regge} thus we do not address such subtleties here.} The stabilizers are $H_1=SO(2)$ and $H_{-1}=SO(1,1)$, giving the following quotient realizations of the model spaces in Table \ref{table:2d-spacetimes}:
\begin{subequations}
\begin{align}
S^{2} \cong SO(3)/SO(2),\quad \mathbb{R}^{2}&\cong ISO(2)/SO(2), \quad \mathbb{H}^{2}\cong SO(2,1)/SO(2),\\
\mathrm{dS}_{2} \cong SO(1,2)/SO(1,1),\quad \mathbb{R}^{1,1}&\cong ISO(1,1)/SO(1,1),\quad \mathrm{AdS}_{2}\cong SO(2,1)/SO(1,1). 
\end{align}
\end{subequations}

\subsection{Homogeneous geometry from BF theory and dilaton gravity}\label{subsec:continuum-BF-to-regge}

Two-dimensional BF theory with gauge group $G_{\sigma,\kappa}$ encodes constant curvature geometry in the Cartan formulation \cite{wise_symmetric_2009}. For an invertible zweibein, the constant curvature geometry is encoded in the relations
\begin{equation}
g_{\mu\nu} = \eta_{ab}e^{a}_{\mu}e^{b}_{\nu},\quad d\omega+\sigma\frac{\kappa}{2l^2}\epsilon_{ab}e^a\wedge e^b = 0,\quad de^{a}+\omega\wedge e^{b}\epsilon_{b}{}^{a} =0 .   \label{eq:homo-M-e-omega-relations}
\end{equation}
A topological field theory admitting these relations as part of its dynamics can be naturally constructed. We package the $(\omega, e^{a})$ data together into the Cartan connection,
\begin{equation}
    A:= \omega J + e^{a}P_{a}\in \Omega^{1}(M,\mathfrak{g}_{\sigma,\kappa}).\label{eq:A-cartan}
\end{equation}
The curvature 2-form $\cF(A) = dA+\frac{1}{2}\comm{A}{A}$ associated to $A$ decomposes as
\begin{equation}
\mathcal{F}(A)=FJ+T^aP_a ,\label{eq:A-curvature-decomp}
\end{equation}
with
\begin{subequations}{\label{eq:2d_torsion_curvature}}
\begin{align}
F &= d\omega+\sigma\frac{\kappa}{2l^2}\epsilon_{ab}e^a\wedge e^b\\
T^a &= de^a+\omega\wedge e^b\epsilon_{b}{}^{a}.\label{eq:cartan-torsion}
\end{align}
\end{subequations}
The Cartan curvature $F$ and torsion $T^a$ vanish precisely when \eqref{eq:homo-M-e-omega-relations} holds,
\begin{equation}
\cF(A) = 0 \iff F = 0 \text{ and } T^{a}=0.\label{eq:homogeneous_constraints}
\end{equation}
To impose these conditions dynamically, introduce the Lagrange multiplier
\begin{equation}
B:=\Phi \tilde{J}+X_a\tilde{P}^a \in \Omega^{0}(M,\fg^{*}),
\label{eq:B-cartan}
\end{equation}
such that along with the invariant pairing \eqref{eq:g-g*-Cartan-pairing} the BF action \cite{horowitz_exactly_1989, baez_introduction_1999}
\begin{equation}
S_{\mathrm{BF}}[A,B] =-\frac{1}{8\pi G_{N}}\int_M\langle B,\cF(A)\rangle=-\frac{1}{8\pi G_{N}}\int_M\Big(X_aT^a+\Phi F\Big)
\end{equation}
is precisely the topological field theory admitting the desired dynamics. Indeed, varying $S_{\mathrm{BF}}[A,B]$ with respect to $X_a$ and $\Phi$ imposes \eqref{eq:homogeneous_constraints},
\begin{equation}
\delta_{X_{a}}S = 0 \implies T^a\approx 0, \quad \delta_{\Phi}S = 0 \implies F\approx 0,
\end{equation}
where we henceforth adopt the notation ``$\approx$'' to mean going on-shell.

The connection to the discussions of \S\ref{subsec:homo-spacetime-isometries} comes from the fact that $S_{\mathrm{BF}}[A,B]$ has $G_{\sigma,\kappa}$ as its gauge group. Indeed it satisfies the infinitesimal gauge symmetry
\begin{subequations}
\begin{align}
\delta_\lambda A=d_A\lambda, &\quad \delta_\lambda B=B\lhd\lambda,\quad \lambda=\theta J+\rho^aP_a\in \fg_{\sigma,\kappa} \label{eq:BF_gauge_transformation}\\
\implies \delta_\lambda S_{\mathrm{BF}} &= \int_M \langle B\lhd \lambda,\cF(A)\rangle + \int_M \langle B,[\cF(A),\lambda]\rangle =0,\label{eq:BF_gauge_invariance}
\end{align}
\end{subequations}
which in components, reads
\begin{subequations}{\label{eq:BF-component-gauge-transformations}}
\begin{align}
\delta_\lambda \omega &= \langle d\theta + \comm{e}{\rho}, \tilde{J}\rangle = d\theta + \sigma\frac{\kappa}{l^2}\epsilon_{ab}e^a\rho^b,\label{eq:gauge_omega}\\
\delta_\lambda e^a &= \langle d\rho+\comm{\omega}{\rho}+\comm{e}{\theta},\tilde{P}^{a}\rangle = d\rho^a+\omega\rho^b\epsilon_{b}{}^{a}-\theta e^{b}\epsilon_{b}{}^{a},\label{eq:gauge_frame}\\
\delta_\lambda \Phi &= \langle X\lhd \rho,J\rangle = -X_{a}\rho^{b}\epsilon_{b}{}^{a},\label{eq:gauge_phi}\\
\delta_\lambda X_a &= \langle X\lhd\theta+ \Phi\lhd \rho ,P_{a}\rangle = \theta \epsilon_{a}{}^{b}X_{b}-\sigma\frac{\kappa}{l^{2}}\Phi\rho^{b}\epsilon_{ab} .\label{eq:gauge_X}
\end{align}
\end{subequations}
Thus BF theory encodes the local homogeneous geometry of $M$ in its non-degenerate zweibein sector, with the model isometry group as its gauge group.

The two-dimensional topological BF theory we have built has a natural relation with two-dimensional (linear) dilaton gravity (LDG). 
The standard (Euclidean) action for LDG is given by (omitting boundary terms),
\begin{equation}
     S_{\mathrm{LDG}}[g,\Phi;\Lambda_{\mathrm{LDG}}] = \frac{1}{16\pi G_{N}}\int_{M}d^{2}x\sqrt{g}\Phi(R[g]-2\Lambda_{\mathrm{LDG}}),\label{eq:linear-dilaton-gravity-action}
\end{equation}
where $\Phi$ is a scalar field referred to as the dilaton field. For $\Lambda_{\mathrm{LDG}}=-1$ this is Euclidean Jackiw--Teitelboim (JT) gravity \cite{mertens_solvable_2023}. The equations of motion are given by,
\begin{equation}
    R-2\Lambda_{\mathrm{LDG}}\approx 0,\quad \nabla_{\mu}\nabla_{\nu}\Phi - g_{\mu\nu}\nabla^{2}\Phi-\Lambda_{\mathrm{LDG}}g_{\mu\nu}\Phi \approx 0.\label{eq:LDG-continuum-EOM}
\end{equation}
For an invertible zweibein, imposing the torsion constraint \eqref{eq:cartan-torsion} gives \cite{mertens_solvable_2023}
\begin{equation}
    d\omega = -\frac{1}{2}d^{2}x\sqrt{g}R[g],\quad \frac{1}{2}\epsilon_{ab}e^{a}\wedge e^{b} = d^{2}x\sqrt{g},
\end{equation}
thus,
\begin{equation}
    S_{\mathrm{BF}}[A,B] \overset{T^{a}\approx 0}{\longmapsto} \frac{1}{16\pi G_{N}}\int_{M}d^{2}x\sqrt{g}\Phi(R[g]-2\Lambda) = S_{\mathrm{LDG}}[g,\Phi;\Lambda_{\mathrm{LDG}}=\Lambda]. \label{eq:BF-LDG-relation}
\end{equation}
The characterization of homogeneous geometry by BF theory and this continuum relation between BF and LDG motivate the approach we take in the following sections.

\section{Homogeneous Regge discretization of dilaton gravity and its path integral}\label{sec:regge}

Here we give the \textit{homogeneous} Regge discretization of two-dimensional gravity. Unlike the regular equilateral hyperbolic lattice used in \cite{erdmenger_discrete_2025}, our construction allows variable edge lengths and general triangulations. 

First we will discretize the LDG action using Regge techniques with homogeneously curved simplices. Indeed, it has been argued in higher dimensions that using homogeneous simplices as opposed to the historically flat Regge simplices \cite{regge_general_1961} is a more adequate tool for discretization in the presence of homogeneous curvature \cite{bahr_regge_2010, dittrich_covariant_2010, bonzom_diracs_2013,barrett_first_1994}. Moreover, homogeneously curved simplices will be the natural structure recovered in the BF formulation. Then, we give the associated Regge discretized path integral, taken to be induced by the forthcoming BF to Regge correspondence of \S\ref{sec:recovering-homo-regge}. 

\subsection{Homogeneous Regge discretization of the dilaton gravity action}\label{subsec:Regge-discretization-of-LDG}

Discretization begins by choosing a triangulation $\Delta$ of $M$, comprised of vertices $v\in V(\Delta)$, edges $e\in E(\Delta)$, and triangles $\tau\in T(\Delta)$. Poincar\'{e} duality identifies its corresponding dual complex $\Delta^*$, comprised of nodes $n\in V(\Delta^*)$, links $\ell\in E(\Delta^*)$, and faces $f\in F(\Delta^*)$, dual respectively to triangles, edges, and vertices. We write $|V(\Delta)|$, $|E(\Delta)|$, and $|T(\Delta)|$ for the size of these sets, and similarly for $\Delta^*$. The idea behind a two-dimensional Regge discretization using homogeneously curved simplices is that the total continuum scalar curvature is taken to be supported on two regions in $\Delta$: on each triangle $\tau$, and around each vertex $v$ (see Figure \ref{fig:spherical}).

\begin{figure}[ht]
    \centering
    \includegraphics[width=0.8\linewidth]{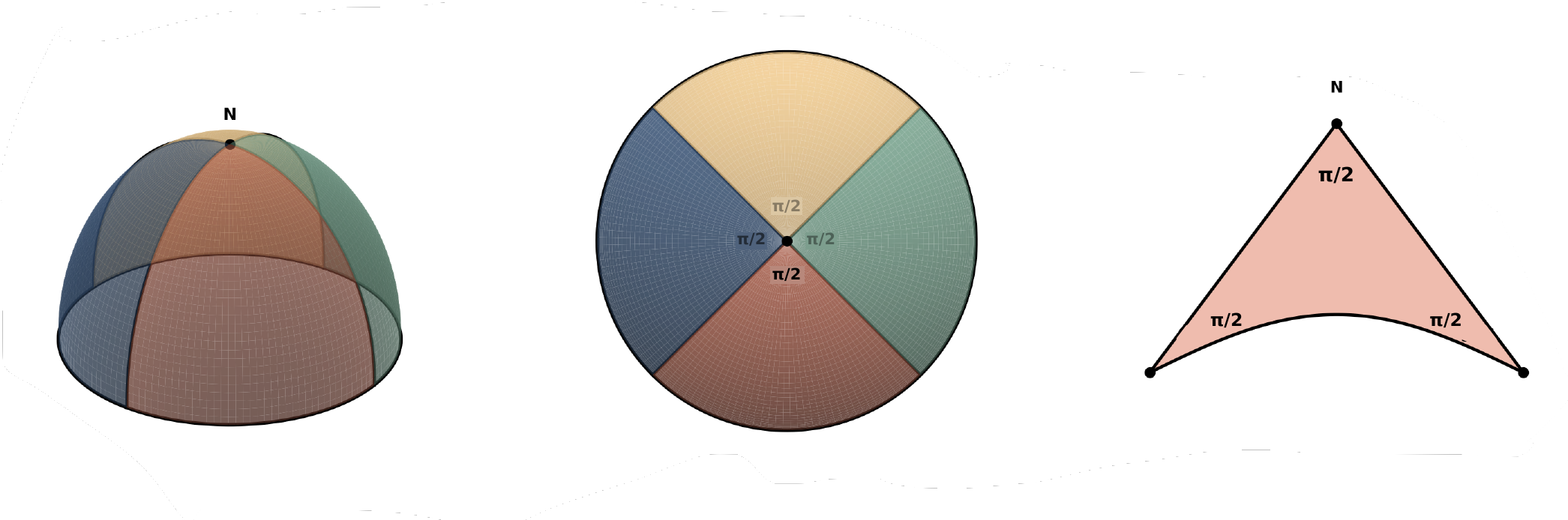}
    \caption{On the left side, we have a triangulation of the northern hemisphere using spherical triangles built from geodesics. In the center, we illustrate how in this case the deficit angle is zero modulo $2\pi$. In the presence of additional curvature concentrated at the vertex, the sum of the angles would not be $2\pi$. On the right side, the positive curvature of the spherical triangle is manifest in that the sum of the angles is larger than $\pi$.}
    \label{fig:spherical}
\end{figure}

\subsubsection{Geometry of the triangulation}\label{subsubsec:regge-triangulation}

Each $\tau$ is taken to have homogeneous sectional curvature $K$ by definition.\footnote{On a Riemannian manifold the sectional curvature $K$ is related to the Riemann tensor via,
\begin{equation}
K(u,v) = \frac{\langle R(u,v)v,u\rangle}{\langle u, u\rangle\langle v, v\rangle - \langle u, v\rangle^{2}},\quad \langle u, v\rangle:= g_{\mu\nu}u^{\mu}v^{\nu}    
\end{equation}
where $u$ and $v$ are linearly independent tangent vectors at the same point on the Riemannian manifold. In two dimensions the sectional curvature is exactly equal to the Gaussian curvature, or equivalently one-half the scalar curvature.} We choose
\begin{equation}
    K = \Lambda,\label{eq:K-Lambda-relation}
\end{equation}
so that $R=2K$ solves the LDG curvature equation inside each simplex. The remaining curvature is concentrated at the vertices. We describe the simplex geometry using the model space $X$ in \eqref{eq:X-as-quotient}, via a quadric $X_\Lambda$ for $\Lambda\ne0$, and an affine plane for $\Lambda=0$. 

It will be convenient to  introduce the curvature dependent functions
\begin{equation}{\label{eq:regge-curvature-functions}}
    C_\Lambda(l)
    :=
    \begin{cases}
        \cos(\sqrt{\Lambda}\,l),
        & \Lambda>0,
        \\[1mm]
        \cosh(\sqrt{-\Lambda}\,l),
        & \Lambda<0,\\[1mm]
        1 & \Lambda=0
    \end{cases}, \quad 
    S_\Lambda(l)
    :=
    \begin{cases}
        \displaystyle
        \frac{\sin(\sqrt{\Lambda}\,l)}{\sqrt{\Lambda}},
        & \Lambda>0,
        \\[3mm]
        \displaystyle
        \frac{\sinh(\sqrt{-\Lambda}\,l)}{\sqrt{-\Lambda}},
        & \Lambda<0
        \\[1mm]
        l & \Lambda=0.
    \end{cases}
\end{equation}
They satisfy
\begin{equation}
    C_\Lambda(l)^2+\Lambda S_\Lambda(l)^2=1,\quad C_{\Lambda}(l):= \dv{l}S_{\Lambda}(l),\quad \dv{l}C_{\Lambda}(l) = -\Lambda S_{\Lambda}(l), \quad l>0.
    \label{eq:CS-identity}
\end{equation}

\paragraph{Constant curvature: $\Lambda\neq0$.}

The homogeneous space can be realized as the following two-dimensional quadric,
\begin{equation}
    X_{\Lambda}:= \Big\{\mathcal{X}\in \mathbb{R}^{3}\;\Big|\;\mathcal{X}\cdot \mathcal{X} = \frac{1}{\Lambda}\Big\}.\label{eq:quadric-homogeneous-space}
\end{equation}
where the scalar product is Euclidean  or Lorentzian. Therefore a vertex $v\in V(\Delta)$ can be specified by a vector $\mathcal{X}_v\in X_{\Lambda}$. We will use the same notation to characterize the vector in the ambient space $\mathbb{R}^3$ and the vertex in $X_\Lambda$. The spacetime metric is discretized by assigning a geodesic distance $l_{ij}$ to every edge formed by vertices $\mathcal{X}_i$ and $\mathcal{X}_j$, and is encoded in the scalar product of their associated ambient vectors through
\begin{equation}
    \Lambda \mathcal{X}_i\cdot \mathcal{X}_j
    =
    C_\Lambda(l_{ij}).
    \label{eq:ambient-distance}
\end{equation} 
Each simplex is a geodesic triangle, so its edge lengths satisfy the appropriate triangle inequalities. Its vertices are represented by three points,
\begin{equation}
    \mathcal{X}_1,\mathcal{X}_2,\mathcal{X}_3\in X_\Lambda,
    \quad
    \mathcal{X}_i\cdot \mathcal{X}_i=\frac{1}{\Lambda}.
\end{equation}
The geodesic starting from $\mathcal{X}_i$ in the direction of a unit tangent vector $t_{ij}\in T_{\mathcal{X}_i}X_\Lambda$ can be parametrized as\footnote{These geodesics would be space-like in the Lorentzian case.}
\begin{equation}
    \mathcal{X}_{ij}(s)
    =
    C_\Lambda(s)\mathcal{X}_i
    +
    S_\Lambda(s)t_{ij},
    \quad
    \mathcal{X}_i\cdot t_{ij}=0,
    \quad
    t_{ij}\cdot t_{ij}=1.
    \label{eq:ambient-geodesic}
\end{equation}
In particular, choosing $s=l_{ij}$ gives the endpoint,
\begin{equation}
    \mathcal{X}_j
    =
    C_\Lambda(l_{ij})\mathcal{X}_i
    +
    S_\Lambda(l_{ij})t_{ij}.
    \label{eq:ambient-edge-endpoint}
\end{equation}
The tangent vector associated with the oriented edge $(ij)$ is therefore
\begin{equation}
    t_{ij}
    =
    \frac{
        \mathcal{X}_j-C_\Lambda(l_{ij})\mathcal{X}_i
    }{
        S_\Lambda(l_{ij})
    }.
    \label{eq:tangent-from-vertices}
\end{equation}
Thus the edge of the triangle is completely characterized by its initial vertex $\mathcal{X}_i$, its length $l_{ij}$, and its initial tangent vector $t_{ij}$.

Consider now the two edges $(ij)$ and $(ik)$ meeting at the vertex $\mathcal{X}_i$. Their angle $\beta_i$ is simply the angle between the corresponding tangent vectors,
\begin{equation}
    \cos\beta_i
    =
    t_{ij}\cdot t_{ik}.
    \label{eq:angle-tangent}
\end{equation}
Using \eqref{eq:ambient-distance} and \eqref{eq:tangent-from-vertices}, one obtains
\begin{equation}
    \cos\beta_i
    =
    \frac{
        \Lambda \mathcal{X}_j\cdot \mathcal{X}_k
        -
        C_\Lambda(l_{ij})C_\Lambda(l_{ik})
    }{
        \Lambda
        S_\Lambda(l_{ij})S_\Lambda(l_{ik})
    }=\frac{
       C_{\Lambda}(l_{jk})
        -
        C_\Lambda(l_{ij})C_\Lambda(l_{ik})
    }{
        \Lambda
        S_\Lambda(l_{ij})S_\Lambda(l_{ik})
    },
    \label{eq:angle-ambient}
\end{equation}
thus
\begin{equation}
    C_\Lambda(l_{jk})
    =
    C_\Lambda(l_{ij})C_\Lambda(l_{ik})
    +
    \Lambda
    S_\Lambda(l_{ij})S_\Lambda(l_{ik})
    \cos\beta_i,
    \label{eq:curved-cosine-law}
\end{equation}
which is the spherical or hyperbolic cosine law depending on the sign of $\Lambda$.

The geometry of a triangle $\tau = (123)$ can equivalently be encoded in its Gram matrix. It is natural to define the Gram matrix elements
\begin{equation}
    (G_{\tau,\Lambda})_{ij}
    :=
    \Lambda \mathcal{X}_i\cdot \mathcal{X}_j,
\end{equation}
such that from \eqref{eq:ambient-distance} the full Gram matrix is of the form
\begin{equation}
    G_{\tau,\Lambda}
    =
    \begin{pmatrix}
        1
        &
        C_\Lambda(l_{12})
        &
        C_\Lambda(l_{13})
        \\
        C_\Lambda(l_{12})
        &
        1
        &
        C_\Lambda(l_{23})
        \\
        C_\Lambda(l_{13})
        &
        C_\Lambda(l_{23})
        &
        1
    \end{pmatrix}.
    \label{eq:vertex-gram-matrix}
\end{equation}
The Gram matrix therefore contains all the information about the edge lengths and, through \eqref{eq:angle-ambient}, the internal angles of the triangle as well.

The area $\mathcal{V}_{\tau,\Lambda}$ of $\tau$ follows from Gauss--Bonnet,
\begin{equation}
    {
    \mathcal{V}_{\tau,\Lambda}
    =
    \frac{
        \beta_1+\beta_2+\beta_3-\pi
    }{\Lambda}.
    }
    \label{eq:curved-triangle-area}
\end{equation}
For $\Lambda>0$ this is the usual spherical excess, while for $\Lambda<0$ both the numerator and $\Lambda$ are negative.

\paragraph{Flat case: $\Lambda = 0$.}

After choosing an origin in the Euclidean plane $\mathbb{R}^2$, represent the vertices by
\begin{equation}
    x_1,x_2,x_3\in\mathbb R^2.
\end{equation}
The oriented edge from $x_i$ to $x_j$ is the vector
\begin{equation}
    E_{ij}:=x_j-x_i,
\end{equation}
with length
\begin{equation}
    l_{ij}^2=E_{ij}\cdot E_{ij}.
\end{equation}
Introducing the unit tangent vector
\begin{equation}
    t_{ij}=\frac{E_{ij}}{l_{ij}},
\end{equation}
the geodesic joining the two vertices takes the form
\begin{equation}
    x_{ij}(s)
    =
    x_i+s\,t_{ij},
    \quad
    0\leq s\leq l_{ij},
    \label{eq:flat-geodesic}
\end{equation}
in particular,
\begin{equation}
    x_j=x_i+l_{ij}t_{ij}.
\end{equation}
At the vertex $x_i$, the angle between the edges $(ij)$ and $(ik)$ is therefore
\begin{equation}
    \cos\beta_i
    =
    t_{ij}\cdot t_{ik}
    =
    \frac{
        E_{ij}\cdot E_{ik}
    }{
        l_{ij}l_{ik}
    }.
\end{equation}
Since
\begin{equation}
    E_{jk}=E_{ik}-E_{ij},
\end{equation}
one immediately obtains the standard Euclidean cosine law
\begin{equation}
    l_{jk}^2
    =
    l_{ij}^2+l_{ik}^2
    -
    2l_{ij}l_{ik}\cos\beta_i.
    \label{eq:flat-cosine-law}
\end{equation}
In contrast with the curved case, the scalar products $x_i\cdot x_j$ depend on the arbitrary choice of origin and therefore do not directly encode the geometry of the triangle. The natural Gram matrix is instead the Gram matrix of two edge vectors. For instance, choosing the vertex $x_1$,
\begin{equation}
    G_{\tau,0}
    =
    \begin{pmatrix}
        E_{12}\cdot E_{12}
        &
        E_{12}\cdot E_{13}
        \\
        E_{13}\cdot E_{12}
        &
        E_{13}\cdot E_{13}
    \end{pmatrix}
    =
    \begin{pmatrix}
        l_{12}^2
        &
        l_{12}l_{13}\cos\beta_1
        \\
        l_{12}l_{13}\cos\beta_1
        &
        l_{13}^2
    \end{pmatrix}.
    \label{eq:flat-edge-gram}
\end{equation}
Its determinant gives the area,
\begin{equation}
    {
    \mathcal{V}_{\tau,0}
    =
    \frac12\sqrt{\det G_{\tau,0}}
    =
    \frac12
    l_{12}l_{13}\sin\beta_1.
    }
    \label{eq:flat-triangle-area}
\end{equation}
The flat limit follows from
\begin{equation}
    C_\Lambda(l)
    =
    1-\frac{\Lambda l^2}{2}
    +O(\Lambda^2),
    \quad
    S_\Lambda(l)
    =
    l-\frac{\Lambda l^3}{6}
    +O(\Lambda^2).
\end{equation}
Expanding the curved cosine law \eqref{eq:curved-cosine-law} to first order in $\Lambda$ then gives
\begin{equation}
    l_{jk}^2
    =
    l_{ij}^2+l_{ik}^2
    -
    2l_{ij}l_{ik}\cos\beta_i,
\end{equation}
recovering the flat result.

\paragraph{Deficit angle around a vertex.}

Consider now a triangulation obtained by gluing such geodesic triangles along edges, and let $v$ be an internal vertex. Denote by $\tau\supset v$ the triangles incident at $v$, and by $\beta_v^\tau\in(0,\pi)$ the ordinary positive internal angle of a consistently oriented, non-folded triangle $\tau$ at this vertex. In this sector no additional orientation sign is required, and the  deficit angle associated with $v$ is defined by
\begin{equation}
    {
    \delta_{v}(l;\Lambda)
    =
    2\pi
    -
    \sum_{\tau\supset v}\beta_{v}^{\tau}(l;\Lambda) .
    }
    \label{eq:deficit-angle-vertex}
\end{equation}
The deficit measures the distributional Gaussian curvature at $v$, in excess of the curvature $\Lambda$ already present inside each triangle. When $\Lambda=0$, all curvature is concentrated at vertices. We emphasize that the $\Lambda$ dependence of the deficit angle is solely inherited from the $\Lambda$ dependence of the internal angles $\beta^{\tau}_{v}$ through the relevant cosine law.

\subsubsection{The discretized action}\label{subsubsec:regge-LDG-action}

We proceed now with the discretization of the LDG action \eqref{eq:linear-dilaton-gravity-action}. First we discretize the dilaton contribution by assigning its values to vertices $v\in V(\Delta)$ as $\Phi_{v}:=\Phi(v)$. The curvature/metric data has been assigned to the triangulation by homogeneous curvature $\Lambda$ within simplices, the possibility of a non-trivial deficit angle \eqref{eq:deficit-angle-vertex} at vertices, and geodesic distances $l_{ij}$ between adjacent vertices. Putting both of these contributions to the curvature together we arrive at the following relation,
\begin{equation}
    \frac{1}{16 \pi G_{N}}\int_{M}d^{2}x\sqrt{g}\Phi R = \frac{K}{8\pi G_{N}}\sum_{\tau}\int_{\tau}d^{2}x\sqrt{g}\Phi +\frac{1}{8\pi G_{N}} \sum_{v}\Phi_{v}\delta_{v}(l;K).
\end{equation}
Subtracting the cosmological constant contribution (see  \eqref{eq:linear-dilaton-gravity-action}) on both sides of this expression and using the relation \eqref{eq:K-Lambda-relation}, the first term on the right hand side above cancels the cosmological volume contribution, giving the final homogeneous Regge discretized LDG action,
\begin{equation}
    S_{\mathrm{LDG},\Delta} = \frac{1}{8\pi G_{N}}\sum_{v}\Phi_{v}\delta_{v}(l;\Lambda).\label{eq:discretized-LDG-action}
\end{equation}
Henceforth we omit the overall numerical factors from \eqref{eq:discretized-LDG-action}. Varying with respect to $\Phi_{v}$ imposes $\delta_{v}(l;\Lambda)\approx 0$ for all $v$. As mentioned, this constraint combined with $K = \Lambda$ recovers exactly the continuum constraint $R-2\Lambda \approx 0$. Thus, within each simplex the continuum constraint has been solved identically. On the other hand, varying with respect to $l_{e}$ gives,
\begin{equation}
    \sum_{\tau\supset e}\sum_{v\in \tau}\Phi_{v}\pdv{\beta^{\tau}_{v}(l;\Lambda)}{l_{e}} \approx 0,\quad \forall e\in E(\Delta).
\end{equation}
In Appendix \ref{app:discrete-dilaton-eom} we show that this equation admits a direct geometrical interpretation as the gluing condition for solutions of the continuum dilaton equation in the different simplicies.

\subsection{Path integral discretization}\label{subsec:regge-path-int-discretization}

We now discretize the formal path integral 
\begin{equation}
    \int [dg][d\Phi] e^{iS_{\mathrm{LDG}}[g,\Phi;\Lambda]}.
\end{equation}
The continuum action may simply be replaced with its discretized version \eqref{eq:discretized-LDG-action}, thus the remaining ingredient is an appropriate discretized measure. 

For $\Lambda\leq0$, we will derive in \S\ref{sec:recovering-homo-regge}  an appropriate choice of measure via its correspondence with BF theory. It will be shown  that this induced measure leads indeed to a topological model, satisfying the expected Pachner move invariance. Therefore, we present here only a summary of the Regge path integral construction, leaving its justification to \S\ref{sec:recovering-homo-regge}.

Recalling \eqref{eq:vertex-gram-matrix}, for each $\tau$ we define the Jacobian factors\footnote{We emphasize that in the $\Lambda \neq 0$ case, $V_{\tau,\Lambda}$ is not the volume $\mathcal{V}_{\tau,\Lambda}$ of $\tau$ (cf. \eqref{eq:curved-triangle-area}). Only in the $\Lambda =0$ does this Jacobian factor recover two times the simplex volume.}
\begin{equation}{\label{eq:triangle-volume-gram-matrix}}
    V_{\tau,\Lambda}:= \begin{cases}
        \frac{1}{\abs{\Lambda}}\sqrt{\mathrm{det}(G_{\tau,\Lambda})},\: \Lambda \neq 0,\\
        2\mathcal{V}_{\tau,0},\quad \Lambda = 0.
    \end{cases}
\end{equation}
An edge length assignment is \textit{admissible} if each triangle satisfies the appropriate triangle inequalities and $V_{\tau,\Lambda}>0$. Denote the space of such assignments by $\cD_\Delta$. The measure we take is given by
\begin{equation}
    \Bigg[\prod_{e\in E(\Delta)}S_{\Lambda}(l_{e})dl_{e}\Bigg]\Bigg[\prod_{\tau\in T(\Delta)}\frac{1}{V_{\tau,\Lambda}}\Bigg]\Bigg[\prod_{v\in V(\Delta)}d\Phi_{v}\Bigg],
\end{equation}
where we recall $S_{\Lambda}(l_{\bullet})$ is defined by \eqref{eq:regge-curvature-functions}.

The domain of the dilaton determines the geometric curvature constraint imposed
by the path integral. For the vertex coupling
$e^{i\Phi_v\delta_v}$, integrating $\Phi_v$ over $\mathbb{R}$
imposes $\delta_v=0$, whereas summing $\Phi_v$ over $\mathbb{Z}$
imposes this condition only modulo $2\pi$. The latter prescription
therefore admits nonzero integer multiples of $2\pi$, which have
trivial rotational holonomy but need not define a regular geometry.
For nondegenerate Riemannian triangulations, these additional sectors
have negative deficit and describe conical excesses, locally
characterized as branched geometries \cite{Troyanov1986EuclideanSurfaces}. The real-valued dilaton excludes these
vertex singularities. It is therefore natural to consider the integration over $\mathbb{R}$. We will see in \S\ref{subsubsec:BF-to-homogeneous-space}, that the BF approach leads instead to an integration over $SO(2)$, hence it will  allow for branched geometries, unless restricted by hand.

Integrating each $\Phi_{v}$ over $\mathbb{R}$  gives the proposed Regge path integral for a given observable $\mathcal{O}(l)$ in terms of the edge length $l$.
\begin{equation}
    Z_{\mathrm{LDG}}^{\Delta}[\mathcal{O}]:= \int_{\mathcal{D}_{\Delta}}\Bigg[\prod_{e\in E(\Delta)}S_{\Lambda}(l_{e})dl_{e}\Bigg]\Bigg[\prod_{\tau\in T(\Delta)}\frac{1}{V_{\tau,\Lambda}}\Bigg]\Bigg[\prod_{v\in V(\Delta)}\delta_{\mathbb{R}}(\delta_{v}(l;\Lambda))\Bigg] \mathcal{O}(l).\label{eq:homo-Regge-path-integral}
\end{equation}
For  the BF-to-Regge comparison we restrict to a subcomplex triangulating a disk $D\subset M$.
This restriction is relevant in the present derivation because it avoids the additional global holonomy data associated with non-contractible cycles; a global equivalence on arbitrary topology is not established here. The restriction of \eqref{eq:homo-Regge-path-integral} to this subcomplex is given by,
\begin{equation}
   Z_{\mathrm{LDG}}^{\Delta,D}[\mathcal{O}]:= \int_{\mathcal{D}_{\Delta}}\Bigg[\prod_{e\in E(\Delta)}S_{\Lambda}(l_{e})dl_{e}\Bigg]\Bigg[\prod_{\tau\in T(\Delta)}\frac{1}{V_{\tau,\Lambda}}\Bigg]\Bigg[\prod_{v\in V^{0}(\Delta)}\delta_{\mathbb{R}}(\delta_{v}(l;\Lambda))\Bigg]\mathcal{O}(l).\label{eq:disk-homo-regge-path-integral}
\end{equation}
Here $\Delta$ denotes the subcomplex triangulating the disk $D\subset M$, and $V^0(\Delta)$ its interior vertices. We now construct the BF discretization underlying the BF to Regge correspondence, from which the above Regge path integral may be derived.

\section{A universal discretization of two-dimensional BF theory}\label{sec:split-discretization-of-2D-BF}

We construct a ribbon discretization of two-dimensional BF theory following \cite{dupuis_q-deformed_2020,meusburger_hopf_2021,meusburger_kitaev_2017,spies_poisson-geometric_2021}. Sections \ref{subsubsec:split-discretized-BF-action}--\ref{subsec:BF-pachner-moves} define its holonomies and constraints, establish gauge and cut invariance, and prove the local Pachner identities.

The construction is ``universal'' in the sense that a partial gauge fixing (``pgf'') procedure can be used to recover either the usual BF discretization on dual links \cite{livine_2d_2001} or its counterpart on triangle edges \cite{oeckl_generalized_2003,witten_quantum_1991,dittrich_flux_2015, dupuis_discretization_2017} (see \S\ref{subsubsec:recovering-standard-BF} and Figure \ref{fig:three-BF-discretizations}). The ribbon refinement therefore adds only pure gauge data. An alternative partial gauge fixing based on the Iwasawa decomposition of $G_{\sigma,\kappa}$ yields the \textit{split} BF model whose path integral is shown in \S\ref{sec:recovering-homo-regge} to recover the homogeneous Regge gravity path integral \eqref{eq:disk-homo-regge-path-integral}. We note that this type of gauge-fixing was already used in the three dimensional gravity context \cite{dupuis_q-deformed_2020}, to move from a Chern-Simons picture to the BF one. 

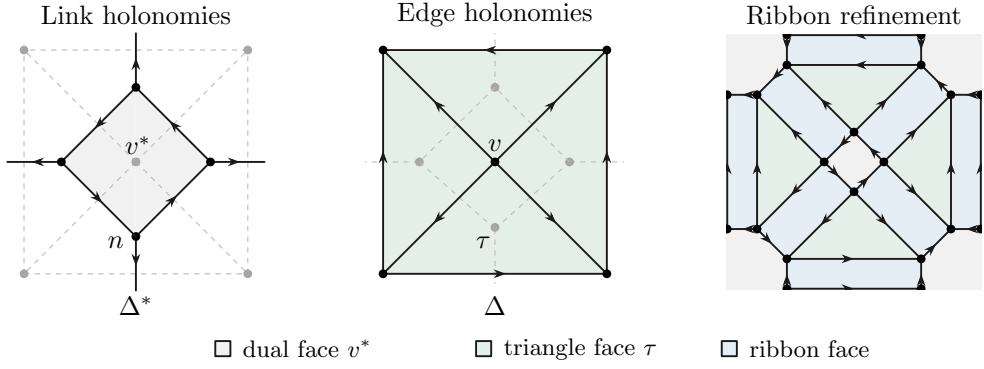
\begin{figure}[ht]
    \centering
    \begin{tikzpicture}[font=\small]
\path[fill=jtGrey] 
(1.9000000,3.110000) -- 
(2.886667,2.150000) -- 
(1.900000,1.210000) -- 
(1.900000,3.136667) -- 
(0.913333,2.150000) -- 
(1.900000,1.210000) --cycle;
\draw[jt faint] (1.900000,2.150000) -- (0.420000,0.670000);
\draw[jt faint] (0.420000,0.670000) -- (3.380000,0.670000);
\draw[jt faint] (1.900000,2.150000) -- (0.420000,3.630000);
\draw[jt faint] (3.380000,3.630000) -- (0.420000,3.630000);
\draw[jt faint] (1.900000,2.150000) -- (3.380000,3.630000);
\draw[jt faint] (0.420000,0.670000) -- (0.420000,3.630000);
\draw[jt faint] (3.380000,0.670000) -- (3.380000,3.630000);
\draw[jt faint] (1.900000,2.150000) -- (3.380000,0.670000);
\draw[jt edge,jt arrow] (1.900000,1.163333) -- 
(2.886667,2.150000);
\draw[jt edge,jt arrow] (1.900000,1.163333) -- (1.900000,0.448000);
\fill[black] (1.900000,1.163333) circle[radius=1.65pt];
\draw[jt edge,jt arrow] (2.886667,2.150000) --  
(1.900000,3.136667);
\draw[jt edge,jt arrow] (2.886667,2.150000) -- (3.602000,2.150000);
\fill[black] (2.886667,2.150000) circle[radius=1.65pt];
\draw[jt edge,jt arrow] (1.900000,3.136667) -- 
(0.913333,2.150000);
\draw[jt edge,jt arrow] (1.900000,3.136667) -- (1.900000,3.852000);
\fill[black] (1.900000,3.136667) circle[radius=1.65pt];
\draw[jt edge,jt arrow] (0.913333,2.150000) -- 
(1.900000,1.163333);
\draw[jt edge,jt arrow] (0.913333,2.150000) -- (0.198000,2.150000);
\fill[black] (0.913333,2.150000) circle[radius=1.65pt];
\fill[black!35] (1.900000,2.150000) circle[radius=1.65pt];
\fill[black!35] (0.420000,0.670000) circle[radius=1.65pt];
\fill[black!35] (3.380000,0.670000) circle[radius=1.65pt];
\fill[black!35] (3.380000,3.630000) circle[radius=1.65pt];
\fill[black!35] (0.420000,3.630000) circle[radius=1.65pt];
\node[jt label] at (1.929600,2.372000) {$v^{*}$};
\node[jt label] at (1.633600,1.084400) {$n$};
\node[jt label] at (1.900000,4.100000) {Link holonomies};
\node[jt label] at (1.900000,0.250000) {$\Delta^*$};
\path[fill=jtSage] (6.650000,2.150000) -- (5.170000,0.670000) -- (8.130000,0.670000) -- cycle;
\path[fill=jtSage] (6.650000,2.150000) -- (8.130000,0.670000) -- (8.130000,3.630000) -- cycle;
\path[fill=jtSage] (6.650000,2.150000) -- (8.130000,3.630000) -- (5.170000,3.630000) -- cycle;
\path[fill=jtSage] (6.650000,2.150000) -- (5.170000,3.630000) -- (5.170000,0.670000) -- cycle;

\draw[jt faint] (7.6366,2.150000) -- 
(8.3586,2.150000) ;
\draw[jt faint] (6.65,3.136667) -- 
(6.65,3.852) ;
\draw[jt faint] (5.65,2.150000) -- 
(4.928,2.150000) ;
\draw[jt faint] (6.65,1.263333) -- 
(6.65,0.5413) ;

\draw[jt faint] (6.65,1.263333) -- 
(7.6366,2.150000) ;
\fill[black!35] (7.6366,2.150000) circle[radius=1.65pt];
\draw[jt faint] 
(7.6366,2.150000) --  
(6.65,3.136667);
\fill[black!35] (6.65,3.136667) circle[radius=1.65pt];
\draw[jt faint] 
(6.65,3.136667) -- 
(5.65,2.150000);
\fill[black!35] (5.65,2.150000) circle[radius=1.65pt];
\draw[jt faint] (5.65,2.150000) -- 
(6.65,1.263333);

\draw[jt edge,jt arrow] (6.650000,2.150000) -- (5.170000,0.670000);
\draw[jt edge,jt arrow] (5.170000,0.670000) -- (8.130000,0.670000);
\draw[jt edge,jt arrow] (6.650000,2.150000) -- (5.170000,3.630000);
\draw[jt edge,jt arrow] (8.130000,3.630000) -- (5.170000,3.630000);
\draw[jt edge,jt arrow] (6.650000,2.150000) -- (8.130000,3.630000);
\draw[jt edge,jt arrow] (5.170000,0.670000) -- (5.170000,3.630000);
\draw[jt edge,jt arrow] (8.130000,0.670000) -- (8.130000,3.630000);
\draw[jt edge,jt arrow] (6.650000,2.150000) -- (8.130000,0.670000);
\fill[black] (6.650000,2.150000) circle[radius=1.65pt];
\fill[black] (5.170000,0.670000) circle[radius=1.65pt];
\fill[black] (8.130000,0.670000) circle[radius=1.65pt];
\fill[black] (8.130000,3.630000) circle[radius=1.65pt];
\fill[black] (5.170000,3.630000) circle[radius=1.65pt];
\fill[black!35] (6.650000,1.284400) circle[radius=1.65pt];
\node[jt label] at (6.650000,4.100000) {Edge holonomies};
\node[jt label] at (6.650000,0.250000) {$\Delta$};
\node[jt label] at (6.650000,2.372000) {$v$};
\node[jt label] at (6.50000,1.084400) {$\tau$};
\begin{scope}[jt arrow/.style={postaction={decorate},decoration={markings,
  mark=at position .5 with {\arrow{Stealth[length=1.65mm,width=1.15mm]}}}}]
\clip (9.712800,0.462800) rectangle (13.087200,3.837200);
\path[fill=jtGrey] (9.712800,0.462800) -- (13.087200,0.462800) -- (13.087200,3.837200) -- (9.712800,3.837200) -- cycle;
\path[fill=jtSage] (11.400000,1.755333) -- (10.512000,0.867333) -- (12.288000,0.867333) -- cycle;
\path[fill=jtSage] (11.794667,2.150000) -- (12.682667,1.262000) -- (12.682667,3.038000) -- cycle;
\path[fill=jtSage] (11.400000,2.544667) -- (12.288000,3.432667) -- (10.512000,3.432667) -- cycle;
\path[fill=jtSage] (11.005333,2.150000) -- (10.117333,3.038000) -- (10.117333,1.262000) -- cycle;
\path[fill=jtSage] (12.288000,0.472667) -- (10.512000,0.472667) -- (11.400000,-0.415333) -- cycle;
\path[fill=jtSage] (13.077333,3.038000) -- (13.077333,1.262000) -- (13.965333,2.150000) -- cycle;
\path[fill=jtSage] (10.512000,3.827333) -- (12.288000,3.827333) -- (11.400000,4.715333) -- cycle;
\path[fill=jtSage] (9.722667,1.262000) -- (9.722667,3.038000) -- (8.834667,2.150000) -- cycle;
\path[fill=jtBlue] (11.400000,1.755333) -- (10.512000,0.867333) -- (10.117333,1.262000) -- (11.005333,2.150000) -- cycle;
\path[fill=jtBlue] (10.512000,0.867333) -- (12.288000,0.867333) -- (12.288000,0.472667) -- (10.512000,0.472667) -- cycle;
\path[fill=jtBlue] (11.400000,1.755333) -- (12.288000,0.867333) -- (12.682667,1.262000) -- (11.794667,2.150000) -- cycle;
\path[fill=jtBlue] (12.682667,1.262000) -- (12.682667,3.038000) -- (13.077333,3.038000) -- (13.077333,1.262000) -- cycle;
\path[fill=jtBlue] (11.794667,2.150000) -- (12.682667,3.038000) -- (12.288000,3.432667) -- (11.400000,2.544667) -- cycle;
\path[fill=jtBlue] (12.288000,3.432667) -- (10.512000,3.432667) -- (10.512000,3.827333) -- (12.288000,3.827333) -- cycle;
\path[fill=jtBlue] (11.400000,2.544667) -- (10.512000,3.432667) -- (10.117333,3.038000) -- (11.005333,2.150000) -- cycle;
\path[fill=jtBlue] (10.117333,1.262000) -- (10.117333,3.038000) -- (9.722667,3.038000) -- (9.722667,1.262000) -- cycle;
\path[fill=jtGrey] (11.400000,1.755333) -- (11.794667,2.150000) -- (11.400000,2.544667) -- (11.005333,2.150000) -- cycle;
\path[fill=jtGrey] (10.512000,0.472667) -- (10.512000,0.867333) -- (10.117333,1.262000) -- (9.722667,1.262000) -- cycle;
\path[fill=jtGrey] (12.288000,0.472667) -- (13.077333,1.262000) -- (12.682667,1.262000) -- (12.288000,0.867333) -- cycle;
\path[fill=jtGrey] (12.288000,3.432667) -- (12.682667,3.038000) -- (13.077333,3.038000) -- (12.288000,3.827333) -- cycle;
\path[fill=jtGrey] (9.722667,3.038000) -- (10.117333,3.038000) -- (10.512000,3.432667) -- (10.512000,3.827333) -- cycle;
\draw[jt edge,jt arrow] (11.400000,1.755333) -- (10.512000,0.867333);
\draw[jt edge,jt arrow] (10.512000,0.867333) -- (12.288000,0.867333);
\draw[jt edge,jt arrow] (11.400000,1.755333) -- (12.288000,0.867333);
\draw[jt edge,jt arrow] (11.794667,2.150000) -- (12.682667,1.262000);
\draw[jt edge,jt arrow] (12.682667,1.262000) -- (12.682667,3.038000);
\draw[jt edge,jt arrow] (11.794667,2.150000) -- (12.682667,3.038000);
\draw[jt edge,jt arrow] (11.400000,2.544667) -- (12.288000,3.432667);
\draw[jt edge,jt arrow] (12.288000,3.432667) -- (10.512000,3.432667);
\draw[jt edge,jt arrow] (11.400000,2.544667) -- (10.512000,3.432667);
\draw[jt edge,jt arrow] (11.005333,2.150000) -- (10.117333,3.038000);
\draw[jt edge,jt arrow] (10.117333,1.262000) -- (10.117333,3.038000);
\draw[jt edge,jt arrow] (11.005333,2.150000) -- (10.117333,1.262000);
\draw[jt edge,jt arrow] (10.512000,0.472667) -- (12.288000,0.472667);
\draw[jt edge,jt arrow] (10.512000,0.472667) -- (11.400000,-0.415333);
\draw[jt edge,jt arrow] (11.400000,-0.415333) -- (12.288000,0.472667);
\draw[jt edge,jt arrow] (13.077333,1.262000) -- (13.077333,3.038000);
\draw[jt edge,jt arrow] (13.077333,1.262000) -- (13.965333,2.150000);
\draw[jt edge,jt arrow] (13.965333,2.150000) -- (13.077333,3.038000);
\draw[jt edge,jt arrow] (12.288000,3.827333) -- (10.512000,3.827333);
\draw[jt edge,jt arrow] (12.288000,3.827333) -- (11.400000,4.715333);
\draw[jt edge,jt arrow] (11.400000,4.715333) -- (10.512000,3.827333);
\draw[jt edge,jt arrow] (9.722667,1.262000) -- (9.722667,3.038000);
\draw[jt edge,jt arrow] (9.722667,3.038000) -- (8.834667,2.150000);
\draw[jt edge,jt arrow] (8.834667,2.150000) -- (9.722667,1.262000);
\draw[jt edge,jt arrow] (11.005333,2.150000) -- (11.400000,1.755333);
\draw[jt edge,jt arrow] (10.117333,1.262000) -- (10.512000,0.867333);
\draw[jt edge,jt arrow] (10.512000,0.472667) --
(10.512000,0.867333);
\draw[jt edge,jt arrow] 
(12.288000,0.472667)--
(12.288000,0.867333) ;
\draw[jt edge,jt arrow] (11.400000,1.755333) -- (11.794667,2.150000);
\draw[jt edge,jt arrow] (12.288000,0.867333) -- (12.682667,1.262000);
\draw[jt edge,jt arrow] 
(13.077333,1.262000) --
(12.682667,1.262000);
\draw[jt edge,jt arrow] 
(13.077333,3.038000)--(12.682667,3.038000);
\draw[jt edge,jt arrow] (11.794667,2.150000) -- (11.400000,2.544667);
\draw[jt edge,jt arrow] (12.682667,3.038000) -- (12.288000,3.432667);
\draw[jt edge,jt arrow] 
(12.288000,3.827333) --
(12.288000,3.432667) ;
\draw[jt edge,jt arrow] 
(10.512000,3.827333) --
(10.512000,3.432667);
\draw[jt edge,jt arrow] (11.400000,2.544667) -- (11.005333,2.150000);
\draw[jt edge,jt arrow] (10.512000,3.432667) -- (10.117333,3.038000);
\draw[jt edge,jt arrow] (10.117333,1.262000) -- (9.722667,1.262000);
\draw[jt edge,jt arrow] (10.117333,3.038000) -- (9.722667,3.038000);
\fill[black] (11.400000,1.755333) circle[radius=1.65pt];
\fill[black] (10.512000,0.867333) circle[radius=1.65pt];
\fill[black] (12.288000,0.867333) circle[radius=1.65pt];
\fill[black] (11.794667,2.150000) circle[radius=1.65pt];
\fill[black] (12.682667,1.262000) circle[radius=1.65pt];
\fill[black] (12.682667,3.038000) circle[radius=1.65pt];
\fill[black] (11.400000,2.544667) circle[radius=1.65pt];
\fill[black] (12.288000,3.432667) circle[radius=1.65pt];
\fill[black] (10.512000,3.432667) circle[radius=1.65pt];
\fill[black] (11.005333,2.150000) circle[radius=1.65pt];
\fill[black] (10.117333,3.038000) circle[radius=1.65pt];
\fill[black] (10.117333,1.262000) circle[radius=1.65pt];
\fill[black] (12.288000,0.472667) circle[radius=1.65pt];
\fill[black] (10.512000,0.472667) circle[radius=1.65pt];
\fill[black] (11.400000,-0.415333) circle[radius=1.65pt];
\fill[black] (13.077333,3.038000) circle[radius=1.65pt];
\fill[black] (13.077333,1.262000) circle[radius=1.65pt];
\fill[black] (13.965333,2.150000) circle[radius=1.65pt];
\fill[black] (10.512000,3.827333) circle[radius=1.65pt];
\fill[black] (12.288000,3.827333) circle[radius=1.65pt];
\fill[black] (11.400000,4.715333) circle[radius=1.65pt];
\fill[black] (9.722667,1.262000) circle[radius=1.65pt];
\fill[black] (9.722667,3.038000) circle[radius=1.65pt];
\fill[black] (8.834667,2.150000) circle[radius=1.65pt];
\end{scope}
\node[jt label] at (11.400000,4.100000) {Ribbon refinement};
\path[fill=jtGrey,jt edge] (2.950000,-0.420000) -- (3.150000,-0.420000) -- (3.150000,-0.220000) -- (2.950000,-0.220000) -- cycle;
\node[jt label,anchor=west,font=\footnotesize] at (3.260000,-0.320000) {dual face $v^{*}$};
\path[fill=jtSage,jt edge] (6.400000,-0.420000) -- (6.600000,-0.420000) -- (6.600000,-0.220000) -- (6.400000,-0.220000) -- cycle;
\node[jt label,anchor=west,font=\footnotesize] at (6.710000,-0.320000) {triangle face $\tau$};
\path[fill=jtBlue,jt edge] (9.650000,-0.420000) -- (9.850000,-0.420000) -- (9.850000,-0.220000) -- (9.650000,-0.220000) -- cycle;
\node[jt label,anchor=west,font=\footnotesize] at (9.960000,-0.320000) {ribbon face};
\end{tikzpicture}

    \caption{Three BF discretizations. Left: holonomies on oriented links of $\Delta^*$, with grey dual face $v^*$ and node $n=\tau^*$. Center: holonomies on oriented edges of $\Delta$, with green triangle face $\tau=n^*$. Right: the ribbon refinement, with holonomies on all sides of the blue ribbon faces (cf. Figure \ref{fig:ribboncells}). Partial gauge fixing collapses the ribbons to recover either of the first two discretizations.}
    \label{fig:three-BF-discretizations}
\end{figure}

\subsection{From continuum to discrete}\label{subsubsec:split-discretized-BF-action}

We start with the standard continuum two-dimensional BF action (omitting overall constant factors)
\begin{equation}
    S_{\mathrm{BF}}[A,B]:= \int_{M}\langle B, \cF(A)\rangle\label{eq:2d_bf_action}
\end{equation}
with gauge group $G_{\sigma,\kappa}$.

\paragraph{Discretization and ribbon structure.}

Recall the triangulation notation $\Delta$ of $M$ as described in the first paragraph of \S\ref{subsec:Regge-discretization-of-LDG}. A ribbon refinement of such a general triangulation is constructed as follows. 

For every $v\in \partial \tau$ choose a point near $v$ in the interior of $\tau$ and label it $(\tau, v)$. Across each original edge $e\in E(\Delta)$ connect $(\tau_{R}, v)$ to $(\tau_{L}, v)$ near $v = s(e)$ and near $v = t(e)$ respectively without any crossings, where we define a left triangle $\tau_{L}$ and right triangle $\tau_{R}$ respectively as the compatibly oriented and oppositely oriented boundary sides of $e$. This forms a ribbon refinement of $\Delta$ with the following data:
\begin{enumerate}
    \item A quadrilateral ribbon for each $e$, with \textit{corners}  $(\tau_{R}, s(e)), (\tau_{L},s(e)), (\tau_{L}, t(e)),(\tau_{R},t(e))$.
    \item A triangle face inside each original $\tau$.
    \item A polygonal dual face surrounding each original $v$.
\end{enumerate}
The first item above is the ribbon refinement data, whereas the latter two items are the standard data of $\Delta$ and $\Delta^{*}$.\footnote{One may think of ribbon corners as nodes in the dual face to $v$ by Poincar\'{e} duality, subject to the following key difference: in $\Delta^{*}$, the two faces associated to $s(e)$ and $t(e)$ share an oriented link crossing $e$, whereas in the ribbon construction this shared link is blown up into the two oriented sides $(\tau_{R},s(e))\to (\tau_{L},s(e))$ and $(\tau_{R},t(e))\to (\tau_{L},t(e))$ of the ribbon. Analogously, the single shared edge $e$ along which two triangles in $\Delta$ are glued is blown up into the two oriented sides $(\tau_{L},s(e))\to (\tau_{L},t(e))$ and $(\tau_{R},s(e)) \to (\tau_{R},t(e))$ of the ribbon.} The number of ribbon corners, ribbon sides, and total faces\footnote{The total number of faces is the sum of the number of ribbon faces $\abs{E(\Delta)}$, triangle faces $\abs{T(\Delta)}$, and dual faces $\abs{F(\Delta^{*})} = \abs{V(\Delta)}$.} are respectively,
\begin{equation}
    2\abs{E(\Delta)},\quad 4\abs{E(\Delta)}, \quad \abs{E(\Delta)}+\abs{T(\Delta)}+\abs{V(\Delta)}.
\end{equation}
The ribbon cellulation is obtained by a local refinement of the original triangulation and therefore represents the same surface. As a consistency check, it
 has the same Euler characteristic as the original triangulation $\Delta$, 
\begin{equation}
    2\abs{E(\Delta)}-4\abs{E(\Delta)}+(\abs{E(\Delta)}+\abs{T(\Delta)}+\abs{V(\Delta)}) = \abs{V(\Delta)}-\abs{E(\Delta)}+\abs{T(\Delta)}.
\end{equation}

We discretize $A$ using off-shell $G_{\sigma,\kappa}$-holonomies as follows. Choose a reference point $p\in M$ and define 
\begin{equation}
    g(\tau, v) = \mathcal{P}\mathrm{exp}\Big(\int_{p\to(\tau,v)}A\Big),\quad g(p) = 1
\end{equation}
such that $g(\tau,v)$ denotes the value of the field $g$ at the corner $(\tau, v)$.\footnote{We emphasize that $(\tau,v)$ is a \textit{label} for a single point in $M$, see Fig. \ref{fig:ribboncells}, it is not a tuple of a point ``$\tau$'' and a point ``$v$''. In this sense one still has the usual interpretation $g:M\to G_{\sigma,\kappa}$.} In a simply connected neighborhood of each face, flatness $\cF(A)\approx 0$ makes this expression independent of the chosen path. Our fundamental discrete variables are defined as the (off-shell) $G_{\sigma,\kappa}$-holonomies
\begin{subequations}
\begin{align}
\bar{h}^{v}_{e}&:=\mathcal{P}\mathrm{exp}\Big(\int_{(\tau_{R},v)\to (\tau_{L},v)}A\Big)\in G_{\sigma,\kappa},\quad v = s(e), t(e),\label{eq:framed-h-holonomy-def}\\
\bar{\ell}^{\tau}_{e}&:=\mathcal{P}\mathrm{exp}\Big(\int_{(\tau,s(e))\to (\tau,t(e))}A\Big)\in G_{\sigma,\kappa},\quad \tau = \tau_{L},\tau_{R}.\label{eq:framed-ell-holonomy-def}
\end{align}
\end{subequations}
The complete ribbon structure from this discretization of $A$ is captured by Figure \ref{fig:ribboncells}.
\begin{figure}[ht]
\centering
\begin{tikzpicture}[scale=.95,font=\small]
 \fill[ribbonblue!10] (-1,-2) rectangle (1,2);
 \fill[facegreen!9] (-4.5,-1.45)--(-1,-2)--(-1,2)--(-4.5,1.45)--cycle;
 \fill[facegreen!9] (4.5,-1.45)--(1,-2)--(1,2)--(4.5,1.45)--cycle;
 \fill[black!5] (-1,2)--(-1.4,3.9)--(1.4,3.9)--(1,2)--cycle;
 \fill[black!5] (-1,-2)--(-1.4,-3.9)--(1.4,-3.9)--(1,-2)--cycle;

 \coordinate (Ls) at (-1,-2);
 \coordinate (Rs) at (1,-2);
 \coordinate (Lt) at (-1,2);
 \coordinate (Rt) at (1,2);

 \draw[transport] (Rs)--node[below=3pt] {$\bar h_e^{s(e)}$} (Ls);
 \draw[transport] (Rt)--node[above=3pt] {$\bar h_e^{t(e)}$} (Lt);
 \draw[transport] (Ls)--node[left=4pt] {$\bar\ell_e^{\tau_L}$} (Lt);
 \draw[transport] (Rs)--node[right=4pt] {$\bar\ell_e^{\tau_R}$} (Rt);

 \foreach \p in {Ls,Rs,Lt,Rt} \node[corner] at (\p) {};
 \node[above left=2pt] at (Lt) {$(\tau_L,t(e))$};
 \node[above right=2pt] at (Rt) {$(\tau_R,t(e))$};
 \node[below left=2pt] at (Ls) {$(\tau_L,s(e))$};
 \node[below right=2pt] at (Rs) {$(\tau_R,s(e))$};

 \node[align=center,text=ribbonblue] at (0,0) {ribbon face\\at $e$};
 \node[align=center,text=facegreen] at (-3.0,0) {triangle face\\at $\tau_L$};
 \node[align=center,text=facegreen] at (3.0,0) {triangle face\\at $\tau_R$};
 \node[align=center] at (0,3.30) {dual face\\at $t(e)$};
 \node[align=center] at (0,-3.30) {dual face\\at $s(e)$};
\end{tikzpicture}
\caption{A ribbon and its four neighboring faces. All ribbon sides carry $G_{\sigma,\kappa}$-holonomies of the same connection $A$. The two paths from $(\tau_R,s(e))$ to $(\tau_L,t(e))$ have products $\bar h_e^{s(e)}\bar\ell_e^{\tau_L}$ and $\bar\ell_e^{\tau_R}\bar h_e^{t(e)}$ respectively. The shaded neighboring regions continue into the rest of the cellulation. We will often refer to $\bar{h}^{s(e)}_{e}, \bar{h}^{t(e)}_{e}$ as the \textit{short} side holonomies, and $\bar{\ell}^{\tau_{L}}_{e},\bar{\ell}^{\tau_{R}}_{e}$ as the \textit{long} side holonomies.}
\label{fig:ribboncells}
\end{figure}
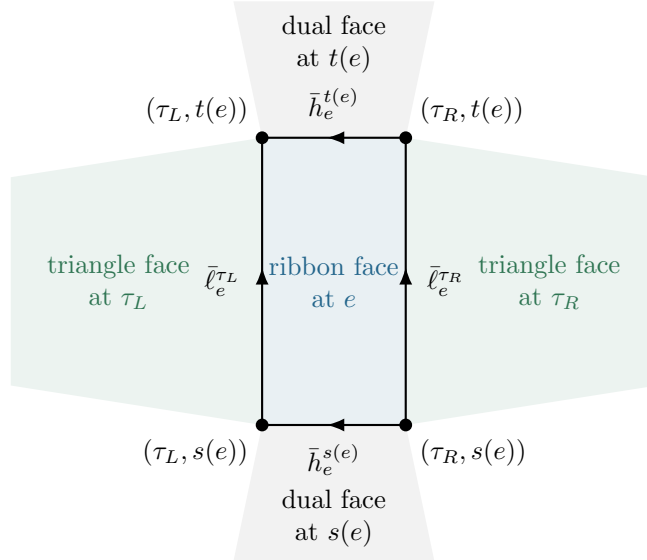

The next step is to determine the appropriate discrete constraints to impose in this scheme. The key point is that we have three regions on which to impose the continuum constraint $\cF(A)\approx 0$: on each ribbon face, on each triangle face, and on each dual face. In each case we have path independence on-shell, consequently,
\begin{equation}{\label{eq:framed-flatness}}
    \cF(A)\approx 0 \implies \begin{cases}
        \bar{h}^{v}_{e}\approx g(\tau_{R},v)^{-1}g(\tau_{L},v),\:\forall v\in \{s(e), t(e)\},\\
        \bar{\ell}^{\tau}_{e} \approx g(\tau, s(e))^{-1}g(\tau, t(e)),\:\forall \tau\in \{\tau_{L},\tau_{R}\}.
    \end{cases} 
\end{equation}
On each ribbon face the implication is that
\begin{equation}
    \bar{h}^{s(e)}_{e}\bar{\ell}^{\tau_{L}}_{e} \approx g(\tau_{R}, s(e))^{-1}g(\tau_{L},t(e)) \approx \bar{\ell}^{\tau_{R}}_{e}\bar{h}^{t(e)}_{e},
\end{equation}
therefore,
\begin{equation}
    \cF(A)\approx 0 \text{ on a ribbon face}\implies  \bar{\mathcal{W}}_{e}:=\bar{h}^{s(e)}_{e}\bar{\ell}^{\tau_{L}}_{e} (\bar{h}^{t(e)}_{e})^{-1}(\bar{\ell}^{\tau_{R}}_{e})^{-1}\approx 1.
\end{equation}
On each triangle face $\tau  = (v_{1}v_{2}v_{3})$ we define (off-shell),
\begin{equation}
    \bar{\cG}_{\tau}:= \bar{L}_{1}\bar{L}_{2}\bar{L}_{3},\quad \bar{L}_{i}:= (\bar{\ell}^{\tau}_{e_{i}})^{\epsilon_{\tau}(e_{i})},
\end{equation}
where $\bar{\ell}^{\tau}_{e_{i}}$ is defined by \eqref{eq:framed-ell-holonomy-def} and $\epsilon_{\tau}(e) = \pm 1$ depending on if the orientation of $e\in \partial\tau$ agrees with $\partial \tau$, therefore,
\begin{equation}
    \cF(A)\approx 0\text{ on a triangle face}\implies \bar{\cG}_{\tau}\approx g(\tau,v_1)^{-1}g(\tau,v_2)g(\tau,v_2)^{-1}g(\tau,v_3)g(\tau,v_3)^{-1}g(\tau,v_1)= 1.
\end{equation}
On each dual face, order the corners $(\tau_{i}, v)$ along its boundary and define
\begin{equation}
    \bar{\cR}_{v}:= \bar{h}_{1}\cdots \bar{h}_{m}
\end{equation}
where each $\bar{h}_{i}$ is shorthand notation for the (oriented) ribbon side holonomy from $(\tau_{i},v)$ to $(\tau_{i+1},v)$ as defined by \eqref{eq:framed-h-holonomy-def}, using an inverse when required and taking $\tau_{m+1} = \tau_{1}$. Therefore,
\begin{equation}
    \cF(A)\approx 0 \text{ on a dual face}\implies \bar{\cR}_{v}\approx \prod_{i=1}^{m}g(\tau_{i},v)^{-1}g(\tau_{i+1},v) = 1.
\end{equation}
We conclude then that the appropriate discretized constraints are,
\begin{subequations}{\label{eq:framed-discreted-constraints}}
\begin{align}
        \bar{h}^{s(e)}_{e}\bar{\ell}^{\tau_{L}}_{e} (\bar{h}^{t(e)}_{e})^{-1}(\bar{\ell}^{\tau_{R}}_{e})^{-1}&\approx 1\text{ on each ribbon face},\\
        \bar{\cG}_{\tau} &\approx 1\text{ on each triangle face},\\
        \bar{\cR}_{v} &\approx 1\text{ on each dual face.}
\end{align}
\end{subequations}
The continuum constraint $\cF(A)\approx 0$ has simply been split and smeared across the three face regions of the discretization. The path integral for this model is then taken to have the following distributional form,
\begin{equation}{\label{eq:ribbon-split-BF-path-integral}}
    Z^{\Delta}_{\mathrm{BF},\mathrm{rib}}:= \int \prod_{e\in E(\Delta)}\Big(d\bar{h}^{s(e)}_{e}d\bar{\ell}^{\tau_{L}}_{e}d\bar{\ell}^{\tau_{R}}_{e}d\bar{h}^{t(e)}_{e}\Big)\prod_{e\in E(\Delta)}\delta_{G_{\sigma,\kappa}}(\bar{\mathcal{W}}_{e})\prod_{\tau\in T(\Delta)}\delta_{G_{\sigma,\kappa}}(\bar{\cG}_{\tau})\prod_{v\in V(\Delta)}\delta_{G_{\sigma,\kappa}}(\bar{\cR}_{v}),
\end{equation}
where all delta functions above are on $G_{\sigma,\kappa}$ and all measures are the unimodular Haar measure on $G_{\sigma,\kappa}$. We treat these expressions as formal BF amplitudes, dividing out by gauge volumes when appropriate.

\subsection{Gauge transformations and path integral invariance}\label{subsubsec:gauge-sym-edge-redundancy}

The discretized gauge symmetries satisfied by \eqref{eq:ribbon-split-BF-path-integral} take a standard form reminiscent of lattice BF gauge theory. A gauge parameter $g_{\tau v}\in G_{\sigma,\kappa}$ acts at each corner $(\tau, v)$. The continuum transformation
\begin{equation}
    A\mapsto g^{-1}Ag+g^{-1}dg,
\end{equation}
includes the ribbon holonomy transformations
\begin{subequations}{\label{eq:framed-holonomy-gauge-transformations}}
    \begin{align}
    \bar{h}^{s(e)}_{e}{}' &= g^{-1}_{\tau_{R}s(e)}\bar{h}^{s(e)}_{e}g_{\tau_{L}s(e)},\quad \bar{\ell}^{\tau_{L}}_{e}{}' = g_{\tau_{L}s(e)}^{-1}\bar{\ell}^{\tau_{L}}_{e}g_{\tau_{L}t(e)},\\
    \bar{h}^{t(e)}_{e}{}' &= g_{\tau_{R}t(e)}^{-1}\bar{h}^{t(e)}_{e}g_{\tau_{L}t(e)},\quad \bar{\ell}^{\tau_{R}}_{e}{}' = g_{\tau_{R}s(e)}^{-1}\bar{\ell}^{\tau_{R}}_{e}g_{\tau_{R}t(e)}.
    \end{align}
\end{subequations}
We see all ribbon sides meeting at a common corner transform under the same $g_{\tau v}$. 

A \textit{cut} chooses the starting corner of a cyclic constraint product. Taking the cuts of $\bar{\cG}_\tau$ and $\bar{\cR}_v$ at $(\tau,v_0)$ and $(\tau_0,v)$ respectively gives
\begin{equation}
   \bar{\mathcal{W}}_{e}' = g_{\tau_{R}s(e)}^{-1}\bar{\mathcal{W}}_{e}g_{\tau_{R}s(e)},\quad \bar{\cG}_{\tau}' = g_{\tau v_{0}}^{-1}\bar{\cG}_{\tau}g_{\tau v_{0}},\quad \bar{\cR}_{v}' = g_{\tau_{0} v}^{-1}\bar{\cR}_{v}g_{\tau_{0}v}.
\end{equation}
The intermediate gauge parameters cancel, leaving conjugation at the cut.

Unimodularity of $G_{\sigma,\kappa}$ gives gauge and cut invariance. Each measure is invariant under left and right multiplication,
\begin{equation}
    d\bar{h}^{s(e)}_{e}{}' = d\bar{h}^{s(e)}_{e},\quad d\bar{\ell}^{\tau_{L}}_{e}{}' = d\bar{\ell}^{\tau_{L}}_{e},\quad  d\bar{h}^{t(e)}_{e}{}' = d\bar{h}^{t(e)}_{e},\quad d\bar{\ell}^{\tau_{R}}_{e}{}' = d\bar{\ell}^{\tau_{R}}_{e}.
\end{equation}
The delta functions of the constraints are invariant given that they satisfy the conjugation invariant identity,
\begin{equation}
    \delta_{G_{\sigma,\kappa}}(g^{-1}ag) = \delta_{G_{\sigma,\kappa}}(a),\:\forall g,a\in G_{\sigma,\kappa}
\end{equation}
again by unimodularity. A cyclic change of cut also conjugates the constraint, so the same identity gives cut invariance.

\subsubsection{Recovering the standard BF discretization}\label{subsubsec:recovering-standard-BF}

To recover the edge holonomy discretization of BF, we collapse link ribbon sides by gauge fixing the vertex faces. Exchanging gauge fixing of vertex faces with gauge fixing of triangle faces recovers the link discretization instead. The gauge invariance demonstrated in the previous section allows us to do this easily.

At each vertex, writing $\bar{\cR}_{v} = \bar{h}_{1}\cdots \bar{h}_{m}$, choose the tree for gauge fixing along the ribbon sides from $\bar{h}_{2}$ to $\bar{h}_{m}$ such that the corner of $\bar{h}_{1}$ is the root of the tree. At every non-root point in this tree we use the gauge fixing insertion,
\begin{equation}
   1= \int dg_{\tau_{i}v}\delta_{G_{\sigma,\kappa}}(g_{\tau_{i}v}^{-1}\bar{h}_{i}g_{\tau_{i+1}v}),\quad g_{\tau_{i}v} = \bar{h}_{i}g_{\tau_{i+1}v},\quad i = 2,\ldots m,\quad g_{\tau_{1}v} = 1.
\end{equation}
This sets $\bar h_2=\cdots=\bar h_m=1$ and reduces the vertex constraint to
\begin{equation}
    \bar{\cR}_{v}\mapsto \bar{h}_{1},
\end{equation}
thus integrating over the remaining $\bar{h}_{1}$ root variables sets them to $1$,
\begin{equation}
 \int d\bar h_1\,\delta_{G_{\sigma,\kappa}}(\bar{\cR}_{v})f(\bar h_1)\mapsto \int d\bar h_1\,\delta_{G_{\sigma,\kappa}}(\bar h_1)f(\bar h_1)=f(1),
\end{equation}
where ``$\mapsto$'' denotes the insertion of the gauge fixing factor. All $\bar{h}_{i}$ ribbon holonomies are thereby removed.

The ribbon flatness constraint is the last remaining piece, which after this partial gauge fixing takes the form
\begin{equation}
    \bar{\mathcal{W}}_{e}\mapsto \bar{\ell}^{\tau_{L}}_{e}(\bar{\ell}^{\tau_{R}}_{e})^{-1}.
\end{equation}
Integrating over all $\delta_{G_{\sigma,\kappa}}(\bar{\mathcal{W}}_{e})$ sets $\bar{g}_{e}:=\bar{\ell}^{\tau_{L}}_{e} = \bar{\ell}^{\tau_{R}}_{e}$, collapsing all ribbon faces. The only remaining data in \eqref{eq:ribbon-split-BF-path-integral} is triangle face constraints $\bar{\cG}_{\tau}$ depending on only the $G_{\sigma,\kappa}$-holonomies $\bar{g}_{e}$. Altogether this gives the final partially gauge-fixed form of the path integral \eqref{eq:ribbon-split-BF-path-integral},
\begin{equation}
    \int \prod_{e\in E(\Delta)}d\bar{g}_{e}\prod_{\tau\in T(\Delta)}\delta_{G_{\sigma,\kappa}}\Big(\prod_{e\subset \partial \tau}\bar{g}_{e}^{\epsilon_{\tau}(e)}\Big),
\end{equation}
which is exactly the form of the BF path integral when discretized using edge $G_{\sigma,\kappa}$-holonomies. Thus, as claimed, our ribbon blow up procedure leading to \eqref{eq:ribbon-split-BF-path-integral} only introduced pure gauge data. We reiterate that one could instead perform the gauge fixing on the triangle faces to recover the standard spin foam type discretization of $A$ into $G_{\sigma,\kappa}$-holonomies on links \cite{livine_2d_2001} instead of edges.

\subsection{Pachner move invariance}\label{subsec:BF-pachner-moves}

For both the $2\leftrightarrow 2$ and $3\leftrightarrow 1$ moves we choose one frame at each original boundary vertex, fixing all boundary link sides of the ribbons to the identity. The displayed boundary holonomies are held fixed in these common frames, and only interior holonomies are integrated. Invariance under both moves makes use of the identity
\begin{equation}
    \int dg \delta_{G_{\sigma,\kappa}}(ag^{-1})\delta_{G_{\sigma,\kappa}}(gb) = \delta_{G_{\sigma,\kappa}}(ab),
\end{equation}
which follows by unimodularity of $G_{\sigma,\kappa}$.

\paragraph{The $2\leftrightarrow 2$ Pachner move.}

We compare the two diagonal refinements of a quadrilateral $(v_1v_2v_3v_4)$ with the same boundary holonomies $\bar a,\bar b,\bar c,\bar d$ (Figure \ref{fig:BF-2-to-2}).

Beginning with insertion of a diagonal $(v_{1}v_{3})$, the quadrilateral is split into two triangles $\tau_{L} = (v_{1}v_{3}v_{4})$ and $\tau_{R} = (v_{1}v_{2}v_{3})$ glued by a ribbon face with edge holonomies $\bar{h}^{v_{1}}_{e_{13}}, \bar{h}^{v_{3}}_{e_{13}}, \bar{\ell}^{\tau_{L}}_{e_{13}}, \bar{\ell}^{\tau_{R}}_{e_{13}}$ where $v_{1} = s(e_{13})$ and $v_{3} = t(e_{13})$. The relevant constraints are those for each interior face,
\begin{equation}
    \bar{\mathcal{W}}_{e_{13}} = \bar{h}^{v_{1}}_{e_{13}}\bar{\ell}^{\tau_{L}}_{e_{13}}(\bar{h}^{v_{3}}_{e_{13}})^{-1}(\bar{\ell}^{\tau_{R}}_{e_{13}})^{-1},\quad \bar{\cG}_{\tau_{R}} = \bar{a}\bar{b}(\bar{\ell}^{\tau_{R}}_{e_{13}})^{-1},\quad \bar{\cG}_{\tau_{L}} = \bar{\ell}^{\tau_{L}}_{e_{13}}\bar{c}\bar{d}.
\end{equation}
Our boundary convention sets $\bar h_{e_{13}}^{v_1}=\bar h_{e_{13}}^{v_3}=1$. Integrating the two interior long sides gives
\begin{align}
    \mathcal{A}_{(13)}^{\mathrm{gf}}&:= \int d\bar{\ell}^{\tau_L}_{e_{13}}d\bar{\ell}^{\tau_R}_{e_{13}}\delta_{G_{\sigma,\kappa}}(\bar{\mathcal W}_{e_{13}})\delta_{G_{\sigma,\kappa}}(\bar{\cG}_{\tau_R})\delta_{G_{\sigma,\kappa}}(\bar{\cG}_{\tau_L})\nonumber\\
    &=\int d\bar{\ell}^{\tau_L}_{e_{13}}\delta_{G_{\sigma,\kappa}}\Big(\bar a\bar b(\bar{\ell}^{\tau_L}_{e_{13}})^{-1}\Big)\delta_{G_{\sigma,\kappa}}\Big(\bar{\ell}^{\tau_L}_{e_{13}}\bar c\bar d\Big).
\end{align}
Hence
\begin{equation}
    \mathcal{A}_{(13)}^{\mathrm{gf}} = \delta_{G_{\sigma,\kappa}}\Big(\bar a\bar b\bar c\bar d\Big).
\end{equation}
Performing the analogous construction for the case of a $(v_{2}v_{4})$ insertion, the quadrilateral is split into two triangles $\tau_{L} = (v_{1}v_{2}v_{4})$ and $\tau_{R} = (v_{2}v_{3}v_{4})$ glued by a ribbon face with edge holonomies $\bar{h}^{v_{2}}_{e_{24}}, \bar{h}^{v_{4}}_{e_{24}}, \bar{\ell}^{\tau_{L}}_{e_{24}}, \bar{\ell}^{\tau_{R}}_{e_{24}}$ where $v_{2} = s(e_{24})$ and $v_{4} = t(e_{24})$. The interior face constraints are then
\begin{equation}
        \bar{\mathcal{W}}_{e_{24}} = \bar{h}^{v_{2}}_{e_{24}}\bar{\ell}^{\tau_{L}}_{e_{24}}(\bar{h}^{v_{4}}_{e_{24}})^{-1}(\bar{\ell}^{\tau_{R}}_{e_{24}})^{-1},\quad \bar{\cG}_{\tau_{R}} = \bar{b}\bar{c}(\bar{\ell}^{\tau_{R}}_{e_{24}})^{-1},\quad \bar{\cG}_{\tau_{L}} = \bar{\ell}^{\tau_{L}}_{e_{24}}\bar{d}\bar{a}.
\end{equation}
With $\bar h_{e_{24}}^{v_2}=\bar h_{e_{24}}^{v_4}=1$, the same integration gives
\begin{align}
    \mathcal{A}_{(24)}^{\mathrm{gf}}&:=\int d\bar{\ell}^{\tau_L}_{e_{24}}d\bar{\ell}^{\tau_R}_{e_{24}}\delta_{G_{\sigma,\kappa}}(\bar{\mathcal W}_{e_{24}})\delta_{G_{\sigma,\kappa}}(\bar{\cG}_{\tau_R})\delta_{G_{\sigma,\kappa}}(\bar{\cG}_{\tau_L})\nonumber\\
    &=\int d\bar{\ell}^{\tau_L}_{e_{24}}\delta_{G_{\sigma,\kappa}}\Big(\bar b\bar c(\bar{\ell}^{\tau_L}_{e_{24}})^{-1}\Big)\delta_{G_{\sigma,\kappa}}\Big(\bar{\ell}^{\tau_L}_{e_{24}}\bar d\bar a\Big),
\end{align}
so
\begin{equation}
    \mathcal{A}_{(24)}^{\mathrm{gf}} = \delta_{G_{\sigma,\kappa}}\Big(\bar{b}\bar{c}\bar{d}\bar{a}\Big)=\delta_{G_{\sigma,\kappa}}\Big(\bar{a}\bar{b}\bar{c}\bar{d}\Big)
\end{equation}
which exactly matches $\mathcal{A}_{(13)}^{\mathrm{gf}}$, therefore, indeed the path integral \eqref{eq:ribbon-split-BF-path-integral} is invariant under the $2\leftrightarrow 2$ Pachner move.

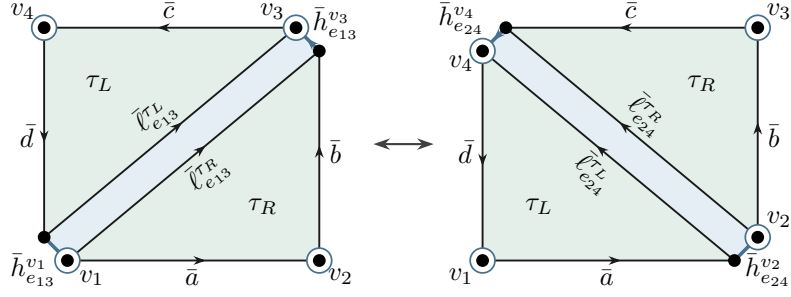
\begin{figure}[ht]
    \centering
    \begin{tikzpicture}[x=0.7cm,y=0.7cm]
  \begin{scope}[xshift=1.3cm]
    \coordinate (R1) at (.44,0);
    \coordinate (R2) at (5.2,0);
    \coordinate (R3) at (5.2,3.96);
    \coordinate (L1) at (0,.44);
    \coordinate (L3) at (4.76,4.4);
    \coordinate (L4) at (0,4.4);

    \fill[jtSage] (R1)--(R2)--(R3)--cycle;
    \fill[jtSage] (L1)--(L3)--(L4)--cycle;
    \fill[jtBlue] (R1)--(L1)--(L3)--(R3)--cycle;

    \draw[jt edge,jt arrow] (R1)--(R2);
    \draw[jt edge,jt arrow] (R2)--(R3);
    \draw[jt edge,jt arrow] (L3)--(L4);
    \draw[jt edge,jt arrow] (L4)--(L1);
    \draw[jt edge,jt arrow] (L1)--(L3);
    \draw[jt edge,jt arrow] (R1)--(R3);
    \draw[jt tree,jt arrow] (R1)--(L1);
    \draw[jt tree,jt arrow] (R3)--(L3);

    \node[jt label,anchor=north] at (2.8,-.10) {$\bar a$};
    \node[jt label,anchor=west] at (5.32,2.05) {$\bar b$};
    \node[jt label,anchor=south] at (2.35,4.50) {$\bar c$};
    \node[jt label,anchor=east] at (-.12,2.35) {$\bar d$};
    \node[jt label,rotate=39.75] at (2.10,2.75)
      {$\bar\ell^{\tau_L}_{e_{13}}$};
    \node[jt label,rotate=39.75] at (3.10,1.65)
      {$\bar\ell^{\tau_R}_{e_{13}}$};
    \node[jt label,rotate=0,anchor=north] at (-.2,.2)
      {$\bar h^{v_1}_{e_{13}}$};
    \node[jt label,rotate=0,anchor=south west] at (5.0,4.1)
      {$\bar h^{v_3}_{e_{13}}$};
    \node[jt label] at (1.05,3.35) {$\tau_L$};
    \node[jt label] at (4.13,1.02) {$\tau_R$};
    \node[jt label,anchor=north] at (0.9,-.10) {$v_{1}$};
    \node[jt label,anchor=north west] at (5.29,-.08) {$v_2$};
    \node[jt label,anchor=south west] at (4.0,4.48) {$v_3$};
    \node[jt label,anchor=south east] at (-.09,4.48) {$v_4$};

    \foreach \point in {R1,R2,R3,L1,L3,L4}
      \node[corner] at (\point) {};
    \foreach \point in {R1,R2,L3,L4}{
      \node[circle,draw=jtInkBlue,fill=white,line width=.65pt,
        inner sep=3.3pt] at (\point) {};
      \node[corner] at (\point) {};
    }

  \end{scope}

  \draw[jt move,{Stealth[length=2.2mm,width=1.5mm]}-{Stealth[length=2.2mm,width=1.5mm]}]
    (8.08,2.2)--(9.22,2.2);
  
  \begin{scope}[xshift=7.1cm]
    \coordinate (L1) at (0,0);
    \coordinate (L2) at (4.76,0);
    \coordinate (L4) at (0,3.96);
    \coordinate (R2) at (5.2,.44);
    \coordinate (R3) at (5.2,4.4);
    \coordinate (R4) at (.44,4.4);

    \fill[jtSage] (L1)--(L2)--(L4)--cycle;
    \fill[jtSage] (R2)--(R3)--(R4)--cycle;
    \fill[jtBlue] (R2)--(L2)--(L4)--(R4)--cycle;

    \draw[jt edge,jt arrow] (L1)--(L2);
    \draw[jt edge,jt arrow] (R2)--(R3);
    \draw[jt edge,jt arrow] (R3)--(R4);
    \draw[jt edge,jt arrow] (L4)--(L1);
    \draw[jt edge,jt arrow] (L2)--(L4);
    \draw[jt edge,jt arrow] (R2)--(R4);
    \draw[jt tree,jt arrow] (R2)--(L2);
    \draw[jt tree,jt arrow] (R4)--(L4);

    \node[jt label,anchor=north] at (2.35,-.10) {$\bar a$};
    \node[jt label,anchor=west] at (5.32,2.35) {$\bar b$};
    \node[jt label,anchor=south] at (2.80,4.50) {$\bar c$};
    \node[jt label,anchor=east] at (-.12,2.05) {$\bar d$};
    \node[jt label,rotate=-39.75] at (2.10,1.65)
      {$\bar\ell^{\tau_L}_{e_{24}}$};
    \node[jt label,rotate=-39.75] at (3.10,2.75)
      {$\bar\ell^{\tau_R}_{e_{24}}$};
    \node[jt label,anchor=north] at (5.4,0.2) 
      {$\bar h^{v_2}_{e_{24}}$};
    \node[jt label,anchor=south] at (-0.4,4.2)
        {$\bar h^{v_4}_{e_{24}}$};
    \node[jt label] at (1.07,1.02) {$\tau_L$};
    \node[jt label] at (4.15,3.35) {$\tau_R$};

    \node[jt label,anchor=north east] at (-.09,-.08) {$v_1$};
    \node[jt label,anchor=north west] at (5.29,1) {$v_2$};
    \node[jt label,anchor=south west] at (5.29,4.48) {$v_3$};
    \node[jt label,anchor=south east] at (-.09,3.5) {$v_4$};

    \foreach \point in {L1,L2,L4,R2,R3,R4}
      \node[corner] at (\point) {};
    \foreach \point in {L1,R2,R3,L4}{
      \node[circle,draw=jtInkBlue,fill=white,line width=.65pt,
        inner sep=3.3pt] at (\point) {};
      \node[corner] at (\point) {};
    }
    
  \end{scope}
\end{tikzpicture}
    \caption{The ribbon BF $2\leftrightarrow2$ move with common boundary holonomies $\bar a,\bar b,\bar c,\bar d$. Circled corners mark the retained boundary frames, the blue ribbon sides are fixed to the identity by this frame convention.}
    \label{fig:BF-2-to-2}
\end{figure}

\paragraph{The $3\leftrightarrow 1$ Pachner move.}

Insert $v_0$ inside $\tau_0=(v_1v_2v_3)$, creating edges $e_4=(v_0v_2)$, $e_5=(v_0v_3)$, and $e_6=(v_0v_1)$ and their ribbons. The boundary holonomies remain $\bar a,\bar b,\bar c$ (Figure \ref{fig:BF-3-to-1}), with $\bar h_{e_4}^{v_2}=\bar h_{e_5}^{v_3}=\bar h_{e_6}^{v_1}=1$.
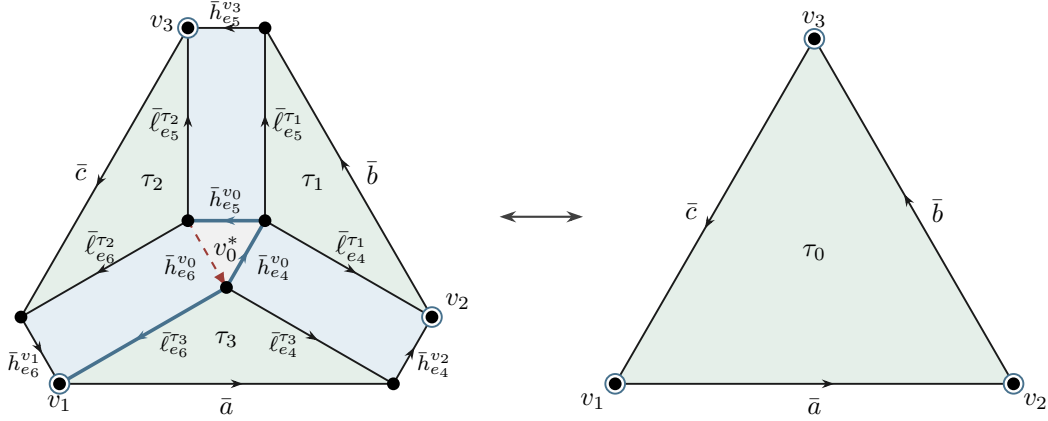
\begin{figure}[ht]
    \centering
    \begin{tikzpicture}[x=.85cm,y=.85cm, every node/.style={jt label}]
  
  \coordinate (C) at (3.10,2.00);
  \coordinate (F) at ($(C)+(60:1.20)$);
  \coordinate (I) at ($(C)+(120:1.20)$);
  \coordinate (B) at ($(C)+(-30:3.00)$);
  \coordinate (D) at ($(F)+(-30:3.00)$);
  \coordinate (E) at ($(F)+(90:3.00)$);
  \coordinate (G) at ($(I)+(90:3.00)$);
  \coordinate (A) at ($(C)+(210:3.00)$);
  \coordinate (H) at ($(I)+(210:3.00)$);

  \fill[jtSage] (A)--(B)--(C)--cycle;
  \fill[jtSage] (D)--(E)--(F)--cycle;
  \fill[jtSage] (G)--(H)--(I)--cycle;
  \fill[jtBlue] (C)--(F)--(D)--(B)--cycle;
  \fill[jtBlue] (F)--(I)--(G)--(E)--cycle;
  \fill[jtBlue] (I)--(C)--(A)--(H)--cycle;
  \fill[jtGrey] (C)--(F)--(I)--cycle;

  \draw[jt edge,jt arrow] (A)--(B)
    node[midway,below=4pt] {$\bar a$};
  \draw[jt edge,jt arrow] (D)--(E)
    node[midway,right=5pt] {$\bar b$};
  \draw[jt edge,jt arrow] (G)--(H)
    node[midway,left=5pt] {$\bar c$};

  \draw[jt edge,jt arrow] (C)--(B);
  \node[font=\scriptsize] at (4.00,1.1)
    {$\bar\ell_{e_4}^{\tau_3}$};
  \draw[jt edge,jt arrow] (F)--(D)
    node[pos=.52,above=1pt] {$\bar\ell_{e_4}^{\tau_1}$};
  \draw[jt edge,jt arrow] (F)--(E)
    node[midway,right=1pt] {$\bar\ell_{e_5}^{\tau_1}$};
  \draw[jt edge,jt arrow] (I)--(G)
    node[midway,left=0.5pt] {$\bar\ell_{e_5}^{\tau_2}$};
  \draw[jt edge,jt arrow] (I)--(H)
    node[pos=0.52,above=1pt] {$\bar\ell_{e_6}^{\tau_2}$};
  \draw[jt tree,jt arrow] (C)--(A);
  \node[font=\scriptsize] at (2.30,1.1)
    {$\bar\ell_{e_6}^{\tau_3}$};

  \draw[jt tree,jt arrow] (C)--(F);
  \node[font=\scriptsize] at (3.85,2.30)
    {$\bar h_{e_4}^{v_0}$};
  \draw[jt tree,jt arrow] (F)--(I)
    node[midway,above=1pt,font=\scriptsize]
      {$\bar h_{e_5}^{v_0}$};
 
  \draw[freeside] (I)--(C);
  \node[font=\scriptsize] at (2.4,2.30)
    {$\bar h_{e_6}^{v_0}$};

  \draw[jt edge,jt arrow] (B)--(D);
  \node[anchor=west,font=\scriptsize]
    at (6.0,.80) {$\bar h_{e_4}^{v_2}$};
  \draw[jt edge,jt arrow] (E)--(G);
  \node[above=5pt,font=\scriptsize]
    at (3.10,5.8) {$\bar h_{e_5}^{v_3}$};
  \draw[jt edge,jt arrow] (H)--(A);
  \node[anchor=east,font=\scriptsize]
    at (.3,0.8) {$\bar h_{e_6}^{v_1}$};

  \node at (3.10,1.2) {$\tau_3$};
  \node at (4.45,3.63) {$\tau_1$};
  \node at (1.95,3.63) {$\tau_2$};
   \node at (3.10,2.6) {$v_{0}^{*}$};
  
  \node[below=3pt] at (.50,.50) {$v_1$};
  \node[below=3pt] at (6.7,2.1) {$v_2$};
  \node[above=3pt] at (2.10,5.8) {$v_3$};

  \foreach \p in {A,D,G} {
    \draw[draw=jtInkBlue,fill=white,line width=.72pt]
      (\p) circle[radius=3.7pt];
  }
  \foreach \p in {A,B,C,D,E,F,G,H,I} \node[corner] at (\p) {};

  \draw[jt move,{Stealth[length=2.2mm,width=1.5mm]}-{Stealth[length=2.2mm,width=1.5mm]}]
    (7.35,3.10)--(8.65,3.10);

  \coordinate (J) at (9.15,.50);
  \coordinate (K) at ($(J)+(0:6.20)$);
  \coordinate (L) at ($(J)+(60:6.20)$);
  \fill[jtSage] (J)--(K)--(L)--cycle;
  \draw[jt edge,jt arrow] (J)--(K)
    node[midway,below=4pt] {$\bar a$};
  \draw[jt edge,jt arrow] (K)--(L)
    node[midway,right=5pt] {$\bar b$};
  \draw[jt edge,jt arrow] (L)--(J)
    node[midway,left=5pt] {$\bar c$};
  \node at (12.25,2.55) {$\tau_0$};
  \node[below left=3pt] at (J) {$v_1$};
  \node[below right=3pt] at (K) {$v_2$};
  \node[above=3pt] at (L) {$v_3$};
  \foreach \p in {J,K,L} {
    \draw[draw=jtInkBlue,fill=white,line width=.72pt]
      (\p) circle[radius=3.7pt];
    \node[corner] at (\p) {};
  }
\end{tikzpicture}
    \caption{The ribbon BF $3\leftrightarrow1$ move with common boundary holonomies $\bar a,\bar b,\bar c$. The blue tree consists of $\bar h_{e_4}^{v_0},\bar h_{e_5}^{v_0},\bar\ell_{e_6}^{\tau_3}$ and is rooted at $(\tau_3,v_1)$. Circled corners mark the retained boundary frames; their connecting boundary short sides are fixed to the identity.}
    \label{fig:BF-3-to-1}
\end{figure}
The relevant amplitude for the single triangle configuration is simply,
\begin{equation}
    \mathcal{A}_{\tau_{0}} = \delta_{G_{\sigma,\kappa}}(\bar{a}\bar{b}\bar{c})
\end{equation}
The three triangle amplitude is given by integrating the following constraints over all ribbon holonomy data,
\begin{subequations}
\begin{align}
    \bar{\cR}_{v_{0}}&=\bar{h}^{v_{0}}_{e_{4}}\bar{h}^{v_{0}}_{e_{5}}\bar{h}^{v_{0}}_{e_{6}}\\
    \bar{\mathcal{W}}_{e_{4}} = \bar{h}^{v_{0}}_{e_{4}}\bar{\ell}^{\tau_{1}}_{e_{4}}(\bar{h}^{v_{2}}_{e_{4}})^{-1}(\bar{\ell}^{\tau_{3}}_{e_{4}})^{-1},\quad   \bar{\mathcal{W}}_{e_{5}} &= \bar{h}^{v_{0}}_{e_{5}}\bar{\ell}^{\tau_{2}}_{e_{5}}(\bar{h}^{v_{3}}_{e_{5}})^{-1}(\bar{\ell}^{\tau_{1}}_{e_{5}})^{-1},\quad   \bar{\mathcal{W}}_{e_{6}} = \bar{h}^{v_{0}}_{e_{6}}\bar{\ell}^{\tau_{3}}_{e_{6}}(\bar{h}^{v_{1}}_{e_{6}})^{-1}(\bar{\ell}^{\tau_{2}}_{e_{6}})^{-1}\\
    \bar{\cG}_{\tau_{1}} = \bar{\ell}^{\tau_{1}}_{e_{4}}\bar{b}(\bar{\ell}^{\tau_{1}}_{e_{5}})^{-1},\quad \bar{\cG}_{\tau_{2}} &= \bar{\ell}^{\tau_{2}}_{e_{5}}\bar{c}(\bar{\ell}^{\tau_{2}}_{e_{6}})^{-1},\quad \bar{\cG}_{\tau_{3}} = \bar{\ell}^{\tau_{3}}_{e_{6}}\bar{a}(\bar{\ell}^{\tau_{3}}_{e_{4}})^{-1}
\end{align}
\end{subequations}
thus,
\begin{equation}
    \mathcal{A}_{\tau_{1}\tau_{2}\tau_{3}} = \int \Big(\prod_{i=4}^{6}d\bar{h}^{v_{0}}_{e_{i}}d\bar{\ell}^{\tau_{L}}_{e_{i}}d\bar{\ell}^{\tau_{R}}_{e_{i}}\Big)\Big(\prod_{i=4}^{6}\delta_{G_{\sigma,\kappa}}(\bar{\mathcal{W}}_{e_{i}})\Big)\Big(\prod_{j=1}^{3}\delta_{G_{\sigma,\kappa}}(\bar{\cG}_{\tau_{j}})\Big)\delta_{G_{\sigma,\kappa}}(\bar{\cR}_{v_{0}})
\end{equation}
We perform some gauge fixing by rooting the tree consisting of $\bar h_{e_4}^{v_0},\bar h_{e_5}^{v_0},\bar\ell_{e_6}^{\tau_3}$ at the boundary corner $(\tau_3,v_1)$. Gauge transformations at its three interior corners set these holonomies to the identity without changing the boundary data. The remaining $\bar h_{e_6}^{v_0}$ is then set to the identity by $\delta_{G_{\sigma,\kappa}}(\bar{\cR}_{v_0})$, giving
\begin{align}
    \mathcal{A}_{\tau_{1}\tau_{2}\tau_{3}}^{\mathrm{gf}} =& \int d\bar{\ell}^{\tau_{1}}_{e_{4}}d\bar{\ell}^{\tau_{3}}_{e_{4}}d\bar{\ell}^{\tau_{2}}_{e_{5}}d\bar{\ell}^{\tau_{1}}_{e_{5}}d\bar{\ell}^{\tau_{2}}_{e_{6}}\Bigg[\delta_{G_{\sigma,\kappa}}\Big(\bar{\ell}^{\tau_{1}}_{e_{4}}(\bar{\ell}^{\tau_{3}}_{e_{4}})^{-1}\Big)\delta_{G_{\sigma,\kappa}}\Big(\bar{\ell}^{\tau_{2}}_{e_{5}}(\bar{\ell}^{\tau_{1}}_{e_{5}})^{-1}\Big)\delta_{G_{\sigma,\kappa}}\Big(\bar{\ell}^{\tau_{2}}_{e_{6}}\Big)\nonumber\\
    &\times \delta_{G_{\sigma,\kappa}}\Big(\bar{\ell}^{\tau_{1}}_{e_{4}}\bar{b}(\bar{\ell}^{\tau_{1}}_{e_{5}})^{-1}\Big)\delta_{G_{\sigma,\kappa}}\Big(\bar{\ell}^{\tau_{2}}_{e_{5}}\bar{c}(\bar{\ell}^{\tau_{2}}_{e_{6}})^{-1}\Big)\delta_{G_{\sigma,\kappa}}\Big(\bar{a}(\bar{\ell}^{\tau_{3}}_{e_{4}})^{-1}\Big)\Bigg]
\end{align}
Performing the integrations over the gauge-fixed ribbon constraints sets $\bar{\ell}^{\tau_{1}}_{e_{4}} = \bar{\ell}^{\tau_{3}}_{e_{4}}, \bar{\ell}^{\tau_{2}}_{e_{5}} = \bar{\ell}^{\tau_{1}}_{e_{5}},$ and $\bar{\ell}^{\tau_{2}}_{e_{6}} = 1$, giving,
\begin{align}
    \mathcal{A}_{\tau_{1}\tau_{2}\tau_{3}}^{\mathrm{gf}} &=\int d\bar{\ell}^{\tau_{3}}_{e_{4}}d\bar{\ell}^{\tau_{1}}_{e_{5}}\Bigg[\delta_{G_{\sigma,\kappa}}\Big(\bar{\ell}^{\tau_{3}}_{e_{4}}\bar{b}(\bar{\ell}^{\tau_{1}}_{e_{5}})^{-1}\Big)\delta_{G_{\sigma,\kappa}}\Big(\bar{\ell}^{\tau_{1}}_{e_{5}}\bar{c}\Big)\delta_{G_{\sigma,\kappa}}\Big(\bar{a}(\bar{\ell}^{\tau_{3}}_{e_{4}})^{-1}\Big)\Bigg]\nonumber\\
    &=\delta_{G_{\sigma,\kappa}}(\bar{a}\bar{b}\bar{c})
\end{align}
which is exactly equal to the one triangle amplitude $\mathcal{A}_{\tau_{0}}$. Therefore, as claimed, the gauge-fixed local amplitude agrees with the one-triangle amplitude, establishing the $3\leftrightarrow1$ Pachner identity after the corresponding gauge redundancy has been removed.

\subsection{Split BF: The Regge adapted partial gauge fixing}\label{subsec:regge-adapted-pgf}

In this section we detail another partial gauge fixing of \eqref{eq:ribbon-split-BF-path-integral} which reduces it to a form that is well adapted to recovering the path integral of homogeneous Regge gravity. Motivated by the continuum correspondence between BF and LDG  in \S\ref{subsec:continuum-BF-to-regge}, the  goal is to reduce the triangle and vertex constraints to complementary subgroup sectors of $G_{\sigma,\kappa}$. This is achieved by using an Iwasawa-type matched-pair splitting $G_{\sigma,\kappa}\sim H_{\sigma}\bowtie AN_{\sigma,\kappa}$. In the Euclidean hyperbolic case this is the usual Iwasawa factorization, while in the flat Euclidean case it reduces to the semidirect-product factorization. The corresponding group factorization is global for the Euclidean cases $\kappa=0,-1$ and for  the flat Lorentzian case. For curved Lorentzian cases the  splitting  is still available locally.

Indeed, for $\Lambda\ne0$, the $P_a$ in the Cartan basis do not span a subalgebra, obstructing an independent subgroup discretization of the torsion constraint. The Iwasawa basis removes this obstruction. Related constructions in three-dimensional gravity \cite{dupuis_origin_2020,girelli_beyond_2025} induce this change through a boundary volume term, whereas here we change the Lie algebra basis directly.

Unlike the ribbon collapse in \S\ref{subsubsec:recovering-standard-BF}, this procedure removes only the $H_\sigma$ part of triangle holonomies and the $AN_{\sigma,\kappa}$ part of vertex holonomies. It retains the ribbon structure and the subgroup constraints which are shown in \S\ref{sec:recovering-homo-regge} to be related to triangle closure and the vanishing of the deficit angles.

\subsubsection{The Iwasawa decomposition}

Recall from \S\ref{subsec:continuum-BF-to-regge} that by respectively splitting the $A$ and $B$ fields into components using the definitions \eqref{eq:A-cartan} and \eqref{eq:B-cartan} one arrives at the following form of the continuum BF action,
\begin{equation}
    S_{\mathrm{BF}}[A,B] = S_{\mathrm{BF}}[e,\omega, X, \Phi] = \int_{M}\Big(\Phi F + X_{a}T^{a}\Big),\label{eq:cartan-BF-action}
\end{equation}
where $F$ and $T^{a}$ are given by \eqref{eq:2d_torsion_curvature}. To switch from the Cartan decomposition to the Iwasawa decomposition we employ the change of Lie algebra basis
\begin{equation}
    P_{a}\mapsto P'_{a} := P_{a}-n_{a}J,\label{eq:P-Iwasawa-basis}
\end{equation}
where $n_{a}$ is a fixed\footnote{By fixed we mean $dn_{a} = \delta n_{a} = 0$.} internal vector satisfying
\begin{equation}
    n_{a}n^{a} =\eta^{ab}n_{a}n_{b}= -\sigma\frac{\kappa}{l^{2}}.\label{eq:n-normalization}
\end{equation}
The Lie brackets \eqref{eq:isometry_algebra} of $\fg_{\sigma,\kappa}$ in this new basis take the form
\begin{equation}
    \comm{J}{P'_{a}} = \epsilon_{a}{}^{b}(P'_{b}+n_{b}J),\quad  \comm{P'_{a}}{P'_{b}} = \epsilon_{ab}n^{c}P'_{c},
\end{equation}
which we see allows $\comm{P'_{a}}{P'_{b}}$ to close and to form the algebra $\an_{\sigma,\kappa}$.

\medskip

The Iwasawa decomposition of $\fg_{\sigma,\kappa}$ is therefore
\begin{equation}
     \mathfrak{g}_{\sigma,\kappa}\cong \mathfrak{h}_{\sigma}\bowtie \mathfrak{an}_{\sigma,\kappa}\cong \mathfrak{an}_{\sigma,\kappa}\bowtie \mathfrak{h}_{\sigma},
\end{equation}
where $\mathfrak{h}_{\sigma}$ and $\mathfrak{an}_{\sigma,\kappa}$ are respectively spanned by $J$ and $P'_{a}$, subject to the appropriate signature and brackets defined by the values of the parameters $\sigma$ and $\kappa$.

\medskip 

We should emphasize the Euclidean case with positive cosmological constant. In this case the normalization condition \eqref{eq:n-normalization} requires the vector $n_{a}$ to be purely imaginary. This reflects the fact that the compact real Lie algebra $\so(3)$ cannot contain a non-compact real subalgebra. Consequently, rather than working directly with $\so(3)$, one must pass to its complexification and subsequently impose appropriate reality conditions in order to recover the desired real geometry. A closely related situation arises in three-dimensional gravity \cite{dupuis_origin_2020}. The treatment of the corresponding case $\g_{1,1}$ therefore requires a more detailed analysis of the complexified algebra, together with the associated reality conditions. Such technical developments go beyond the scope of the present paper, and in particular beyond the argument we wish to develop here. We therefore leave a systematic analysis of this case for future work.

\medskip

At the level of the $A$ field, to keep it invariant, a compensating field redefinition is induced,
\begin{equation}
    A = \omega J +e^{a}P_{a} = \Omega J + e^{a}P'_{a}\iff \Omega:= \omega +e^{a}n_{a}.\label{eq:Omega-Iwasawa-basis}
\end{equation}
At the level of the BF action \eqref{eq:2d_bf_action}, to ensure the pairing \eqref{eq:g-g*-Cartan-pairing} is invariant under this change of Lie algebra basis a corresponding change of basis and field redefinition for $B$ is induced,
\begin{equation}
    B = \Phi \tilde{J} + X_{a}\tilde{P}^{a} = \Phi \tilde{J}'+\chi_{a}\tilde{P}^{a} \iff  \tilde{J}':=\tilde{J}+n_{a}\tilde{P}^{a} \text{ and }\chi_{a}:=  X_{a}-\Phi n_{a}.
\end{equation}
A quick check verifies that indeed the pairing relations \eqref{eq:g-g*-Cartan-pairing} are preserved,
\begin{equation}
    \langle J,\tilde{J}'\rangle = 1,\quad \langle P'_{a},\tilde{P}^{b}\rangle = \delta_{a}^{b},\quad \langle J,\tilde{P}^{a}\rangle = \langle P'_{a},\tilde{J}'\rangle = 0,
\end{equation}
and that the dressing actions \eqref{eq:Cartan-Dressing-Lie-Alg} in terms of $J, P'_{a}, \tilde{J}',$ and $\tilde{P}^{a}$ are given by,
\begin{equation}
    \tilde{J}'\lhd J = n_{b}\epsilon_{a}{}^{b}\tilde{P}^{a}, \quad \tilde{J}'\lhd P'_{a} = -n_{b}\epsilon_{a}{}^{b}\tilde{J}',\quad \tilde{P}^{a}\lhd J = \epsilon_{b}{}^{a}\tilde{P}^{b},\quad \tilde{P}^{a}\lhd P'_{b} = -\epsilon_{b}{}^{a}\tilde{J}'+n^{a}\epsilon_{bc}\tilde{P}^{c}.\label{eq:Iwaswa-Dressing-Lie-Alg}
\end{equation}
The punchline is that now the component fields of $A$ and $B$ are valued in the Iwasawa defined subalgebras of the isometry algebra (and its dual).

The continuum form of the Iwasawa decomposed BF action \eqref{eq:2d_bf_action} is then given by,
\begin{equation}
    S_{\mathrm{BF}}[e,\Omega,\chi,\Phi] = \int_{M}\Big(\langle \Phi, d\Omega + \Omega \lhd e\rangle + \langle \chi, de+\Omega\rhd e+\frac{1}{2}\comm{e}{e}\rangle\Big), \label{eq:BF-Iwasawa-action}
\end{equation}
which by construction is exactly equal to the Cartan decomposed BF action \eqref{eq:cartan-BF-action}. By varying \eqref{eq:BF-Iwasawa-action} respectively with respect to $\Phi$ and $\chi_{a}$, we arrive at the Iwasawa decomposed versions of the Cartan flatness and torsion constraints \eqref{eq:2d_torsion_curvature},
\begin{subequations}{\label{eq:2d_iwasawa_torsion_curvature}}
\begin{align}
    \cR&:= d\Omega + \Omega\lhd e = d\Omega + \Omega \wedge e^{a}\epsilon_{a}{}^{b}n_{b}\approx 0,\\
   \mathcal{G}^{a}&:= de^{a}+(\Omega\rhd e)^{a}+\frac{1}{2}\comm{e}{e}^{a} = de^{a}+\Omega \wedge e^{b}\epsilon_{b}{}^{a}+\frac{n^{a}}{2}\epsilon_{bc}e^{b}\wedge e^{c}\approx 0.
\end{align}
\end{subequations}
The two curvature constraints are related by
\begin{equation}
    \cR = F+n_{a}\cG^{a}.\label{eq:Iwasawa-Cartan-curvature-relation}
\end{equation}
We see then that on-shell of $\cG^{a}\approx 0$ we have $\cR\approx F$. The gauge transformation rules for $\Omega, e^{a}, \Phi,$ and $\chi_{a}$ are easily derived from \eqref{eq:BF-component-gauge-transformations} and \eqref{eq:Iwaswa-Dressing-Lie-Alg} using $\lambda = \theta J +\rho^{a}P_{a}=\vartheta J+\rho^{a}P'_{a}$, $\vartheta := \theta + n^{a}\rho_{a}$, giving,
\begin{subequations}
\begin{align}
    \delta_{\lambda}\Omega &= \langle d\vartheta + \comm{\Omega}{\rho}+\comm{e}{\vartheta}, \tilde{J}'\rangle =  d\vartheta + \Omega \rho^{b}\epsilon_{b}{}^{a}n_{a}-\vartheta e^{a}\epsilon_{a}{}^{b}n_{b},\\
    \delta_{\lambda}e^{a} &= \langle d\rho + \comm{\Omega}{\rho}+\comm{e}{\vartheta}+\comm{e}{\rho}, \tilde{P}^{a}\rangle =d\rho^{a}+\Omega \rho^{b}\epsilon_{b}{}^{a}-\vartheta e^{b}\epsilon_{b}{}^{a}+n^{a}\epsilon_{cd}e^{c}\rho^{d},\\
    \delta_{\lambda}\Phi &= \langle \Phi\lhd \rho+\chi\lhd \rho, J\rangle =-\chi_{a}\rho^{b}\epsilon_{b}{}^{a}-\Phi n_{a}\rho^{b}\epsilon_{b}{}^{a},\\
    \delta_{\lambda}\chi_{a} &= \langle \Phi\lhd\vartheta+\chi\lhd\vartheta+\chi\lhd\rho , P'_{a}\rangle = \vartheta\epsilon_{a}{}^{b}\chi_{b}+\vartheta\epsilon_{a}{}^{b}\Phi n_{b}-n^{c}\chi_{c}\epsilon_{ab}\rho^{b}.
\end{align}    
\end{subequations}

Whenever the corresponding global group factorization exists, an element  $g\in G_{\sigma,\kappa}$ may be uniquely decomposed as,
\begin{equation}
    g = h\ell = \tilde{\ell}\tilde{h},\quad \tilde{h} = h\lhd \ell, \quad \tilde{\ell} = h\rhd \ell, \quad h,\tilde{h}\in H_{\sigma},\quad \ell,\tilde{\ell}\in AN_{\sigma,\kappa}.\label{eq:G-Iwasawa-decomp}
\end{equation}
Such $h,\tilde{h}, \ell,$ and $\tell$ are related to $\Omega$ and $e^{a}$ on-shell via the relations,
\begin{equation}{\label{eq:Omega-e-on-shell}}
   \cF(A)\approx 0 \implies A \approx g^{-1}dg \implies\begin{cases}
       \Omega  \approx (\ell^{-1}h^{-1}dh\ell)\rvert_{\mathfrak{h}_{\sigma}},\\
       e^{a} \approx (\ell^{-1}h^{-1}dh\ell)^{a}\rvert_{\mathfrak{an}_{\sigma,\kappa}}+(\ell^{-1}d\ell)^{a}.
   \end{cases}
\end{equation}
It is the discrete analogs of these $h, \th, \ell$ and $\tell$ that we wish to recover from \eqref{eq:ribbon-split-BF-path-integral} in our partial gauge fixing procedure.

\subsubsection{Partial gauge fixing towards Regge}

We use unbarred symbols $h_e^v\in H_\sigma$ and $\ell_e^\tau\in AN_{\sigma,\kappa}$ for the subgroup-reduced holonomies. Their face constraints are similarly denoted without bars,
\begin{subequations}{\label{eq:unbarred-ribbon-constraints}}
\begin{align}
\mathcal{W}_{e}&:= h^{s(e)}_{e}\ell^{\tau_{L}}_{e}(h^{t(e)}_{e})^{-1}(\ell^{\tau_{R}}_{e})^{-1}\in G_{\sigma,\kappa},\\
\cG_{\tau}&:= L_{1}L_{2}L_{3}\in AN_{\sigma,\kappa},\quad L_{i}:=(\ell^{\tau}_{e_{i}})^{\epsilon_{\tau}(e_{i})}\in AN_{\sigma,\kappa},\\
\cR_{v}&:= h_{1}\cdots h_{m}\in H_{\sigma},\quad h_{i}:=\begin{cases}
(h^{s(e_{i})}_{e_{i}})^{\pm1}\in H_{\sigma},\\
(h^{t(e_{i})}_{e_{i}})^{\pm1}\in H_{\sigma}.
\end{cases}
\end{align}
\end{subequations}
We call $\cR_v$ and $\cG_\tau$ the split BF curvature and torsion constraints. The ribbon constraint $\mathcal W_e\approx1$ gives the matched-pair relations \cite{chari_guide_1995}
\begin{equation}
   h^{s(e)}_{e}\ell^{\tau_{L}}_{e} = \ell^{\tau_{R}}_{e}h^{t(e)}_{e}
    \label{eq:ribbon-factorization},\quad \ell^{\tau_{R}}_{e} = h^{s(e)}_{e}\rhd \ell^{\tau_{L}}_{e},\quad h^{t(e)}_{e} = h^{s(e)}_{e}\lhd \ell^{\tau_{L}}_{e}.
\end{equation}
Thus, unlike in the Cartan decomposition, the ribbon face constraint with the Iwasawa decomposition gives independent relations for the subgroup holonomies on the sides of the ribbon. It is then the goal of our partial gauge fixing procedure to achieve,
\begin{equation}
    \bar{h}^{v}_{e}\mapsto h^{v}_{e},\quad \bar{\ell}^{\tau}_{e}\mapsto \ell^{\tau}_{e},\quad \bar{\mathcal{W}}_{e}\mapsto \mathcal{W}_{e},\quad \bar{\cG}_{\tau}\mapsto \cG_{\tau}, \quad \bar{\cR}_{v}\mapsto \cR_{v}.
\end{equation}
We detail this construction now.

\paragraph{Choose cuts and specify trees.} The first step is to fix cuts for the constraints and define trees for the gauge fixing procedure.

Write $\bar{\cG}_\tau=\bar L_1\bar L_2\bar L_3$ with cut $(\tau,v_1)$ and $\bar{\cR}_v=\bar h_1\cdots\bar h_m$ with cut $(\tau_1,v)$. Each factor in these products runs from corner $i$ to corner $i+1$, cyclically. Omitting the first factor defines a tree rooted at the cut (Figure \ref{fig:split-BF-cuts}). Our goal is to remove the $H_\sigma$ parts of $\bar L_2,\bar L_3$ and the $AN_{\sigma,\kappa}$ parts of $\bar h_2,\ldots,\bar h_m$ along these trees.
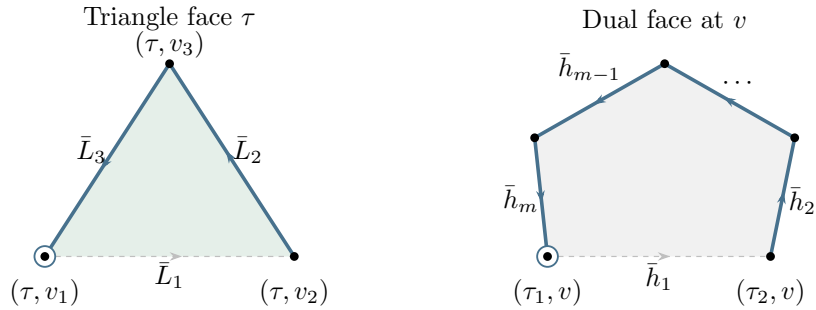
\begin{figure}[ht]
    \centering
    \begin{tikzpicture}[font=\small]
\path[fill=jtSage] (0.750000,1.000000) -- (4.050000,1.000000) -- (2.400000,3.550000) -- cycle;
\draw[jt faint,jt arrow] (0.750000,1.000000) -- (4.050000,1.000000);
\node[jt label] at (2.400000,0.750000) {$\bar L_1$};
\draw[jt tree,jt arrow] (4.050000,1.000000) -- (2.400000,3.550000);
\node[jt label] at (3.434893,2.410813) {$\bar L_2$};
\draw[jt tree,jt arrow] (2.400000,3.550000) -- (0.750000,1.000000);
\node[jt label] at (1.365107,2.410813) {$\bar L_3$};
\draw[draw=jtInkBlue,line width=.7pt,fill=white] (0.750000,1.000000) circle[radius=3.9pt];
\fill[black] (0.750000,1.000000) circle[radius=1.65pt];
\node[jt label] at (0.750000,0.520000) {$(\tau,v_1)$};
\fill[black] (4.050000,1.000000) circle[radius=1.65pt];
\node[jt label] at (4.050000,0.520000) {$(\tau,v_2)$};
\fill[black] (2.400000,3.550000) circle[radius=1.65pt];
\node[jt label] at (2.400000,3.800000) {$(\tau,v_3)$};
\node[jt label] at (2.400000,4.150000) {Triangle face $\tau$};
\path[fill=jtGrey] (7.400000,1.000000) -- (10.350000,1.000000) -- (10.670000,2.570000) -- (8.950000,3.550000) -- (7.230000,2.570000) -- cycle;
\draw[jt faint,jt arrow] (7.400000,1.000000) -- (10.350000,1.000000);
\node[jt label] at (8.875000,0.730000) {$\bar h_1$};
\draw[jt tree,jt arrow] (10.350000,1.000000) -- (10.670000,2.570000);
\node[jt label] at (10.774561,1.731077) {$\bar h_2$};
\draw[jt tree,jt arrow] (10.670000,2.570000) -- (8.950000,3.550000);
\node[jt label] at (9.943664,3.294593) {$\cdots$};
\draw[jt tree,jt arrow] (8.950000,3.550000) -- (7.230000,2.570000);
\node[jt label] at (7.956336,3.494593) {$\bar h_{m-1}$};
\draw[jt tree,jt arrow] (7.230000,2.570000) -- (7.400000,1.000000);
\node[jt label] at (7.046569,1.755934) {$\bar h_m$};
\draw[draw=jtInkBlue,line width=.7pt,fill=white] (7.400000,1.000000) circle[radius=3.9pt];
\fill[black] (7.400000,1.000000) circle[radius=1.65pt];
\fill[black] (10.350000,1.000000) circle[radius=1.65pt];
\fill[black] (10.670000,2.570000) circle[radius=1.65pt];
\fill[black] (8.950000,3.550000) circle[radius=1.65pt];
\fill[black] (7.230000,2.570000) circle[radius=1.65pt];
\node[jt label] at (7.400000,0.520000) {$(\tau_1,v)$};
\node[jt label] at (10.350000,0.520000) {$(\tau_2,v)$};
\node[jt label] at (8.950000,4.150000) {Dual face at $v$};
\end{tikzpicture}
    \caption{For clarity we illustrate the local trees defined by the products $\bar{\cG}_{\tau}$ and $\bar{\cR}_{v}$. The blue arrows denote the trees and the ordering along the tree upon which the gauge fixing is performed. The circled corner denotes the cut of the respective constraints, and the dotted lines denote the holonomy excluded from the trees.}
    \label{fig:split-BF-cuts}
\end{figure}

\paragraph{Determine the gauge fixing parameters.}

Consider the Iwasawa factorization of the gauge parameter $g_{\tau v}\in G_{\sigma,\kappa}$ at each corner $(\tau, v)$ as,
\begin{equation}
    g_{\tau v} = h_{\tau v}\ell_{\tau v},\quad h_{\tau v}\in H_{\sigma}, \quad \ell_{\tau v}\in AN_{\sigma,\kappa},
\end{equation}
such that $h_{\tau v_{1}} = 1$ at each triangle root, and $\ell_{\tau_{1} v} = 1$ at each vertex/dual face root. To achieve our pgf goal the gauge transformation rules \eqref{eq:framed-holonomy-gauge-transformations} tell us we must determine the appropriate gauge parameters $h_{\tau v_{i}}$ and $\ell_{\tau_{i} v}$ such that,
\begin{subequations}
    \begin{align}
        \bar{L}'_{i} &=g_{\tau v_i}^{-1}\, \bar{L}_{i}\,  g_{\tau v_{i+1}}=  \ell_{\tau v_{i}}^{-1}(h_{\tau v_{i}}^{-1}\,\bar{L}_{i}\, h_{\tau v_{i+1}})\ell_{\tau v_{i+1}}\in AN_{\sigma,\kappa},\quad i = 2, 3\\
        \bar{h}_{i}' &=g_{\tau_i v}^{-1}\, \bar{h}_{i}\,  g_{\tau_{i+1} v} = \ell_{\tau_{i}v}^{-1}(h_{\tau_{i}v}^{-1}\,\bar{h}_{i}\, h_{\tau_{i+1}v})\ell_{\tau_{i+1}v}\in H_{\sigma},\quad i = 2,\ldots, m.
    \end{align}
\end{subequations}
The necessary gauge parameters to achieve the required fixing are thus,
\begin{subequations}{\label{eq:split-gauge-recursion}}
    \begin{align}
         h_{\tau v_{i}} &= (\bar{L}_{i}h_{\tau v_{i+1}})\rvert_{H_{\sigma}},\quad i = 3, 2,\quad v_{4} = v_{1}\\
         \ell_{\tau_{i}v} &= (h_{\tau_{i}v}^{-1}\bar{h}_{i}h_{\tau_{i+1}v}\ell_{\tau_{i+1}v})\rvert_{AN_{\sigma,\kappa}},\quad i = m,\ldots, 2,\quad \tau_{m+1} = \tau_{1}.
    \end{align}
\end{subequations}
One can think of these choices as arising from the factorization of the $\bar{L}_{i}$ and $\bar{h}_{i}$ into their $AN_{\sigma,\kappa}$ and $H_{\sigma}$ components as,
\begin{subequations}
\begin{align}
    \bar{L}_{i} &= H_{i}L_{i} \implies (\bar{L}_{i})\rvert_{H_{\sigma}} = H_{i}, \quad (\bar{L}_{i})\rvert_{AN_{\sigma,\kappa}} = L_{i},\\
    \bar{h}_{i} &= \ell_{i}h_{i}\implies (\bar{h}_{i})\rvert_{H_{\sigma}} = h_{i},\quad (\bar{h}_{i})\rvert_{AN_{\sigma,\kappa}} = \ell_{i},
\end{align}
\end{subequations}
then using the gauge transformations \eqref{eq:framed-holonomy-gauge-transformations} to cancel the undesired $H_{i}$ and $\ell_{i}$ parts.

\paragraph{Inserting the gauge fixing factors.}

To determine and insert the appropriate gauge fixing factors into \eqref{eq:ribbon-split-BF-path-integral} we must first address an issue of integration: despite $G_{\sigma,\kappa}$ being unimodular, for $\Lambda\neq 0$ the subgroup $AN_{\sigma,\kappa}$ is \textit{not} unimodular. We thus use the general result that since $H_{\sigma}, AN_{\sigma,\kappa}\subset G_{\sigma,\kappa}$ and $H_{\sigma}\cap AN_{\sigma,\kappa} = \{1\}$, the unimodular Haar measure $dg$ on $G_{\sigma,\kappa}$ decomposes as \cite{knapp_lie_1996,dowd_notes_nodate,nachbin_haar_1976},
\begin{equation}
    dg = dhd_{R}\ell = d_{L}\tilde{\ell} d\th,\quad g = h\ell = \tilde{\ell}\th\in G_{\sigma,\kappa}, \quad h,\th\in H_{\sigma},\quad \ell,\tell\in AN_{\sigma,\kappa},\label{eq:path-int-measure-v1}
\end{equation}
where $dh$ is the unimodular Haar measure on $H_{\sigma}$ and $d_{L}$ and $d_{R}$ are respectively the left and right invariant Haar measures on $AN_{\sigma,\kappa}$. See Appendix \ref{app:measure-considerations} for useful integration relations in this context.

At each corner we integrate the unfixed parts of the gauge parameter with the measures listed in the table below.
\begin{center}
\begin{tabular}{@{}lll@{}}
\toprule
Corner $(\tau, v)$ is a root for & Gauge parameter & Parameter measure \\
\midrule
neither tree & $g_{\tau v}\in G_{\sigma,\kappa}$ & $dg_{\tau v}=dh_{\tau v}\,d_R\ell_{\tau v}$\\
tree in $\Delta$ & $\ell_{\tau v}\in AN_{\sigma,\kappa}$ & $d_L\ell_{\tau v}$\\
tree in $\Delta^{*}$ & $h_{\tau v}\in H_{\sigma}$ & $dh_{\tau v}$\\
both trees & $g_{\tau v} = 1\in G_{\sigma,\kappa}$ & no integration\\
\bottomrule
\end{tabular}
\captionof{table}{Gauge fixing parameter and integration scheme for the partial gauge fixing procedure.}
\end{center}
To implement this scheme in the path integral \eqref{eq:ribbon-split-BF-path-integral}, we insert the following gauge fixing factor,
\begin{equation}
    1 = \int \prod_{(\tau, v)}d\mu(g_{\tau v})\prod_{\tau}\prod_{i=2}^{3}\delta_{H_{\sigma}}\Big((\bar{L}_{i}')\rvert_{H_{\sigma}}\Big)\prod_{v}\prod_{i=2}^{m}\delta_{AN_{\sigma,\kappa}}\Big((\bar{h}_{i}')\rvert_{AN_{\sigma,\kappa}}\Big),\label{eq:to-regge-pgf-factor}
\end{equation}
where the respective delta functions above are those induced on $H_{\sigma}$ and $AN_{\sigma, \kappa}$ by the factorization,
\begin{equation}
\delta_{G_{\sigma,\kappa}}(g) = \delta_{H_{\sigma}}(h)\delta_{AN_{\sigma,\kappa}}(\ell) = \delta_{AN_{\sigma,\kappa}}(\tell)\delta_{H_{\sigma}}(\th),\quad g = h\ell = \tell \th,\label{eq:G-delta-factorization}
\end{equation}
and it is understood that in each case $d\mu(g_{\tau v})$ denotes the appropriate integration measure and gauge parameter specified in the above table. The insertion \eqref{eq:to-regge-pgf-factor} can be shown to have unit Faddeev-Popov determinant by using the relations detailed in Appendix \ref{app:measure-considerations} to account for the non-unimodularity of $AN_{\sigma,\kappa}$. Indeed, consider first a corner $(\tau_{i}, v_{j})$ not corresponding to a root of either tree, then \eqref{eq:to-regge-pgf-factor} contributes,
\begin{align}
    &\int dg_{\tau_{i} v_{j}}\delta_{H_{\sigma}}\Big((\bar{L}_{j}')\rvert_{H_{\sigma}}\Big)\delta_{AN_{\sigma,\kappa}}\Big((\bar{h}_{i}')\rvert_{AN_{\sigma,\kappa}}\Big)= \int dh_{\tau_{i} v_{j}}d_{R}\ell_{\tau_{i} v_{j}}\delta_{H_{\sigma}}\Big((\bar{L}_{j}')\rvert_{H_{\sigma}}\Big)\delta_{AN_{\sigma,\kappa}}\Big((\bar{h}_{i}')\rvert_{AN_{\sigma,\kappa}}\Big)\nonumber\\
    &=\int dh_{\tau_{i} v_{j}}d_{R}\ell_{\tau_{i} v_{j}}\delta_{H_{\sigma}}\Big(\ell_{\tau_{i} v_{j}}^{-1}\rhd\Big(h_{\tau_{i} v_{j}}^{-1}H_{j}(L_{j}\rhd h_{\tau_{i} v_{j+1}})\Big)\Big)\delta_{AN_{\sigma,\kappa}}\Big(\ell_{\tau_{i}v_{j}}^{-1}(h_{\tau_{i}v_{j}}^{-1}\bar{h}_{i}h_{\tau_{i+1}v_{j}}\ell_{\tau_{i+1}v_{j}})\rvert_{AN_{\sigma,\kappa}}\Big)\nonumber\\
    &=\int dh_{\tau_{i} v_{j}}d_{R}\ell_{\tau_{i} v_{j}}\Delta_{R}(\ell_{\tau_{i}v_{j}})^{-1}\delta_{H_{\sigma}}\Big(h_{\tau_{i} v_{j}}^{-1}H_{j}(L_{j}\rhd h_{\tau_{i} v_{j+1}})\Big)\delta_{AN_{\sigma,\kappa}}\Big(\ell_{\tau_{i}v_{j}}^{-1}(h_{\tau_{i}v_{j}}^{-1}\bar{h}_{i}h_{\tau_{i+1}v_{j}}\ell_{\tau_{i+1}v_{j}})\rvert_{AN_{\sigma,\kappa}}\Big)\nonumber\\
    &=\int dh_{\tau_{i} v_{j}}d_{L}\ell_{\tau_{i} v_{j}}\delta_{H_{\sigma}}\Big(h_{\tau_{i} v_{j}}^{-1}H_{j}(L_{j}\rhd h_{\tau_{i} v_{j+1}})\Big)\delta_{AN_{\sigma,\kappa}}\Big(\ell_{\tau_{i}v_{j}}^{-1}(h_{\tau_{i}v_{j}}^{-1}\bar{h}_{i}h_{\tau_{i+1}v_{j}}\ell_{\tau_{i+1}v_{j}})\rvert_{AN_{\sigma,\kappa}}\Big)\nonumber\\
    &=\int d_{L}\ell_{\tau_{i}v_{j}}\delta_{AN_{\sigma,\kappa}}\Big(\ell_{\tau_{i}v_{j}}^{-1}(u^{-1}\bar{h}_{i}h_{\tau_{i+1}v_{j}}\ell_{\tau_{i+1}v_{j}})\rvert_{AN_{\sigma,\kappa}}\Big) = 1,\quad u:=H_{j}(L_{j}\rhd h_{\tau_{i} v_{j+1}})
\end{align}
The contribution when $(\tau_{i},v_{j})$ instead corresponds to a root of just a $\Delta$ tree is by definition given by plugging $h_{\tau_{i}v_{j}} = 1$ into the first step of the above and omitting the $H_{\sigma}$ delta distribution. Similarly, the contribution when $(\tau_{i}, v_{j})$ corresponds to a root of a $\Delta^{*}$ tree is given by plugging $\ell_{\tau_{i}v_{j}} = 1$ into the first step of the above and omitting the $AN_{\sigma,\kappa}$ delta distribution. These respectively leave one with the contributions
\begin{subequations}
\begin{align}
    &\int d_{L}\ell_{\tau_{i} v_{j}}\delta_{AN_{\sigma,\kappa}}\Big(\ell_{\tau_{i}v_{j}}^{-1}(\bar{h}_{i}h_{\tau_{i+1}v_{j}}\ell_{\tau_{i+1}v_{j}})\rvert_{AN_{\sigma,\kappa}}\Big)  = 1\\
    &\int dh_{\tau_{i} v_{j}}\delta_{H_{\sigma}}\Big(h_{\tau_{i} v_{j}}^{-1}H_{j}(L_{j}\rhd h_{\tau_{i} v_{j+1}})\Big) = 1
\end{align}
\end{subequations}
The case of $(\tau_{i}, v_{j})$ corresponding to a root of both trees by definition provides an identity contribution since we take $g_{\tau_{i}v_{j}} = 1$ and do not integrate.

Inserting \eqref{eq:to-regge-pgf-factor} into \eqref{eq:ribbon-split-BF-path-integral}, gauge invariance allows an inverse gauge transformation to remove the parameters from the integrand. The primed holonomies become unprimed, leaving
\begin{equation}
    \prod_{\tau}\prod_{i=2}^{3}\delta_{H_{\sigma}}\Big(H_{i}\Big)\prod_{v}\prod_{i=2}^{m}\delta_{AN_{\sigma,\kappa}}\Big(\ell_{i}\Big).
\end{equation}
The parameter integrals now factor out as the gauge volumes already divided out by convention. Integrating the remaining delta functions gives
\begin{subequations}
    \begin{align}
       d\bar{L}_{i}\delta_{H_{\sigma}}(H_{i}) &=  dH_{i}d_{R}L_{i}\delta_{H_{\sigma}}(H_{i})\rightarrow d_{R}L_{i},\\
        d\bar{h}_{i}\delta_{AN_{\sigma,\kappa}}(\ell_{i}) &= d_{L}\ell_{i}dh_{i}\delta_{AN_{\sigma,\kappa}}(\ell_{i})\rightarrow  dh_{i},
    \end{align}
\end{subequations}
where the arrow above denotes performing the integration. This deals with the fixing and integrations along all non-root elements of the aforementioned trees.

To address the root elements, note that with $\bar{L}_{1} = H_{1}L_{1}$ and $\bar{h}_{1} = \ell_{1}h_{1}$ the partially gauged fixed constraints in \eqref{eq:ribbon-split-BF-path-integral} now take the form,
\begin{equation}
    \bar{\cG}_{\tau} \rightarrow H_{1}L_{1} L_{2}L_{3} = H_{1}\cG_{\tau},\quad \bar{\cR}_{v}\rightarrow \ell_{1}h_{1}h_{2}\cdots h_{m} = \ell_{1}\cR_{v}.
\end{equation}
Performing the integrations over the remaining $H_{1}$ and $\ell_{1}$ factors follows from \eqref{eq:G-delta-factorization} as
\begin{subequations}
    \begin{align}
        \int dH_{1} d_{R}L_{1} \Big(\delta_{G_{\sigma,\kappa}}(H_{1}L_{1} L_{2}L_{3})f(H,L)\Big) &= \int d_{R}L_{1}\Big(\delta_{AN_{\sigma,\kappa}}(L_{1} L_{2}L_{3})f(1,L)\Big) \nonumber\\
        &= \int d_{R}L_{1}\delta_{AN_{\sigma,\kappa}}(\cG_{\tau})f(1,L),\\
        \int d_{L}\ell_{1}dh_{1}\Big(\delta_{G_{\sigma,\kappa}}(\ell_{1}h_{1}h_{2}\cdots h_{m})f(\ell, h)\Big) &= \int dh_{1}\Big(\delta_{H_{\sigma}}(h_{1}h_{2}\cdots h_{m})f(1,h)\Big)\nonumber\\
        &=\int dh_{1}\delta_{H_{\sigma}}(\cR_{v})f(1,h),
    \end{align}
\end{subequations}
for any smooth test function $f$. Altogether the partial gauge fixing leaves us with subgroup face constraints $\delta_{AN_{\sigma,\kappa}}(\cG_{\tau})$ and $\delta_{H_{\sigma}}(\cR_{v})$, and hence the partially gauged fixed form of the path integral \eqref{eq:ribbon-split-BF-path-integral},
\begin{equation}
    Z^{\Delta}_{\mathrm{BF},\mathrm{rib}}\rvert_{\text{pgf}}= \int \prod_{e\in E(\Delta)}\Big(dh^{s(e)}_{e}d_{R}\ell^{\tau_{L}}_{e}d_{L}\ell^{\tau_{R}}_{e}dh^{t(e)}_{e}\Big)\prod_{e\in E(\Delta)}\delta_{G_{\sigma,\kappa}}(\mathcal{W}_{e})\prod_{\tau\in T(\Delta)}\delta_{AN_{\sigma,\kappa}}(\cG_{\tau})\prod_{v\in V(\Delta)}\delta_{H_{\sigma}}(\cR_{v}).
\end{equation}
The final step is thus to integrate over the ribbon flatness constraint. This can be simply done by noting that $d_{L}\ell^{\tau_{R}}_{e}dh^{t(e)}_{e} = d(\ell^{\tau_{R}}_{e}h^{t(e)}_{e})$, and using,
\begin{equation}
    \int d_L\ell_e^{\tau_R}dh_e^{t(e)}\delta_{G_{\sigma,\kappa}}\Big(h_e^{s(e)}\ell_e^{\tau_L}(h_e^{t(e)})^{-1}(\ell_e^{\tau_R})^{-1}\Big)f(\ell_e^{\tau_R},h_e^{t(e)})=f\Big(h_e^{s(e)}\rhd\ell_e^{\tau_L},h_e^{s(e)}\lhd\ell_e^{\tau_L}\Big).
\end{equation}
This last step gives the final partially gauge-fixed form for \eqref{eq:ribbon-split-BF-path-integral},
\begin{align}{\label{eq:BF-split-path-integral-v3}}
    Z^{\Delta}_{\mathrm{BF},\mathrm{rib}}\rvert_{\text{pgf}} &= \int \prod_{e\in E(\Delta)}dg_{e}\prod_{v\in V(\Delta)}\delta_{H_{\sigma}}(\cR_{v})\prod_{\tau\in T(\Delta)}\delta_{AN_{\sigma,\kappa}}(\cG_{\tau})\nonumber\\
    &=:Z^{\Delta}_{\mathrm{BF},\mathrm{split}},
\end{align}
where $dg_{e} := dh^{s(e)}_{e}d_{R}\ell^{\tau_{L}}_{e} = d_{L}\ell^{\tau_{R}}_{e} dh^{t(e)}_{e}$. 

We remark that the cilia (cuts) in \eqref{eq:BF-split-path-integral-v3} are inherited from the tree roots chosen in the partial gauge fixing of the cut independent ribbon BF path integral \eqref{eq:ribbon-split-BF-path-integral}. Since $AN_{\sigma,\kappa}$ is non-unimodular for $\Lambda\neq0$, cyclically changing the cut in $\cG_{\tau}$ generally introduces a modular factor in its delta distribution, e.g.,
\begin{equation}
    \delta_{AN_{\sigma,\kappa}}\Big(L_{1}^{-1}\cG_{\tau}L_{1}\Big) =\Delta_{R}(L_{1})\delta_{AN_{\sigma,\kappa}}(\cG_{\tau})
\end{equation}
Comparing different cilium choices therefore requires the corresponding change of partial gauge fixing at the level of \eqref{eq:ribbon-split-BF-path-integral}, accounting for the modular factors and dressing Jacobians described in Appendix \ref{app:measure-considerations}. We call $Z^\Delta_{\mathrm{BF},\mathrm{split}}$ the \textit{split BF path integral}. 

\paragraph{Split BF on the disk.} A special case of the split BF path integral \eqref{eq:BF-split-path-integral-v3} that we will need for matching with Regge is its restriction to a subcomplex triangulating a disk $D\subset M$. The restriction to a disk is to allow for $M$ to have a general topology while foregoing a proper treatment of holonomies around non-contractible loops. Extending the correspondence to non-trivial topology would require treating these global holonomies explicitly and is not addressed here.

Let $\Delta$ be an oriented triangulation of a disk $D\subset M$. We add one exterior face $\tau_{0}$ to $\Delta$, and denote the resulting cellulation of $S^{2}$ by $\bar{\Delta}$. Its vertices and edges are those of $\Delta$, and,
\begin{equation}
    T(\bar{\Delta}) = T(\Delta)\cup \{\tau_{0}\}.
\end{equation}
Every boundary edge now has a full ribbon \eqref{eq:ribbon-factorization}, and no weights are assigned to $\tau_{0}$. Choose some $v_{*}\in \partial D$ and cut $\cG_{\tau_{0}}$ at $v_{*}$ and cut $\cR_{v_{*}}$ similarly at $\tau_{0}$. These boundary constraints are redundant by the discrete Bianchi identity, thus we retain only the independent constraints,
\begin{equation}
    \cG_{\tau}\approx 1\quad \forall \tau\in T(\Delta),\quad \cR_{v}\approx 1\quad \forall v\in V(\Delta)\setminus\{v_{*}\}.
\end{equation}
The boundary completed disk formulation of the split BF path integral \eqref{eq:BF-split-path-integral-v3} is thus simply
\begin{equation}
    Z_{\mathrm{BF},\mathrm{split}}^{\Delta,D} = \int \prod_{e\in E(\Delta)}dg_{e}\prod_{\tau\in T(\Delta)}\delta_{AN_{\sigma,\kappa}}(\cG_{\tau})\prod_{v\neq v_{*}}\delta_{H_{\sigma}}(\cR_{v}).\label{eq:split-BF-path-integral-boundary-completion}
\end{equation}
The next section is dedicated to proving the equivalence of \eqref{eq:split-BF-path-integral-boundary-completion} with the Regge path integral \eqref{eq:disk-homo-regge-path-integral}.

\section{Recovering two-dimensional homogeneous Regge gravity from BF theory}\label{sec:recovering-homo-regge}

We now prove the main correspondence in the sector where all ingredients are under control. Let $\Delta$ triangulate a disk and take Euclidean signature $(\sigma=1)$ with $\kappa=0,-1$. We restrict to non-degenerate, consistently oriented, non-folded triangles and to a fixed orientation branch. Starting from  the split BF amplitude \eqref{eq:split-BF-path-integral-boundary-completion}, we first identify the $AN_{1,\kappa}$ triangle constraint with geodesic closure, then identify the $SO(2)$ curvature constraint with the Regge deficit angle modulo $2\pi$, and finally push the Haar measure forward to edge-length variables. The resulting BF amplitude carries a periodic deficit-angle constraint; equality with the real-delta LDG Regge  path integral \eqref{eq:disk-homo-regge-path-integral} follows after selecting the zero-winding sector.

Starting from the split BF path integral \eqref{eq:split-BF-path-integral-boundary-completion} on the disk, the argument explicitly proceeds as follows:
\begin{enumerate}
    \item We define a notion of intrinsic length for each long ribbon side and prove that the ribbon relation \eqref{eq:ribbon-factorization} implies both long ribbon side holonomies $\ell^{\tau_{L}}_{e},\ell^{\tau_{R}}_{e}$ are assigned the same intrinsic length. We then identify $\cG_{\tau}\approx 1$ as the condition for geodesic triangle closure. 
    \item We define a notion of \textit{edge-adapted frame} for every $\tau$, allowing for the manifest identification of the interior vertex curvature constraint $\cR_{v}\approx 1$ with vanishing deficit angle modulo $2\pi$.
    \item We implement the change of coordinates from split BF holonomies to Regge edge lengths, thereby identifying their path integral measures.
    \item Finally, as a consistency check, we show the Pachner move invariance of the Regge path integral.
\end{enumerate}
Results 1. and 2. establish equality of the split BF and Regge constraint surfaces, which combined with 3. gives equality of their path integrals. Result 4. is a quality check to ensure the Regge theory we derive is indeed topological.

\subsection{Triangle closure in split BF}\label{subsec:regge-closure-from-torsion}

The global Iwasawa factorization of $G_{1,\kappa}$ identifies the model space $X$ of \eqref{eq:X-as-quotient} with $AN_{1,\kappa}$. We use this identification to interpret $\cG_\tau\approx1$ as geodesic triangle closure.

The $G_{1,\kappa}$-invariant distance $d_{X}(\bullet,\bullet)$ defines the intrinsic length of $\ell\in AN_{1,\kappa}$ by
\begin{equation}
    l(\ell):=d_{X}(x_{0},\ell x_{0}),\label{eq:ell-intrinsic-length}
\end{equation}
where $x_0\in X$ is the base point fixed by $H_1$.\footnote{Different choices of base point give equivalent homogeneous space descriptions.}

\begin{lemma}[$H_{\sigma}$ dressing and inversion invariance of $l(\ell)$]\label{lem:dressing-and-inversion-invariance-of-distance-function}
    For every $h\in H_{1}$ and $\ell\in AN_{1,\kappa}$, the intrinsic lengths \eqref{eq:ell-intrinsic-length} are $H_{1}$ dressing invariant, and inversion invariant,
    \begin{equation}
        l(h\rhd \ell) = l(\ell),\quad l(\ell^{-1}) = l(\ell).
    \end{equation}
    Thus, if $e$ is an edge with ribbon data \eqref{eq:ribbon-factorization}, then $\ell^{\tau_{L}}_{e}$ and $\ell^{\tau_{R}}_{e}$ assign the same intrinsic length,
    \begin{equation}
        l(\ell^{\tau_{L}}_{e}) = l(\ell^{\tau_{R}}_{e}).
    \end{equation}
\end{lemma}
\begin{proof}
    Since by definition $h\in H_{\sigma}$ fixes $x_{0}$ (cf. \eqref{eq:X-as-quotient}), we have,
    \begin{equation}
        (h\rhd \ell)x_{0} = (h\rhd \ell)(h\lhd \ell)x_{0} = h\ell x_{0}.
    \end{equation}
    The distance function is $G_{1,\kappa}$-invariant, thus,
    \begin{equation}
        l(h\rhd\ell) = d_{X}(x_{0}, h\ell x_{0}) = d_{X}(h^{-1}x_{0},\ell x_{0}) = d_{X}(x_{0},\ell x_{0}) = l(\ell),
    \end{equation}
    as claimed. For inversion, we simply use the symmetry of $d_{X}(\bullet,\bullet)$ in its two arguments,
    \begin{equation}
        l(\ell) = d_{X}(x_{0},\ell x_{0}) = d_{X}(\ell^{-1}x_{0},x_{0}) = d_{X}(x_{0},\ell^{-1}x_{0}) = l(\ell^{-1}),
    \end{equation}
    as claimed. The final claim follows from the above and the ribbon relation \eqref{eq:ribbon-factorization},
    \begin{equation}
        l(\ell^{\tau_{R}}_{e}) = l(h^{s(e)}_{e}\rhd \ell^{\tau_{L}}_{e})= l(\ell^{\tau_{L}}_{e}),
    \end{equation}
    as claimed.
\end{proof}
Consider now a single triangle $\tau$ in the split BF theory, denoting its three oriented edges by $e_{i}$. We adopt the same shorthand as in \eqref{eq:unbarred-ribbon-constraints} to write its associated torsion constraint $\cG_{\tau}$ as,
\begin{equation}
    \cG_{\tau} = L_{1}L_{2}L_{3},\quad L_{i}:=(\ell^{\tau}_{e_{i}})^{\epsilon_{\tau}(e_{i})},\quad i = 1,2,3.
\end{equation}
Each $L_{i}$ factor determines the geodesic from $x_{0}$ to $L_{i}x_{0}$ of length $l(L_{i})$. We now prove that solutions to $\cG_{\tau}\approx 1$ are equivalent to closed oriented homogeneous triangles in $X$, just like the simplices in the homogeneous Regge discretization of two-dimensional gravity.

\begin{proposition}[Geodesic Triangle Closure]\label{prop:triangle-closure}
    On the sector in which the three resulting geodesics are non-degenerate,
    \begin{equation}
        \cG_{\tau} = L_{1}L_{2}L_{3}\approx 1,
    \end{equation}
    is equivalent to the closure of an oriented geodesic/homogeneous triangle in $X$ whose edge lengths are $l(L_{1}), l(L_{2}), l(L_{3})$.
\end{proposition}
\begin{proof}
    Suppose first that $L_{1}L_{2}L_{3}\approx 1$. Using the base point $x_{0}\in X$, define three other points $x_{1}, x_{2}, x_{3}\in X$ by,
    \begin{equation}
        x_{1}:= L_{1}x_{0},\quad x_{2}:= L_{1}L_{2}x_{0},\quad x_{3}:=L_{1}L_{2}L_{3}x_{0}.
    \end{equation}
    Since $L_{1}L_{2}L_{3}\approx 1$, we have $x_{3} \approx x_{0}$, thus $x_{0}\approx x_{3}, x_{1}, x_{2}$ close. Moreover, the lengths of each geodesic side of this configuration are indeed $l(L_{i})$,
    \begin{subequations}
    \begin{align}
        d_{X}(x_{0},x_{1}) &= d_{X}(x_{0}, L_{1}x_{0}) = l(L_{1}),\\
        d_{X}(x_{1}, x_{2}) &= d_{X}(L_{1}x_{0}, L_{1}L_{2}x_{0}) = d_{X}(x_{0}, L_{2}x_{0}) = l(L_{2}),\\
        d_{X}(x_{2}, x_{0}) &= d_{X}(L_{1}L_{2}x_{0}, x_{0}) = d_{X}(x_{0}, (L_{1}L_{2})^{-1}x_{0}) \approx d_{X}(x_{0}, L_{3}x_{0}) = l(L_{3}).
    \end{align}
     \end{subequations}
    Conversely, let $(x_{1}, x_{2}, x_{3})$ be a closed geodesic triangle in $X$. By the correspondence of $X$ and $AN_{1,\kappa}$, we may write $x_{i} = q_{i}x_{0}$ with unique $q_{i}\in AN_{1,\kappa}$. Then, defining,
    \begin{equation}
        L_{1} = q_{1}^{-1}q_{2},\quad L_{2} = q_{2}^{-1}q_{3}, \quad L_{3} = q_{3}^{-1}q_{1},
    \end{equation}
    we have,
    \begin{equation}
        L_{1}L_{2}L_{3} = q_{1}^{-1}q_{2}q_{2}^{-1}q_{3}q_{3}^{-1}q_{1} = 1,
    \end{equation}
    which is exactly the on-shell constraint $\cG_{\tau} = L_{1}L_{2}L_{3}\approx 1$. Moreover, the edge lengths are indeed,
    \begin{subequations}
    \begin{align}
        d_{X}(x_{1}, x_{2}) &= d_{X}(q_{1}x_{0}, q_{2}x_{0}) = d_{X}(x_{0}, q_{1}^{-1}q_{2}x_{0}) = l(L_{1}),\\
        d_{X}(x_{2}, x_{3}) &= d_{X}(q_{2}x_{0}, q_{3}x_{0}) = d_{X}(x_{0}, q_{2}^{-1}q_{3}x_{0}) = l(L_{2}),\\
        d_{X}(x_{3}, x_{1}) &= d_{X}(q_{3}x_{0}, q_{1}x_{0}) = d_{X}(x_{0}, q_{3}^{-1}q_{1}x_{0}) = l(L_{3}).
    \end{align}
     \end{subequations}
     Thus indeed $\cG_{\tau} \approx 1$ is equivalent to the closure of an oriented geodesic triangle in $X$, whose edge lengths are $l(L_{1}), l(L_{2}), l(L_{3})$, as claimed.
\end{proof}
Thus the nondegenerate solutions of $\cG_{\tau}\approx1$ give precisely the homogeneous triangles used in the Regge discretization.

\subsection{Edge-adapted frames and Regge deficit angle in split BF}\label{subsec:regge-deficit-from-curvature}

The next step is to prove that on-shell of all $\cG_{\tau}\approx 1$ constraints the discretized Iwasawa curvature constraint $\cR_{v}\approx 1$ is equivalent to the vanishing of the homogeneous Regge deficit angle (modulo $2\pi$). The key ingredient to this construction is to use the residual $H_{1}$ gauge freedom in \eqref{eq:split-BF-path-integral-boundary-completion} to fix \textit{edge-adapted frames} for each triangle. The split BF ribbon holonomies in this frame then have a manifest relation to the desired Regge variables.

Towards this goal, we represent the initial tangent of the geodesic from $x_{0}$ to $\ell x_{0}$, $\ell\neq 1 \in AN_{1,\kappa}$, in geodesic polar coordinates $(l(\ell), \varphi(\ell))$ on $X$. By the identification of $X$ with $AN_{1,\kappa}$, the geodesic polar line element and measure are
\begin{equation}\label{eq:AN-polar-measure}
    ds^{2} = dl^{2}+S_{\Lambda}(l)^{2}d\varphi^{2},\quad d_{L}\ell = S_{\Lambda}(l)dld\varphi
\end{equation}
Indeed, $AN_{1,\kappa}$ acts simply transitively by isometries, so its invariant area measure is a left Haar measure. Recalling for each physical triangle $\tau$ the cut of $\cG_{\tau} = L_{1}L_{2}L_{3}$ is denoted $(\tau, v_{1})$, define $\varphi_{\tau}$ to be the polar angle of the geodesic $x_{0}\rightarrow L_{1}x_{0}$ describing the angular direction of the first edge at the cut. Then, for every $\tau$, the residual $H_{1}$ gauge freedom in \eqref{eq:split-BF-path-integral-boundary-completion} allows one to align the first $\tau$ frame axis with the outgoing edge at the cut.

\begin{lemma}[edge-adapted frames]{\label{lem:edge-adapted-frames}}
    On $\cG_{\tau}\approx 1$, the residual $H_{1}$ gauge freedom may be fixed by 
    \begin{equation}\label{eq:edge-adapted-gauge-condition}
        \varphi_{\tau} = 0 \: (\mathrm{mod}\: 2\pi)
    \end{equation}
    This is the so-called \textit{edge-adapted frame} for $\tau$, aligning the first $\tau$ frame axis with the outgoing edge at the cut. The Faddeev-Popov determinant is one.

    Fixing one interior vertex $v$ with incident edges $e_{1},\ldots, e_{m}$ in cyclic order $\tau_{i}$ between $e_{i-1}$ and $e_{i}$ as in Figure \ref{fig:edge-adapted-frames}, the short side holonomies in the curvature constraint $\cR_{v} = h_{1}\cdots h_{m}$ have the form
    \begin{equation}
        h_{i} = \mathrm{exp}(-\beta^{\tau_{i}}_{v}J)
    \end{equation}
    when each $\tau_{i}$ is in its edge-adapted frame.
\end{lemma}
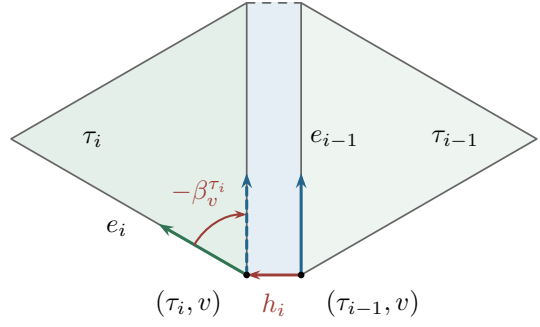
\begin{figure}[ht]
\centering
\begin{tikzpicture}[
  x=1cm,y=1cm,font=\small,line cap=round,line join=round,
  ef edge/.style={draw=black!60,line width=.65pt},
  ef axis/.style={line width=1.15pt,-{Stealth[length=2mm,width=1.4mm]}},
  every node/.style={inner sep=2pt}
]
  
  \coordinate (A) at (-.36,0);
  \coordinate (B) at (.36,0);
  \coordinate (C) at ({-.36-1.8*sqrt(3)},1.8);
  \coordinate (D) at (-.36,3.6);
  \coordinate (E) at (.36,3.6);
  \coordinate (F) at ({.36+1.8*sqrt(3)},1.8);
  \fill[jtSage] (A)--(D)--(C)--cycle;
  \fill[jtSage!65!white] (B)--(F)--(E)--cycle;
  \fill[jtBlue] (A)--(D)--(E)--(B)--cycle;
  \draw[ef edge] (A)--(C)--(D)--(A);
  \draw[ef edge] (B)--(F)--(E)--(B);
  \draw[ef edge,dashed] (D)--(E);

  \node at ({-.36-.6*sqrt(3)-1},1.8) {$\tau_i$};
  \node at ({.36+.6*sqrt(3)+1},1.8) {$\tau_{i-1}$};
  \node at ($(A)!.5!(C)+(-.16,-.2771)$) {$e_i$};
  \node at (0.8,1.8) {$e_{i-1}$};

  \draw[ef axis,draw=ribbonblue] (B)--++(90:1.35);
  \draw[ef axis,draw=facegreen] (A)--++(150:1.35);

  \draw[ef axis,draw=ribbonblue,dash pattern=on 3pt off 2pt]
    (A)--++(90:1.35);
  \draw[draw=cutred,line width=.8pt,
    -{Stealth[length=1.6mm,width=1.1mm]}]
    ($(A)+(150:.80)$) arc[start angle=150,end angle=90,radius=.80];
  \node[text=cutred] at ($(A)+(120:1.25)$) {$-\beta_v^{\tau_i}$};

  \draw[draw=cutred,line width=1.1pt,
    -{Stealth[length=1.8mm,width=1.2mm]}] (B)--(A);
  \node[text=cutred,below] at (0,-.17) {$h_{i}$};
  \fill (A) circle (1.25pt);
  \fill (B) circle (1.25pt);
  \node[below left] at ($(A)+(-.25,-.12)$) {$(\tau_i,v)$};
  \node[below right] at ($(B)+(.25,-.12)$) {$(\tau_{i-1},v)$};
\end{tikzpicture}
\caption{Edge-adapted frames at the corners $(\tau_i,v)$ and $(\tau_{i-1},v)$. The solid arrows are the first frame axes, along $e_i$ and $e_{i-1}$ respectively. The ribbon relation identifies the shared-edge tangent (dashed blue arrow), whose angle in the $\tau_i$ frame is $-\beta_v^{\tau_i}$. Thus the short-side rotation is $\exp(-\beta_v^{\tau_i}J)$.}
\label{fig:edge-adapted-frames}
\end{figure}
\begin{proof}
   Consider first a single $\tau$. Denote the $H_{1}$ gauge parameter at the root by $h_{\tau v_{1}} = e^{\theta_{\tau}J}$, and propagate it to the other corners by the same recursion as \eqref{eq:split-gauge-recursion}. Then, on triangle closure,
    \begin{equation}\label{eq:edge-adapted-residual-transform}
        L'_{i} = h_{\tau v_{i}}^{-1}L_{i}h_{\tau v_{i+1}}\in AN_{1,\kappa},\quad h_{\tau v_{i+1}} = (h_{\tau v_{i}}^{-1}\lhd L_{i})^{-1}, \quad v_{4} = v_{1}
    \end{equation}
    This recursion closes because $\cG_{\tau} \approx 1$. Since $H_{1}$ fixes $x_{0}$,
    \begin{equation}
        L'_{1}x_{0} = h_{\tau v_{1}}^{-1}L_{1}x_{0},\quad L'_{1}L'_{2}x_{0} = h_{\tau v_{1}}^{-1}L_{1}L_{2}x_{0} \implies \varphi'_{\tau} = \varphi_{\tau} - \theta_{\tau}
    \end{equation}
    Thus $\varphi_{\tau} = 0$ (mod $2\pi$) may be imposed by the insertion,
    \begin{equation}
        1 = \int_{0}^{2\pi}d\theta_{\tau}\delta_{2\pi}(\varphi_{\tau}-\theta_{\tau})
    \end{equation}
    giving the first claim.

    The second claim follows from the fact that the shared edge $e_{i-1}$ of $\tau_{i}$ and $\tau_{i-1}$ has coordinates $(1,0)$ in the $\tau_{i-1}$ frame and $(\cos\beta^{\tau_{i}}_{v},-\sin\beta^{\tau_{i}}_{v})$ in the $\tau_{i}$ frame by $\cG_{\tau_{i}}\approx 1$. The ribbon relation \eqref{eq:ribbon-factorization} identifies these tangents, so we must have
     \begin{equation}
        h_{i} = \mathrm{exp}(-\beta^{\tau_{i}}_{v}J)
    \end{equation}
    as claimed.
\end{proof}
Note that the condition \eqref{eq:edge-adapted-gauge-condition} does not align every outgoing edge at every corner. The three corner rotations in \eqref{eq:edge-adapted-residual-transform} are cyclically related on $\cG_{\tau}\approx 1$, leaving only the one independent $H_{1}$ action at the cut.

The relation between $\cR_{v}\approx 1$ and the Regge deficit angle is manifest in edge-adapted frames.

\begin{proposition}[Regge deficit angle recovery]\label{prop:Regge-deficit-recovery}
    For $v\in V^0(\Delta)$, on the nondegenerate sector of the incident $\cG_{\tau_i}\approx1$ with consistently oriented, non-folded triangles, the curvature constraint is given by the homogeneous Regge deficit angle,
    \begin{equation}
        \cR_{v}\overset{\cG_{\tau_{i}}\approx 1}{=} \mathrm{exp}(\delta_{v}(l;\Lambda)J).
    \end{equation}
    Consequently,
    \begin{equation}
        \cR_{v}\approx 1 \iff \delta_{v}(l;\Lambda) = 0 \quad \text{(mod $2\pi$)}.
    \end{equation}
\end{proposition}
\begin{proof}
    By choosing edge-adapted frames for every $\tau_{i}$ incident on $v$, Lemma \ref{lem:edge-adapted-frames} says
    \begin{equation}
        \cR_{v} \overset{\cG_{\tau_{i}}\approx 1}{=} \prod_{\tau\ni v}\mathrm{exp}(-\beta^{\tau}_{v} J) = \mathrm{exp}(-\sum_{\tau\ni v}\beta^{\tau}_{v} J).
    \end{equation}
    Since $\mathrm{exp}(2\pi J) = 1$, we have,
    \begin{equation}
      \cR_{v} \overset{\cG_{\tau_{i}}\approx 1}{=}  \mathrm{exp}\Big((2\pi-\sum_{\tau\ni v}\beta^{\tau}_{v}) J\Big) = \mathrm{exp}(\delta_{v}(l;\Lambda) J).
    \end{equation}
    It is then immediate that $\cR_{v}\approx 1 \iff \delta_{v}(l;\Lambda) = 0 \:\text{(mod $2\pi$)} $, as claimed.
\end{proof}
These results give the following distributional equivalence between the constraint surfaces of split BF and homogeneous Regge,
\begin{equation}
     \delta_{H_{1}}(\cR_{v})\prod_{\tau\ni v}\delta_{AN_{1,\kappa}}(\cG_{\tau}) = \delta_{2\pi}(\delta_{v}(l;\Lambda))\prod_{\tau\ni v}\delta_{AN_{1,\kappa}}(\cG_{\tau}).
\end{equation}
To write down the edge frame fixed version of \eqref{eq:split-BF-path-integral-boundary-completion} we temporarily restore the ribbon face constraint therein,
\begin{equation}
    dg_{e} \leftrightarrow dh^{s(e)}_{e}d_{R}\ell^{\tau_{L}}_{e}d_{L}\ell^{\tau_{R}}_{e}dh^{t(e)}_{e}\delta_{G_{1,\kappa}}(\mathcal{W}_{e})
\end{equation}
For fixed long sides, the ribbon constraint requires
\begin{equation}
    h^{s(e)}_{e}\rhd \ell^{\tau_{L}}_{e} = \ell^{\tau_{R}}_{e},\quad h^{t(e)}_{e} = (\ell^{\tau_{R}}_{e})^{-1}h^{s(e)}_{e}\ell^{\tau_{L}}_{e}\in H_{1}
\end{equation}
So, writing $l^{\tau_{L}}_{e}:=l(\ell^{\tau_{L}}_{e})$ and $l^{\tau_{R}}_{e}:=l(\ell^{\tau_{R}}_{e})$, we have,
\begin{equation}\label{eq:short-ribbon-length-matching}
    \int dh^{s(e)}_{e}dh^{t(e)}_{e}\delta_{G_{1,\kappa}}(\mathcal{W}_{e}) = \frac{\delta_{\mathbb{R}}(l^{\tau_{L}}_{e}-l^{\tau_{R}}_{e})}{S_{\Lambda}(\ell^{\tau_{L}}_{e})}
\end{equation}
The short side holonomy integrations are removed  from \eqref{eq:split-BF-path-integral-boundary-completion} in this way, giving the edge frame fixed expression 
\begin{align}{\label{eq:split-BF-path-integral-edge-adapted}}
    Z_{\mathrm{BF},\mathrm{split},\mathrm{red}}^{\Delta,D}= \int \prod_{\tau\in T(\Delta)}&\Bigg[\prod_{i=1}^{3}d_{R}L_{i}\delta_{AN_{1,\kappa}}(L_{1}L_{2}L_{3})\delta_{2\pi}(\varphi_{\tau}(L_{1}))\Bigg]\nonumber\\
    &\times\prod_{\substack{e\in E(\Delta)\\ e\not\subset\partial D}}\frac{\delta_{\mathbb{R}}(l^{\tau_{L}}_{e}-l^{\tau_{R}}_{e})}{S_{\Lambda}(\ell^{\tau_{L}}_{e})}\prod_{v\in V^{0}(\Delta)}\delta_{2\pi}(\delta_{v}(l;\Lambda)).
\end{align}

\subsection{Identification of the measures and recovery of Regge}\label{subsubsec:BF-to-homogeneous-space}

The previous sections established equivalence of the split BF and Regge constraint surfaces. It remains to evaluate \eqref{eq:split-BF-path-integral-edge-adapted} in Regge length coordinates to recover \eqref{eq:disk-homo-regge-path-integral} and subsequently use this correspondence to prove Pachner move invariance of the Regge path integral.

\subsubsection{Recovery of Regge from split BF}

We start from \eqref{eq:split-BF-path-integral-edge-adapted}. Our first claim addresses the push forward of $d_{R}L_{i}$ to Regge coordinates.

\begin{lemma}[Split BF to edge length measure]\label{lem:split-bf-to-edge-measure}
    For $\tau=(v_1v_2v_3)$ in the selected nondegenerate orientation branch,
    \begin{equation}\label{eq:edge-adapted-triangle-length-measure}
        \prod_{i=1}^3d_RL_i\,
        \delta_{AN_{1,\kappa}}(L_1L_2L_3)\delta_{2\pi}(\varphi_\tau)\longrightarrow
        \frac{S_\Lambda(l_{12})dl_{12}\,S_\Lambda(l_{13})dl_{13}\,
              S_\Lambda(l_{23})dl_{23}}{V_{\tau,\Lambda}}.
    \end{equation}
    Applying the above for a shared edge along with \eqref{eq:short-ribbon-length-matching} then gives,
\begin{equation}\label{eq:edge-adapted-shared-edge-cancellation}
    S_{\Lambda}(l^{\tau_{L}}_{e})dl^{\tau_{L}}_{e}S_{\Lambda}(l^{\tau_{R}}_{e})dl^{\tau_{R}}_{e}\frac{\delta_{\mathbb{R}}(l^{\tau_{L}}_{e}-l^{\tau_{R}}_{e})}{S_{\Lambda}(l^{\tau_{L}}_{e})} = dl_e^{\tau_L}\,
       S_\Lambda(l_e^{\tau_R})dl_e^{\tau_R}\,
       \delta_{\mathbb R}(l_e^{\tau_L}-l_e^{\tau_R})
       \longrightarrow S_\Lambda(l_e)dl_e.
\end{equation}
The arrow denotes integration of the constraints and the induced change of coordinates, with all remaining factors evaluated on their support.
\end{lemma}
\begin{proof}
    Use the developed vertices of Proposition \ref{prop:triangle-closure}, with the root at $x_0$
    \begin{equation}
        q_1=1,\quad q_2=L_1,\quad q_3=L_1L_2,\quad
        l_{ij}=d_X(x_0,q_{i}^{-1}q_jx_0).
    \end{equation}
    In the $q_{i}$ variables,
    \begin{equation}
        \prod_{i=1}^{3}d_{R}L_{i}\delta_{AN_{1,\kappa}}(L_{1}L_{2}L_{3})\longrightarrow d_{R}(q_{2}^{-1}q_{3})d_{R}(q_{3}^{-1}) = d_{L}q_{2}d_{L}q_{3}
    \end{equation}
    Now, the polar angle of $q_2x_0$ is $\varphi_\tau$, while that of $q_3x_0$ differs by the internal angle $\beta_{v_1}^\tau$, with sign fixed by the chosen orientation. Therefore
    \begin{equation}
        d_Lq_2d_Lq_3
        =S_\Lambda(l_{12})dl_{12}\,S_\Lambda(l_{13})dl_{13}\,
          d\varphi_\tau\,d\beta_{v_1}^\tau.
    \end{equation}
    The edge frame delta removes $d\varphi_\tau$. Differentiating the cosine law \eqref{eq:curved-cosine-law}, or \eqref{eq:flat-cosine-law} for $\Lambda=0$, at fixed $l_{12},l_{13}$ gives
    \begin{equation}\label{eq:edge-adapted-cosine-Jacobian}
        \frac{\partial\beta_{v_1}^\tau}{\partial l_{23}}
        =\frac{S_\Lambda(l_{23})}
          {S_\Lambda(l_{12})S_\Lambda(l_{13})\sin\beta_{v_1}^\tau}
        =\frac{S_\Lambda(l_{23})}{V_{\tau,\Lambda}},
    \end{equation}
    where $V_{\tau,\Lambda}$ is the Jacobian factor \eqref{eq:triangle-volume-gram-matrix}. This yields \eqref{eq:edge-adapted-triangle-length-measure} as claimed. The opposite triangle orientation would supply a separate branch and is not included in this measure.

    For a shared edge $e$ the above supplies one length measure from each incident triangle, whereas \eqref{eq:short-ribbon-length-matching} supplies a matching delta and an inverse radial factor, therefore,
    \begin{equation}
        S_{\Lambda}(l^{\tau_{L}}_{e})dl^{\tau_{L}}_{e}S_{\Lambda}(l^{\tau_{R}}_{e})dl^{\tau_{R}}_{e}\frac{\delta_{\mathbb{R}}(l^{\tau_{L}}_{e}-l^{\tau_{R}}_{e})}{S_{\Lambda}(l^{\tau_{L}}_{e})} = dl_e^{\tau_L}\,
       S_\Lambda(l_e^{\tau_R})dl_e^{\tau_R}\,
       \delta_{\mathbb R}(l_e^{\tau_L}-l_e^{\tau_R})
       \longrightarrow S_\Lambda(l_e)dl_e.
    \end{equation}
    as claimed. As expected the two incidence lengths become the same $l_{e}$.
\end{proof}
We now arrive at the final Regge recovery theorem.

\begin{theorem}[Regge from split BF]\label{thm:regge-recovery}
Let $\sigma = 1, \kappa = 0,-1$, and let $\Delta$ triangulate the disk. In the consistently oriented non-degenerate triangle sector, the edge adapted split BF path integral \eqref{eq:split-BF-path-integral-edge-adapted} is,
\begin{align}\label{eq:split-BF-periodic-length-integral}
Z_{\mathrm{BF},\mathrm{split},\mathrm{red}}^{\Delta,D}[\mathcal O]
        =\int_{\mathcal D_\Delta}\Bigg[\prod_{e\in E(\Delta)}S_{\Lambda}(l_{e})dl_{e}\Bigg]\Bigg[\prod_{\tau\in T(\Delta)}\frac{1}{V_{\tau,\Lambda}}\Bigg]\Bigg[\prod_{v\in V^{0}(\Delta)}\delta_{2\pi}(\delta_{v}(l;\Lambda))\Bigg]\mathcal{O}(l)
\end{align}
   As we discussed in \S\ref{subsec:regge-path-int-discretization}, we note that the delta function on $\delta_v$ allows for positive or negative deficit angles, multiple of $2\pi$, which corresponds to conical defects or branched geometries. To recover the projection on the regular case $\delta_v=0$ would amount to restrict appropriately the geometric data, for example such that  $\abs{\delta_{v}(l;\Lambda)}<\pi$ at every interior vertex,
\begin{equation}\label{eq:split-BF-Regge-real-equality}
    Z_{\mathrm{BF},\mathrm{split},\mathrm{red}}^{\Delta,D}[\mathcal O]\Big\rvert_{\abs{\delta_{v}}<\pi} = Z_{\mathrm{LDG}}^{\Delta, D}[\mathcal{O}]
\end{equation}
    where $Z_{\mathrm{LDG}}^{\Delta, D}[\mathcal{O}]$ is given by \eqref{eq:disk-homo-regge-path-integral}.
\end{theorem}
\begin{proof}
    For the first claim apply \eqref{eq:edge-adapted-triangle-length-measure} to every triangle bracket in \eqref{eq:split-BF-path-integral-edge-adapted}. Then apply \eqref{eq:edge-adapted-shared-edge-cancellation} to each interior edge. Exactly one length measure per edge and one inverse Jacobian per triangle remain, proving \eqref{eq:split-BF-periodic-length-integral}.

    For the second claim we use the generic relation
    \begin{equation}\label{eq:2pi-delta-R-delta}
        \delta_{2\pi}(\theta) = \sum_{w\in \mathbb{Z}}\delta_{\mathbb{R}}(\theta-2\pi w)
    \end{equation}
    On $-\pi<\delta_v<\pi$, the only contribution from the right hand side is when $w=0$, so distributionally $\delta_{2\pi}(\delta_v)=\delta_{\mathbb R}(\delta_v)$ in this restricted domain. Combined with the first claim, this gives the LDG integrand of \eqref{eq:disk-homo-regge-path-integral}. Its domain may be extended back to $\mathcal D_\Delta$, since the real delta distributions already restrict it to $\delta_v=0$ at every interior vertex. This proves \eqref{eq:split-BF-Regge-real-equality}.
\end{proof}

\subsubsection{Regge Pachner moves from split BF}

The length measure of Theorem \ref{thm:regge-recovery} also gives local Pachner invariance of the Regge theory. For the calculations to follow, we represent each developed vertex $i$ by $q_i x_0$, with $q_i\in AN_{1,\kappa}$ as in Proposition \ref{prop:triangle-closure}. With vertices $i,j$ fixed, Lemma \ref{lem:split-bf-to-edge-measure} gives
\begin{equation}\label{eq:regge-pachner-point-measure}
 d_Lq_k=
 \frac{S_\Lambda(l_{ik})\,dl_{ik}\,S_\Lambda(l_{jk})\,dl_{jk}}
 {V_{(ijk),\Lambda}}.
\end{equation}
This formula holds on the chosen side of the geodesic $(ij)$, whereas including both sides would count two orientation branches.

\paragraph{The $2\leftrightarrow2$ Pachner move.}
Let $(1234)$ be a convex geodesic quadrilateral. We consider inserting an edge along the diagonal $(13)$, or along the diagonal $(24)$, such that the resulting configurations are two geodesic triangles sharing the inserted edge (see Figure \ref{fig:regge-2-to-2}).
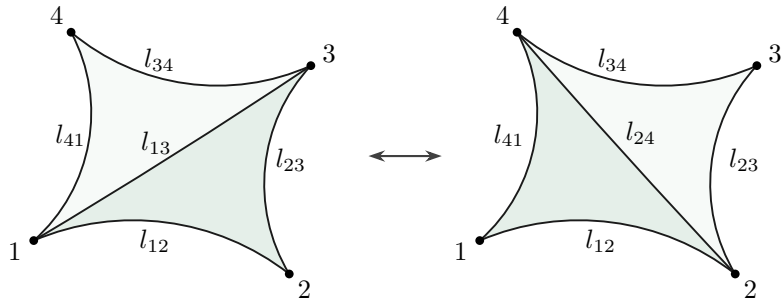
\begin{figure}[ht]
    \centering
   \begin{tikzpicture}[font=\small]
\path[fill=jtSage] (0.352000,0.782000) arc[start angle=113.08378621,end angle=52.01571565,radius=3.35476507cm] arc[start angle=-148.89404816,end angle=-222.95511581,radius=2.30040963cm] arc[start angle=-56.27223260,end angle=-59.20706098,radius=84.64450499cm] -- cycle;
\path[fill=jtSage!30!white] (0.352000,0.782000) arc[start angle=-59.20706098,end angle=-56.27223260,radius=84.64450499cm] arc[start angle=-65.79658485,end angle=-130.06893178,radius=3.01038403cm] arc[start angle=28.12680859,end angle=-48.45097497,radius=2.25936261cm] -- cycle;
\draw[jt edge] (0.352000,0.782000) arc[start angle=113.08378621,end angle=52.01571565,radius=3.35476507cm];
\draw[jt edge] (3.732000,0.340000) arc[start angle=-148.89404816,end angle=-222.95511581,radius=2.30040963cm];
\draw[jt edge] (4.018000,3.096000) arc[start angle=-65.79658485,end angle=-130.06893178,radius=3.01038403cm];
\draw[jt edge] (0.846000,3.538000) arc[start angle=28.12680859,end angle=-48.45097497,radius=2.25936261cm];
\draw[jt edge] (0.352000,0.782000) arc[start angle=-59.20706098,end angle=-56.27223260,radius=84.64450499cm];
\node[jt label] at (1.975635,2.070691) {$l_{13}$};
\node[jt label] at (0.84032,2.175033) {$l_{41}$};
\node[jt label] at (3.754032,1.848222) {$l_{23}$};
\node[jt label] at (2.014032,3.173670) {$l_{34}$};
\node[jt label] at (1.954032,0.78118) {$l_{12}$};
\fill[black] (0.352000,0.782000) circle[radius=1.65pt];
\node[jt label] at (0.102127,0.634815) {$1$};
\fill[black] (3.732000,0.340000) circle[radius=1.65pt];
\node[jt label] at (3.931737,0.129750) {$2$};
\fill[black] (4.018000,3.096000) circle[radius=1.65pt];
\node[jt label] at (4.258202,3.258490) {$3$};
\fill[black] (0.846000,3.538000) circle[radius=1.65pt];
\node[jt label] at (0.657271,3.758184) {$4$};
\path[fill=jtSage] (6.252000,0.782000) arc[start angle=113.08378621,end angle=52.01571565,radius=3.35476507cm] arc[start angle=-136.74289178,end angle=-139.12845511,radius=103.46836239cm] arc[start angle=28.12680859,end angle=-48.45097497,radius=2.25936261cm] -- cycle;
\path[fill=jtSage!30!white] (9.632000,0.340000) arc[start angle=-148.89404816,end angle=-222.95511581,radius=2.30040963cm] arc[start angle=-65.79658485,end angle=-130.06893178,radius=3.01038403cm] arc[start angle=-139.12845511,end angle=-136.74289178,radius=103.46836239cm] -- cycle;
\draw[jt edge] (6.252000,0.782000) arc[start angle=113.08378621,end angle=52.01571565,radius=3.35476507cm];
\draw[jt edge] (9.632000,0.340000) arc[start angle=-148.89404816,end angle=-222.95511581,radius=2.30040963cm];
\draw[jt edge] (9.918000,3.096000) arc[start angle=-65.79658485,end angle=-130.06893178,radius=3.01038403cm];
\draw[jt edge] (6.746000,3.538000) arc[start angle=28.12680859,end angle=-48.45097497,radius=2.25936261cm];
\draw[jt edge] (9.632000,0.340000) arc[start angle=-136.74289178,end angle=-139.12845511,radius=103.46836239cm];
\node[jt label] at (8.384161,2.205364) {$l_{24}$};
\node[jt label] at (6.64032,2.175033) {$l_{41}$};
\node[jt label] at (9.754032,1.848222) {$l_{23}$};
\node[jt label] at (8.014032,3.173670) {$l_{34}$};
\node[jt label] at (7.854032,0.78118) {$l_{12}$};
\fill[black] (6.252000,0.782000) circle[radius=1.65pt];
\node[jt label] at (6.002127,0.634815) {$1$};
\fill[black] (9.632000,0.340000) circle[radius=1.65pt];
\node[jt label] at (9.831737,0.129750) {$2$};
\fill[black] (9.918000,3.096000) circle[radius=1.65pt];
\node[jt label] at (10.158202,3.258490) {$3$};
\fill[black] (6.746000,3.538000) circle[radius=1.65pt];
\node[jt label] at (6.557271,3.758184) {$4$};
\draw[jt move,{Stealth[length=2.2mm,width=1.5mm]}-{Stealth[length=2.2mm,width=1.5mm]}] (4.780000,1.900000) -- (5.750000,1.900000);
\end{tikzpicture}
    \caption{Depicting the hyperbolic case here, the $2\leftrightarrow2$ Pachner move in the Regge discretization shows that the contribution to the path integral from a quadrilateral $(1234)$ with edge inserted along the $(13)$ diagonal is equivalent to the contribution instead with an edge inserted along the $(24)$ diagonal on a common convex sector with boundary data $l_{12}, l_{23}, l_{34}, l_{41}$.}
    \label{fig:regge-2-to-2}
\end{figure}
On the split BF side both cases have measure
\begin{equation}
 \frac1{2\pi}\,d_Lq_2\,d_Lq_3\,d_Lq_4
\end{equation}
and its push forward to length variables is obtained from \eqref{eq:regge-pachner-point-measure} in either case of edge insertion. The resulting Regge amplitudes are thus respectively
\begin{align}
 &\left[\prod_{e\subset\partial(1234)}S_\Lambda(l_e)\,dl_e\right]
 \frac{S_\Lambda(l_{13})\,dl_{13}}
 {V_{(123),\Lambda}V_{(134),\Lambda}},\\
 &\left[\prod_{e\subset\partial(1234)}S_\Lambda(l_e)\,dl_e\right]
 \frac{S_\Lambda(l_{24})\,dl_{24}}
 {V_{(124),\Lambda}V_{(234),\Lambda}}.
\end{align}
They are two coordinate descriptions of the same split BF measure. At fixed boundary lengths, canceling the common boundary factors gives
\begin{equation}\label{eq:regge-2-to-2-measure}
 \frac{S_\Lambda(l_{13})\,|dl_{13}|}
 {V_{(123),\Lambda}V_{(134),\Lambda}}
 =
 \frac{S_\Lambda(l_{24})\,|dl_{24}|}
 {V_{(124),\Lambda}V_{(234),\Lambda}}.
\end{equation}
Here the diagonals are related by the geometry of the same quadrilateral. Integrating over the appropriate length domain proves the $2\leftrightarrow2$ Pachner move invariance. The two angles meeting at each boundary vertex sum to the same quadrilateral angle in both cases of edge insertion, so the surrounding deficit angle constraints are also preserved when this patch is inserted into a larger triangulation.

\paragraph{The $3\leftrightarrow1$ Pachner move.}
Let $\tau_0=(123)$ and insert a vertex $0$ in its interior. As in Figure \ref{fig:regge-3-to-1}, denote the three resulting triangles by
\begin{equation}
 \tau_i:=(0123)\setminus\{i\},\qquad i=1,2,3.
\end{equation}
\begin{figure}[ht]
\centering
\begin{tikzpicture}[font=\small]
\path[fill=jtSage] (2.353000,1.541000) arc[start angle=59.06602341,end angle=55.05659551,radius=32.04352037cm] arc[start angle=-128.64338713,end angle=-173.46876607,radius=5.24199660cm] arc[start angle=-179.54245692,end angle=-177.79310338,radius=74.66613073cm] -- cycle;
\node[jt label] at (2.774252,1.828829) {$\tau_1$};
\node[jt label] at (3.34252,1.835446){$l_{23}$};
\path[fill=jtSage!30!white] (0.365500,0.322000) arc[start angle=124.67683716,end angle=118.36751733,radius=21.18380390cm] arc[start angle=-177.79310338,end angle=-179.54245692,radius=74.66613073cm] arc[start angle=-6.53123393,end angle=-51.35661287,radius=5.24199660cm] -- cycle;
\node[jt label] at (1.981596,1.835446) {$\tau_2$};
\node[jt label] at (1.3,1.835446) {$l_{13}$};
\path[fill=jtSage!60!white] (0.365500,0.322000) arc[start angle=121.02045133,end angle=58.97954867,radius=3.75380327cm] arc[start angle=55.05659551,end angle=59.06602341,radius=32.04352037cm] arc[start angle=118.36751733,end angle=124.67683716,radius=21.18380390cm] -- cycle;
\node[jt label] at (2.322902,1.003598) {$\tau_3$};
\node[jt label] at (2.322902,0.535446){$l_{12}$};
\draw[jt edge] (0.365500,0.322000) arc[start angle=121.02045133,end angle=58.97954867,radius=3.75380327cm];
\draw[jt edge] (4.234500,0.322000) arc[start angle=-128.64338713,end angle=-173.46876607,radius=5.24199660cm];
\draw[jt edge] (2.300000,3.820000) arc[start angle=-6.53123393,end angle=-51.35661287,radius=5.24199660cm];
\draw[jt edge] (2.353000,1.541000) arc[start angle=118.36751733,end angle=124.67683716,radius=21.18380390cm];
\draw[jt edge] (2.353000,1.541000) arc[start angle=59.06602341,end angle=55.05659551,radius=32.04352037cm];
\draw[jt edge] (2.353000,1.541000) arc[start angle=-177.79310338,end angle=-179.54245692,radius=74.66613073cm];
\fill[black] (2.353000,1.541000) circle[radius=1.65pt];
\node[jt label] at (2.533000,1.611000) {$0$};
\fill[black] (0.365500,0.322000) circle[radius=1.65pt];
\node[jt label] at (0.225500,0.102000) {$1$};
\fill[black] (4.234500,0.322000) circle[radius=1.65pt];
\node[jt label] at (4.374500,0.102000) {$2$};
\fill[black] (2.300000,3.820000) circle[radius=1.65pt];
\node[jt label] at (2.300000,4.050000) {$3$};
\path[fill=jtSage] (6.265500,0.322000) arc[start angle=121.02045133,end angle=58.97954867,radius=3.75380327cm] arc[start angle=-128.64338713,end angle=-173.46876607,radius=5.24199660cm] arc[start angle=-6.53123393,end angle=-51.35661287,radius=5.24199660cm] -- cycle;
\node[jt label] at (8.200000,1.567976) {$\tau_0$};
\draw[jt edge] (6.265500,0.322000) arc[start angle=121.02045133,end angle=58.97954867,radius=3.75380327cm];
\draw[jt edge] (10.134500,0.322000) arc[start angle=-128.64338713,end angle=-173.46876607,radius=5.24199660cm];
\draw[jt edge] (8.200000,3.820000) arc[start angle=-6.53123393,end angle=-51.35661287,radius=5.24199660cm];
\fill[black] (6.265500,0.322000) circle[radius=1.65pt];
\node[jt label] at (6.125500,0.102000) {$1$};
\fill[black] (10.134500,0.322000) circle[radius=1.65pt];
\node[jt label] at (10.274500,0.102000) {$2$};
\node[jt label] at (8.2,0.535446){$l_{12}$};
\node[jt label] at (9.200000,1.835446){$l_{23}$};
\fill[black] (8.200000,3.820000) circle[radius=1.65pt];
\node[jt label] at (8.200000,4.050000) {$3$};
\node[jt label] at (7.200000,1.835446){$l_{13}$};
\draw[
  jt move,
  {Stealth[length=2.2mm,width=1.5mm]}-{Stealth[length=2.2mm,width=1.5mm]}
] (4.800000,1.700000) -- (5.700000,1.700000);
\end{tikzpicture}
\caption{Depicting the hyperbolic case here, the $3\leftrightarrow1$ Pachner move for the Regge discretization shows that the amplitude for the configuration of three triangles incident on a shared vertex $v_{0}$ after fixing the inserted vertex is equal to the amplitude of a single triangle with the same boundary data.}
\label{fig:regge-3-to-1}
\end{figure}
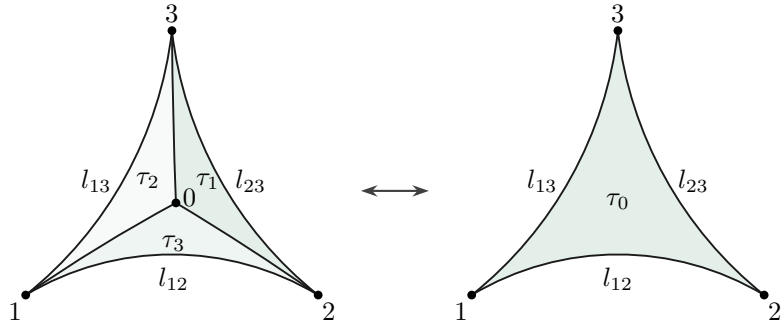
Fix the boundary triangle in $X$ and set $q_1=1$ following Proposition \ref{prop:triangle-closure}. The inserted vertex is represented by $q_0x_0$, where $x_0$ remains the fixed base point of $X$. Then \eqref{eq:regge-pachner-point-measure} gives
\begin{equation}
 d_Lq_0=
 \frac{S_\Lambda(l_{01})\,dl_{01}\,S_\Lambda(l_{02})\,dl_{02}}
 {V_{\tau_3,\Lambda}}.
\end{equation}
For each allowed position, $l_{03}$ is fixed and $\delta_0(l;\Lambda)=0$. On this branch, differentiating the cosine laws at fixed boundary lengths, $l_{01}$, and $l_{02}$ gives
\begin{equation}\label{eq:regge-3-to-1-deficit-jacobian}
 \left|\frac{\partial\delta_0}{\partial l_{03}}\right|_{\delta_0=0}
 =\frac{S_\Lambda(l_{03})V_{\tau_0,\Lambda}}
 {V_{\tau_1,\Lambda}V_{\tau_2,\Lambda}}.
\end{equation}
Indeed, the derivative is
$S_\Lambda(l_{23})\cos\beta_3^{\tau_1}/V_{\tau_1,\Lambda}
+S_\Lambda(l_{13})\cos\beta_3^{\tau_2}/V_{\tau_2,\Lambda}$, thus
using $\beta_3^{\tau_0}=\beta_3^{\tau_1}+\beta_3^{\tau_2}$ and
$V_{\tau,\Lambda}=S_\Lambda(l_{3i})S_\Lambda(l_{3j})\sin\beta_3^\tau$
yields \eqref{eq:regge-3-to-1-deficit-jacobian}. The delta function change of variables therefore gives the measure identity
\begin{equation}\label{eq:regge-3-to-1-measure}
 \frac{d_Lq_0}{V_{\tau_0,\Lambda}}
 =
 \frac{\prod_{i=1}^3S_\Lambda(l_{0i})\,dl_{0i}}
 {V_{\tau_1,\Lambda}V_{\tau_2,\Lambda}V_{\tau_3,\Lambda}}
 \delta_{\mathbb R}\!\left(\delta_0(l;\Lambda)\right),
\end{equation}
where the right-hand side is pushed forward by integrating $l_{03}$. With three non-degenerate incident triangles, their angles at $0$ lie in $(0,\pi)$, so $-\pi<\delta_0<2\pi$. The periodic constraint of Theorem \ref{thm:regge-recovery} consequently has only its zero deficit contribution on this sector.

After removing the common boundary length factors, the three triangle amplitude is
\begin{equation}
 \mathcal A_{\tau_1\tau_2\tau_3}
 =\frac1{V_{\tau_0,\Lambda}}\int d_Lq_0,
\end{equation}
where the integration is over the admissible interior vertex positions. Choose $a,b$ so that $l_{01}=a$, $l_{02}=b$ select one such position on the chosen orientation branch. Its normalized gauge fixing insertion is
\begin{equation}
 \int d_Lq_0\,
 \frac{V_{\tau_3,\Lambda}}
 {S_\Lambda(l_{01})S_\Lambda(l_{02})}
 \delta_{\mathbb R}(l_{01}-a)\,
 \delta_{\mathbb R}(l_{02}-b)=1.
\end{equation}
Thus
\begin{equation}
 \mathcal A_{\tau_1\tau_2\tau_3}^{\mathrm{gf}}
 =\frac1{V_{\tau_0,\Lambda}}
 =\mathcal A_{\tau_0}.
\end{equation}
The surrounding deficit constraints are unchanged because the angles at each boundary vertex add to those of $\tau_0$. This proves the $3\leftrightarrow1$ Pachner move invariance.

\section*{Conclusion}\addcontentsline{toc}{section}{Conclusion}

We have constructed a homogeneous Regge discretization of two-dimensional linear dilaton gravity and derived its path integral geometric measure from a ribbon discretization of BF theory. The ribbon refinement simultaneously keeps the variables associated with  the triangulation and with its dual complex. After a Regge-adapted partial gauge fixing, the short ribbon holonomies carry the rotational curvature data while the long holonomies determine  the intrinsic geometry of the triangles.
A central ingredient in the construction is the Iwasawa decomposition of the BF gauge group, $G_{1,\kappa}\cong SO(2)\bowtie AN_{1,\kappa},$ which separates the rotational and translational sectors in a way adapted to the homogeneous Regge geometry.

In the curved case   the  $AN_{1,\kappa}$ 
factor is non-unimodular, so reduced face distributions depend on a choice of cilium. The underlying ribbon BF theory is nevertheless cut independent: the root dependence of the split description is a dependence on the partial gauge-fixing prescription rather than  a physical ambiguity.

These developments suggest several natural directions for further work. Baratin and Freidel obtained their Regge-type topological models from Feynman amplitudes \cite{baratin_hidden_2007,baratin_hidden_2007-1}, whereas our construction starts from a BF theory and derives the Regge formulation through its discretization and reduction. It would be interesting to close this circle in two dimensions by deriving the present models directly from Feynman diagrams, both for vanishing and non-vanishing cosmological constant.

Table \ref{table:signature-and-curvature-cases} also suggests natural extensions to our works. There is the aforementioned need to formulate reality conditions if one wants to use the split BF to Regge construction in the Euclidean $\Lambda>0$ case. On the other hand, Table \ref{table:signature-and-curvature-cases} also indicates that the split BF to Regge correspondence seems well adapted to all signs of $\Lambda$ in a Lorentzian signature. Therefore it is worthwhile to try and adapt the Lorentzian Regge constructions of \cite{tate_realizability_2012, ito_pinched_2026} to our setting, wherein the correspondence may address all cases. The converse may also be possible if one adopts the viewpoint suggested in \S\ref{subsec:regge-path-int-discretization}, that is, to take the Lorentzian split BF theory as primary and use it to \textit{derive} a proposed appropriate Lorentzian two-dimensional homogeneous Regge theory for general $\Lambda$.

In the quantum theory considered here, the dilaton variables have been fully integrated out. A natural extension would be to retain part of these data and study the corresponding enlarged state-sum description. We have not addressed the representation-theoretic formulation of the model, but it is natural to ask whether the corresponding state sum can be organized in terms of representations of some deformed symmetry structures, reducing in the flat limit to ordinary Poincar\'e representations, in close analogy with the higher-dimensional constructions of Baratin and Freidel \cite{baratin_hidden_2007, baratin_hidden_2007-1}.  In view of the matched pair structure of the underlying BF theory, one may expect $\kappa$-Poincar\'e/Euclidean-type structures to emerge (see \cite{arzano_deformations_2021} for a modern review of deformations of spacetime symmetries). We note in particular that $\kappa$-Poincar\'e deformations are also relevant for a 2-group perspective \cite{Girelli:2021khh}. 

We have also restricted our attention to the simplest dilaton-gravity model. More general two-dimensional dilaton gravities are naturally described by non-linear Poisson sigma models \cite{ikeda_general_1993,ikeda_two-dimensional_1994}, rather than by an ordinary BF theory, and include models such as Liouville-type gravity. In such cases, one expects the underlying algebraic structure to be correspondingly deformed, potentially involving further quantum-group-like structures. It would be interesting to determine to what extent the present discretization can be generalized to these non-linear theories and whether the split triangulation/dual-complex picture survives in an appropriate form.

Finally, a crucial ingredient in the curved construction was the ribbon structure, which allowed us to keep track simultaneously of data living on the triangulation and on the dual complex. An important question is whether an analogous mechanism exists in three and four dimensions. The higher-dimensional generalization may require a more elaborate combinatorial structure, possibly involving higher-gauge-theoretic data rather than an ordinary ribbon description. Understanding this extension would clarify whether the semi-dual discretization uncovered here is a peculiarity of two dimensions or the lower-dimensional manifestation of a more general mechanism.

\section*{Acknowledgments}\addcontentsline{toc}{section}{Acknowledgments}

The authors would like to extend a great appreciation to Giulio Neri, Seth Asante, Aldo Riello, Simon Langenscheidt, and Jos\'e Padua-Arg\"{u}elles for insightful discussions throughout these explorations. Research at the Perimeter Institute is supported in part by the Government of Canada through the Department of Innovation, Science and Economic Development and by the Province of Ontario through the Ministry of Colleges, Universities, Research Excellence and Security. We acknowledge the support of the Natural Sciences and Engineering Research Council of Canada (NSERC) through the Discovery Grant awarded to FG, and the Canadian Graduate Research Scholarship - Doctoral (CGRS D) grant awarded to CP. We further acknowledge that many of the figures throughout this document were generated with  LLM tools, and that LLM's assisted in reviewing the text for typos and overall narrative flow.

\appendix

\section{Regge dilaton equation of motion}\label{app:discrete-dilaton-eom}

Here we show that the equation of motion
\begin{equation}
     \sum_{\tau\supset e}\sum_{v\in \tau}\Phi_{v}\pdv{\beta^{\tau}_{v}(l;\Lambda)}{l_{e}} \approx 0,\quad \forall e\in E(\Delta)
\end{equation}
from varying the Regge discretized LDG action \eqref{eq:discretized-LDG-action} with respect to $l_{e}$ admits a direct geometrical interpretation as the gluing condition for solutions of the continuum dilaton equation.

Consider first a single constant curvature triangle $\tau=(ijk)$ and set
\begin{equation}
    a=l_{jk},\quad
    b=l_{ik},\quad
    c=l_{ij},
\end{equation}
with corresponding internal angles
\begin{equation}
    \alpha=\beta_i^\tau,\quad
    \beta=\beta_j^\tau,\quad
    \gamma=\beta_k^\tau.
\end{equation}
The derivative of the angles with respect to the edge $a$ are
\begin{align}
    \frac{\partial\alpha}{\partial a}
    &=
    \frac{S_\Lambda(a)}{V_{\tau,\Lambda}},
    \\
    \frac{\partial\beta}{\partial a}
    &=
    -\frac{S_\Lambda(b)\cos\gamma}
    {V_{\tau,\Lambda}},
    \\
    \frac{\partial\gamma}{\partial a}
    &=
    -\frac{S_\Lambda(c)\cos\beta}
    {V_{\tau,\Lambda}}.
\end{align}
It follows that the contribution of $\tau$ to the equation associated with the edge $(jk)$ is
\begin{equation}
    \mathcal E_{jk}^{\tau}
    :=
    \sum_{v\in\tau}
    \Phi_v
    \frac{\partial\beta_v^\tau}
    {\partial l_{jk}}
    =
    \frac{1}{V_{\tau,\Lambda}}
    (
    S_\Lambda(a)\Phi_i
    -S_\Lambda(b)\cos\gamma\,\Phi_j
    -S_\Lambda(c)\cos\beta\,\Phi_k
    ),
    \label{eq:triangle-length-eom}
\end{equation}
while the contribution from a given edge $(jk)$ shared by two triangles is 
\begin{equation}
    \mathcal E_{jk}^{\tau}
    +
    \mathcal E_{jk}^{\tau'}
    =0,
\end{equation}
To parameterize the solution inside a single non-degenerate triangle $\tau$, we introduce the constant vector $\Phi_\tau\in \mathbb{R}^3$ in the ambient space such that 
\begin{equation}
   {
    \Phi_v
    :=
    \Phi_\tau\cdot \mathcal X_v.
    }
    \label{eq:discrete-JT-dilaton}
\end{equation}
Each  triangle carries its own ${\Phi}_\tau$, and the Regge equations will determine how these local solutions glue across common edges.

The embedding coordinates of $X_\Lambda$ satisfy
\begin{equation}
    \nabla_\mu\nabla_\nu \mathcal X
    =
    -\Lambda g_{\mu\nu}\mathcal X.
    \label{eq:embedding-hessian}
\end{equation}
Indeed, $\mathcal X$ is normal to the constant curvature quadric, while $\partial_\mu\mathcal X$ is tangent, and the Gauss formula gives
\begin{equation}
    \partial_\mu\partial_\nu\mathcal X
    =
    \Gamma^\rho_{\mu\nu}\partial_\rho\mathcal X
    -
    \Lambda g_{\mu\nu}\mathcal X.
\end{equation}
Since $\Phi_\tau$ is a constant ambient vector, the field
\begin{equation}
    \Phi_\tau(\mathcal X)
    =
    \Phi_\tau\cdot \mathcal X
\end{equation}
 satisfies
\begin{align}
    \nabla_\mu\nabla_\nu\Phi(\mathcal X)
    &=
    \Phi_\tau\cdot
    \nabla_\mu\nabla_\nu\mathcal X
    \nonumber\\
    &=
    -\Lambda g_{\mu\nu}\Phi_\tau(\mathcal X).
\end{align}
so that the interpolation determined by the vertex data is an exact continuum dilaton solution of the dilaton equation in the interior of each constant curvature triangle.

Now let us focus on the geodesic edges characterized in \eqref{eq:ambient-geodesic}. The dilaton along this edge is 
\begin{align}
    \Phi_{ij}(s)
    &=
    \Phi_\tau\cdot \mathcal{X}_{ij}(s)
    \nonumber\\
    &=
    C_\Lambda(s)\Phi_\tau(\mathcal X_i)
    +
    S_\Lambda(s)\,
    \Phi_\tau\cdot t_{ij}.
    \label{eq:dilaton-geodesic-interpolation}
\end{align}
In particular, at $s=l_{ij}$,
\begin{equation}
    \Phi_\tau(\mathcal X_j)
    =
    C_\Lambda(l_{ij})\Phi_\tau(\mathcal X_i)
    +
    S_\Lambda(l_{ij})\,
    \Phi_\tau\cdot t_{ij}.
    \label{eq:dilaton-edge-relation}
\end{equation}
The ambient formulation therefore provides a natural interpolation of the discrete dilaton between neighboring vertices. It is also useful to eliminate the tangent component in favour of the values of the dilaton at the endpoints. From \eqref{eq:dilaton-edge-relation}, one finds
\begin{equation}
    \Phi_\tau\cdot t_{ij}
    =
    \frac{
        \Phi_\tau(\mathcal X_j)-C_\Lambda(l_{ij})\Phi_\tau(\mathcal X_i)
    }{
        S_\Lambda(l_{ij})
    },\label{eq:phi-t}
\end{equation}
and consequently
\begin{equation}
    {
    \Phi_{ij}(s)
    =
    C_\Lambda(s)\Phi_\tau(\mathcal X_i)
    +
    \frac{S_\Lambda(s)}{S_\Lambda(l_{ij})}
    (
        \Phi_\tau(\mathcal X_j)-C_\Lambda(l_{ij})\Phi_\tau(\mathcal X_i)
    ).
    }
    \label{eq:curved-dilaton-interpolation}
\end{equation}
This expression should be compared with the usual linear interpolation of a scalar field on a flat edge.  In the flat limit,
\begin{equation}
    C_\Lambda(s)\rightarrow1,
    \quad
    S_\Lambda(s)\rightarrow s,
\end{equation}
and \eqref{eq:curved-dilaton-interpolation} becomes
\begin{equation}
    \Phi_{ij}(s)
    =
    \Phi_\tau(x_i)
    +
    \frac{s}{l_{ij}}
    (\Phi_\tau(x_j)-\Phi_\tau(x_i)),
\end{equation}
which is precisely the standard affine interpolation of the dilaton along a flat edge.

We note that $ \Phi_\tau\cdot t_{ij}$ can naturally be interpreted as the directional derivative of $\Phi$ in the direction of $t_{ij}$. Indeed, its directional derivative along $t_{ij}\in T_{\mathcal X_i}X_\Lambda$ is simply
\begin{equation}\label{eq:direcder}
    \nabla_t\Phi_\tau |_{\mathcal X_i}
    =
    \frac{d}{ds}\Phi_\tau(\mathcal X(s))|_{s=0}
    =
    \Phi_\tau\cdot\dot {\mathcal X}(0)
     =
    \Phi_\tau\cdot t_{ij}= \frac{\Phi_\tau(\mathcal X_j)-C_\Lambda(a)\Phi_\tau(\mathcal X _i)}
         {S_\Lambda(a)},
\end{equation}
where $\mathcal X(0)=\mathcal X_i$ and $\dot {\mathcal X}(0)=t_{ij}$ and in the last equality, we used \eqref{eq:phi-t}.
Along the geodesic $(ij)$, parametrized by its arc length $s$, as we discussed earlier, the dilaton takes the form
\begin{equation}
    \Phi_{ij}(s)
    =
    C_\Lambda(s)\Phi_\tau(\mathcal X_i)
    +
    S_\Lambda(s)\,
    \Phi_\tau \cdot t_{ij},
\end{equation}
where $t_{ij}$ denotes the initial unit tangent at $X_i$. Since
\begin{equation}
    C_\Lambda''(s)=-\Lambda C_\Lambda(s),
    \quad
    S_\Lambda''(s)=-\Lambda S_\Lambda(s),
\end{equation}
it follows directly that
\begin{equation}
    \frac{d^2\Phi_{ij}}{ds^2}
    +
    \Lambda\Phi_{ij}(s)
    =0,
\end{equation}
which is  the one-dimensional restriction of the continuum dilaton equation along the geodesic edge. 

The only remaining gluing condition is the continuity of
the normal derivative.  At the vertex $j$, let $t$ denote the unit tangent vector to the geodesic $(jk)$, directed from $j$ to $k$, and let $n$ be the unit normal to this geodesic pointing towards the interior of $\tau$. We can use $(t,n)$ as a basis in $T_{\mathcal X}X_{\Lambda}$. The embedding coordinates satisfy
\begin{align}
    \mathcal X_k
    &=
    C_\Lambda(a)\mathcal X_j
    +
    S_\Lambda(a)t,
    \label{eq:Xk-expansion}
    \\
    \mathcal X_i
    &=
    C_\Lambda(c)\mathcal X_j
    +
    S_\Lambda(c)
    (
    \cos\beta\,t+\sin\beta\,n
    ).
    \label{eq:Xi-expansion}
\end{align}
Taking the scalar product with $\Phi_\tau$, the second equation leads to 
\begin{equation}
    \Phi_\tau (\mathcal X_i)
    =
    C_\Lambda(c)\Phi_\tau (\mathcal X_j)
    +
    S_\Lambda(c)
    (
    \cos\beta\,\nabla_t\Phi_\tau
    +
    \sin\beta\,\nabla_n\Phi_\tau
    ).
\end{equation}
Eliminating the tangential derivative using
 \eqref{eq:direcder} yields
\begin{align}
    \nabla_n\Phi_\tau (\mathcal X_j)
    =
    \frac{1}{
    S_\Lambda(a)S_\Lambda(c)\sin\beta}
    \Big(
        &S_\Lambda(a)\Phi_\tau (\mathcal X_i)
        -S_\Lambda(a)C_\Lambda(c)\Phi_\tau (\mathcal X_j)
        \nonumber\\
        &-S_\Lambda(c)\cos\beta\,\Phi_\tau(\mathcal X_k)
        +S_\Lambda(c)C_\Lambda(a)
        \cos\beta\,\Phi_\tau (\mathcal X_j)
    \Big).
    \label{eq:normal-derivative-intermediate}
\end{align}
Using
\begin{equation}
    V_{\tau,\Lambda}
    =
    S_\Lambda(a)S_\Lambda(c)\sin\beta
\end{equation}
and the curved trigonometric identity
\begin{equation}
    S_\Lambda(b)\cos\gamma
    =
    S_\Lambda(a)C_\Lambda(c)
    -
    C_\Lambda(a)S_\Lambda(c)\cos\beta,
\end{equation}
we finally obtain
\begin{equation}
    {
    \nabla_n\Phi_\tau (\mathcal X_j)
    =
    \frac{
    S_\Lambda(a)\Phi_\tau (\mathcal X_i)
    -S_\Lambda(b)\cos\gamma\,\Phi_\tau (\mathcal X_j)
    -S_\Lambda(c)\cos\beta\,\Phi_\tau (\mathcal X_k)
    }{
    V_{\tau,\Lambda}
    }.
    }
    \label{eq:normal-derivative}
\end{equation}
Comparing with \eqref{eq:triangle-length-eom}, we conclude that
\begin{equation}
    {
    \mathcal E_{jk}^{\tau}
    =
    \nabla_{n_{\tau,jk}}\Phi_\tau (\mathcal X_j) .
    }
    \label{eq:triangle-normal-derivative}
\end{equation}
This provides a simple interpretation of the length equations of motion. Consider an interior edge $e=(jk)$ shared by two triangles
\begin{equation}
    \tau=(ijk),
    \quad
    \tau'=(i'jk),
\end{equation}
we must have 
\begin{equation}
    \mathcal E_{jk}^{\tau}
    +
    \mathcal E_{jk}^{\tau'}
    =0,
\end{equation}
and therefore
\begin{equation}
    {
    \nabla_{n_{\tau,e}}\Phi_\tau (\mathcal X_j)
    +
    \nabla_{n_{\tau',e}}\Phi_{\tau'} (\mathcal X_{j})
    =0.
    }
    \label{eq:normal-gluing}
\end{equation}
The two inward-pointing normals have opposite orientations along the common edge,
\begin{equation}
    n_{\tau',e}=-n_{\tau,e}.
\end{equation}
Hence \eqref{eq:normal-gluing} is  the condition that the normal derivatives of the two local dilaton solutions agree at the common vertex. This is enough to obtain agreement along the full edge. Indeed, if $t$ is the unit tangent and $n$ is parallel transported along the geodesic, the continuum dilaton equation gives
\begin{equation}
   \nabla_t\nabla_n\Phi_\tau
    =-\Lambda g(t,n)\Phi_\tau=0.\end{equation}
Thus the normal derivative is constant along the edge, and matching it at one endpoint matches it everywhere on the edge. The Regge length equation therefore implements the $C^1$ gluing of the trianglewise continuum dilaton solutions on the zero-deficit development.

\section{Identities for non-unimodular integration and matched pairs} \label{app:measure-considerations}

Here we prove some useful identities relevant for integration on non-unimodular Lie groups. Recall the decomposition \eqref{eq:path-int-measure-v1} of the unimodular Haar measure on $G_{\sigma,\kappa}$ induced from its Iwasawa decomposition,
\begin{equation}
    dg = dhd_{R}\ell = d_{L}\tilde{\ell} d\th,\quad g = h\ell = \tilde{\ell}\th\in G_{\sigma,\kappa}, \quad h,\th\in H_{\sigma},\quad \ell,\tell\in AN_{\sigma,\kappa}.
\end{equation}
Given that we will make use of both $d_{R}$ and $d_{L}$ on $AN_{\sigma,\kappa}$, we fix now our notation for the relevant modular factors\footnote{Modular factors are homomorphisms. That is, e.g., for $\Delta_{R}$ we have $\Delta_{R}(\ell\ell_{0}) = \Delta_{R}(\ell)\Delta_{R}(\ell_{0})$, $\Delta_{R}(\ell)^{-1} = \Delta_{R}(\ell^{-1})$, and $\Delta_{R}(1) =1 $. Analogously for $\Delta_{L}$.} and Jacobians,
\begin{subequations}
\begin{align}
    d_{R}(\ell_{0}\ell)=:\Delta_{R}(\ell_{0})d_{R}\ell,\quad d_{L}(\tell\ell_{0}) =: \Delta_{L}(\ell_{0})d_{L}\tell,\label{eq:right-left-mod-factor}\\
     d_{R}(h_{0}\rhd \ell) =: J^{R}_{h_{0}}(\ell)d_{R}\ell,\quad d_{L}(\tell\lhd h_{0})=:J^{L}_{h_{0}}(\tell)d_{L}\tell,\label{eq:right-left-Jacobians}
\end{align}
\end{subequations}
as well as the Jacobian factors for the Haar measure on $H_{\sigma}$,
\begin{equation}
d(\ell_{0}\rhd h)  = d(h^{-1}\lhd \ell_{0}^{-1})=: K_{\ell_{0}}(h)dh.
\end{equation}
The left and right invariant Haar measures, as well as the left and right modular factors, by definition obey the relations,
\begin{equation}{\label{eq:left-right-measure-mod-factor-def}}
   d_{L}\ell = d_{R}(\ell^{-1}) = \Delta_{R}(\ell)^{-1}d_{R}\ell,\quad \Delta_{R}(\ell_{0}) = \Delta_{L}(\ell_{0})^{-1}.
\end{equation}
In this context one has the following relations between $d_{L}, d_{R}, \Delta_{R}, \Delta_{L}, J^{R}_{\bullet}(\bullet), J^{L}_{\bullet}(\bullet),$ and $K_{\bullet}(\bullet)$.

\begin{lemma}[Left/Right $AN_{\sigma,\kappa}$ Haar measure dressing invariance]\label{lem:right-left-measure-dressing-invariance}
    The left invariant Haar measure $d_{L}$ on $AN_{\sigma,\kappa}$ and the right invariant Haar measure $d_{R}$ on $AN_{\sigma,\kappa}$ are respectively invariant under the left dressing action by $H_{\sigma}$ and the right dressing action by $H_{\sigma}$,
    \begin{align}
        d_{L}(h_{0}\rhd \tilde{\ell}) = d_{L}\tilde{\ell},\quad d_{R}(\ell\lhd h_{0}) = d_{R}\ell,\quad g = h\ell = \tilde{\ell}\tilde{h}\in G_{\sigma,\kappa},\quad h_{0}\in H_{\sigma} .
    \end{align}
\end{lemma}
\begin{proof}
    First consider the left multiplication by $h_{0}$ on $g = \tilde{\ell}\tilde{h}$,
    \begin{equation}
        g\mapsto h_{0}g = h_{0}\tell\th = (h_{0}\rhd \tell)(h_{0}\lhd \tell)\th\implies \tell' = h_{0}\rhd \tell, \quad \th' = (h_{0}\lhd \tell)\th.
    \end{equation}
    At the level of the Haar measure on $G_{\sigma,\kappa}$, we then have,
    \begin{equation}
        dg\mapsto d(h_{0}g) = d_{L}(h_{0}\rhd\tell)d((h_{0}\lhd \tell)\th) = d_{L}(h_{0}\rhd\tell)d\th,
    \end{equation}
    which, by unimodularity of $dg$ implies,
    \begin{equation}
      d_{L}(h_{0}\rhd\tell)d\th=   d(h_{0}g) = dg = d_{L}\tell d\th \implies  d_{L}(h_{0}\rhd\tell)= d_{L}\tell,
    \end{equation}
    as claimed. Similarly, consider right multiplication by $h_{0}$ on $g = h\ell$,
    \begin{equation}
        g\mapsto g h_{0} = h\ell h_{0} = h(\ell\rhd h_{0})(\ell\lhd h_{0})\implies h' = h(\ell\rhd h_{0}),\quad \ell' = \ell\lhd h_{0},
    \end{equation}
    so at the level of the Haar measure on $G_{\sigma,\kappa}$, we have,
    \begin{equation}
        dg\mapsto d(gh_{0}) = d(h(\ell\rhd h_{0}))d_{R}(\ell\lhd h_{0}) = dh d_{R}(\ell\lhd h_{0}),
    \end{equation}
    which by unimodularity of $dg$ implies,
    \begin{equation}
         dh d_{R}(\ell\lhd h_{0}) = d(gh_{0}) = dg = dh d_{R}\ell\implies d_{R}(\ell\lhd h_{0}) = d_{R}\ell,
    \end{equation}
    as claimed.
\end{proof}
\begin{lemma}[Left/right Jacobians and the right modular factor relation]
The right Jacobian $J^{R}_{\bullet}(\bullet)$, left Jacobian $J^{L}_{\bullet}(\bullet)$, and right modular factor $\Delta_{R}$ respectively defined by \eqref{eq:right-left-Jacobians} and \eqref{eq:right-left-mod-factor} obey the relations,
\begin{equation}
   J^{R}_{h_{0}}(\ell) = \frac{\Delta_{R}(h_{0}\rhd \ell)}{\Delta_{R}(\ell)},\quad J^{L}_{h_{0}}(\tell) = \frac{\Delta_{R}(\tell)}{\Delta_{R}(\tell\lhd h_{0})}.
\end{equation}
\end{lemma}
\begin{proof}
    We begin with the proof of the identity for $J^{R}_{h_{0}}(\ell)$. Using \eqref{eq:left-right-measure-mod-factor-def} and \eqref{eq:right-left-Jacobians} we have,
    \begin{equation}
        d_{L}(h_{0}\rhd \ell) = \Delta_{R}(h_{0}\rhd \ell)^{-1}d_{R}(h_{0}\rhd\ell) = \Delta_{R}(h_{0}\rhd\ell)^{-1}J^{R}_{h_{0}}(\ell)d_{R}\ell.
    \end{equation}
    But, by Lemma \ref{lem:right-left-measure-dressing-invariance}, the left hand side is equal to $d_{L}\ell$, thus,
    \begin{equation}
     \Delta_{R}(\ell)^{-1}d_{R}\ell= d_{L}\ell = d_{L}(h_{0}\rhd\ell) = \Delta_{R}(h_{0}\rhd\ell)^{-1}J^{R}_{h_{0}}(\ell)d_{R}\ell \implies J^{R}_{h_{0}}(\ell) = \frac{\Delta_{R}(h_{0}\rhd \ell)}{\Delta_{R}(\ell)},
    \end{equation}
    as claimed. Very similarly for the $J^{L}_{h_{0}}(\tell)$ identity, we use \eqref{eq:left-right-measure-mod-factor-def}, \eqref{eq:right-left-Jacobians}, and Lemma \ref{lem:right-left-measure-dressing-invariance} to find
    \begin{equation}
       J^{L}_{h_{0}}(\tell)d_{L}\tell =  d_{L}(\tell\lhd h_{0}) = \Delta_{R}(\tell\lhd h_{0})^{-1}d_{R}(\tell\lhd h_{0}) = \Delta_{R}(\tell\lhd h_{0})^{-1}d_{R}\tell = \Delta_{R}(\tell\lhd h_{0})^{-1}\Delta_{R}(\tell)d_{L}\tell,
    \end{equation}
    which indeed implies,
    \begin{equation}
        J^{L}_{h_{0}}(\tell) = \frac{\Delta_{R}(\tell)}{\Delta_{R}(\tell\lhd h_{0})},
    \end{equation}
    as claimed.
\end{proof}
\begin{lemma}[$H_{\sigma}$ Jacobian and right modular factor relation]\label{lem:K-jacobian-right-mod-factor}
    The Jacobian factor $K_{\ell_{0}}(h)$ for the $H_{\sigma}$ Haar measure is related to the right modular factor $\Delta_{R}$ by,
    \begin{equation}
        K_{\ell_{0}}(h) = \Delta_{R}(\ell_{0}\lhd h)^{-1},\quad \ell_{0}\in AN_{\sigma,\kappa}.
    \end{equation}
    When $h\ell = \tell\th$, we have the special cases,
    \begin{equation}
        K_{\ell^{-1}}(h^{-1}) = \Delta_{R}(\tell),\quad K_{\tell}(\th) = \Delta_{R}(\ell)^{-1} .
    \end{equation}
\end{lemma}
\begin{proof}
    Consider the left multiplication by $\ell_{0}$ on $g = h\ell$,
    \begin{equation}
        g\mapsto \ell_{0}g = \ell_{0}h\ell = (\ell_{0}\rhd h)(\ell_{0}\lhd h)\ell\implies h' = \ell_{0}\rhd h, \quad \ell' = (\ell_{0}\lhd h)\ell.
    \end{equation}
    At the level of the Haar measure on $G_{\sigma,\kappa}$ we then have
    \begin{equation}
        dg\mapsto d(\ell_{0}g) = d(\ell_{0}\rhd h)d_{R}((\ell_{0}\lhd h)\ell) = K_{\ell_{0}}(h)\Delta_{R}(\ell_{0}\lhd h)dh d_{R}\ell,
    \end{equation}
    which, by unimodularity of $dg$ implies,
    \begin{equation}
        K_{\ell_{0}}(h)\Delta_{R}(\ell_{0}\lhd h)dh d_{R}\ell = d(\ell_{0}g) = dg = dh d_{R}\ell\implies K_{\ell_{0}}(h) = \Delta_{R}(\ell_{0}\lhd h)^{-1},
    \end{equation}
    as claimed. The special case claim follows immediately from the above proof and the ribbon relation $h\ell = \tell \th$.
\end{proof}

\printbibliography

\end{document}